\documentclass[%
reprint,
superscriptaddress,
nofootinbib,
longbibliography,
noeprint,
amsmath,amssymb,
aps
]{revtex4-2}
\usepackage[utf8]{inputenc}
\usepackage{bbm}
\usepackage[left=2cm,right=2cm,top=2cm,bottom=2cm]{geometry}
\usepackage{enumitem}
\usepackage{adjustbox}
\usepackage{amsmath,amsfonts,amssymb,amsthm,graphics,graphicx}
\usepackage{hyperref}

\newtheorem{theorem}{Theorem}
\theoremstyle{definition}
\usepackage [autostyle, english = american]{csquotes}
\usepackage{subcaption,float}
\graphicspath{{figures/}}
\MakeOuterQuote{"}
\newcommand{\ef}{\tilde{\mathcal{E}}}
\newcommand{\ket}[1]{|{#1}\rangle}
\newcommand{\bra}[1]{\langle {#1} |}
\renewcommand{\mod}{\mathrm{mod \,}}

\def\equationautorefname~#1\null{Eq. (#1)\null}

\newcommand{\appref}[1]{\hyperref[#1]{App.~\ref*{#1}}}
\renewcommand{\vec}[1]{\mathbf{#1}}
\newcommand{\david}[1]{{#1}}
\newcommand{\gs}[1]{{#1}}
\newcommand{\comment}[1]{}
\newcommand{\cnot}{\mathrm{CNOT}}
\newcommand{\iswap}{\mathrm{iSWAP}}
\newcommand{\dam}{$\mathcal{M}$-dynamics}
\newcommand{\Tr}{\mathrm{Tr}}
\usepackage{tikzit}
\usepackage{subfiles}
\makeatletter 
    
\renewcommand\onecolumngrid{
\do@columngrid{one}{\@ne}%
\def\set@footnotewidth{\onecolumngrid}
\def\footnoterule{\kern-6pt\hrule width 1.5in\kern6pt}%
}

\renewcommand\twocolumngrid{
        \def\footnoterule{
        \dimen@\skip\footins\divide\dimen@\thr@@
        \kern-\dimen@\hrule width.5in\kern\dimen@}
        \do@columngrid{mlt}{\tw@}
}%

\makeatother    
\begin{document}

\tikzstyle{x spider}=[fill={rgb,255: red,232; green,165; blue,165}, draw=black, shape=circle]
\tikzstyle{z spider}=[fill={rgb,255: red,216; green,248; blue,216}, draw=black, shape=circle]
\tikzstyle{hadamard}=[fill={rgb,255: red,255; green,255; blue,153}, draw=black, shape=rectangle]
\tikzstyle{U}=[fill=white, draw=black, shape=rectangle, minimum height=0.3cm, minimum width=0.36cm]
\tikzstyle{end}=[fill=white, draw=black, shape=circle, line width=1pt]
\tikzstyle{Urep}=[fill={rgb,255: red,128; green,0; blue,128}, draw=black, shape=rectangle, minimum height=0.8cm, minimum width=0.8cm, opacity=0.7]
\tikzstyle{U2}=[fill={rgb,255: red,255; green,128; blue,0}, draw=black, shape=rectangle, minimum height=0.8cm, minimum width=0.8cm, opacity=1]
\tikzstyle{square}=[fill=white, draw=black, shape=rectangle, line width=1pt]

\tikzstyle{light}=[-, draw={rgb,255: red,191; green,191; blue,191}]
\tikzstyle{arrow}=[line width=1pt, <->]
\tikzstyle{thick}=[-, line width=2pt, draw={rgb,255: red,128; green,0; blue,128}]
\tikzstyle{bigZ}=[-, fill={rgb,255: red,216; green,248; blue,216}]
\tikzstyle{Urep shape}=[-, fill={rgb,255: red,128; green,0; blue,128}, opacity=0.6]
\tikzstyle{U2 shape}=[-, fill={rgb,255: red,255; green,128; blue,0}]
\tikzstyle{dashline}=[-, dashed]

\title{Zero-temperature entanglement membranes in quantum circuits} 
\author{Grace M. Sommers}
\address{Physics Department, Princeton University, Princeton, NJ 08544}
\author{Sarang Gopalakrishnan}
\address{Department of Electrical and Computer Engineering, Princeton University, Princeton, NJ 08544}
\author{Michael J. Gullans}
\address{Joint Center for Quantum Information and Computer Science,
NIST/University of Maryland, College Park, Maryland 20742, USA}
\author{David A. Huse}
\address{Physics Department, Princeton University, Princeton, NJ 08544}
\begin{abstract}

In chaotic quantum systems, the entanglement of a region $A$ can be described in terms of the surface tension of a spacetime membrane pinned to the boundary of $A$. Here, we interpret the tension of this entanglement membrane in terms of the rate at which information ``flows'' across it.  For any orientation of the membrane, one can define (generically nonunitary) dynamics {\it across} the membrane; we explore this dynamics in various space-time translation-invariant (STTI) stabilizer circuits in one and two spatial dimensions.  We find that the flux of information across the membrane in these STTI circuits reaches a steady state.  In the {cases where this dynamics is nonunitary and the steady state flux is nonzero, this occurs because the dynamics across the membrane is unitary in a {\it subspace}} of extensive entropy. This generalized unitarity is present in a broad class of STTI stabilizer circuits, and is also present in some special non-stabilizer models. The existence of multiple unitary (or generalized unitary) directions forces the entanglement membrane tension to be a piecewise linear function of the orientation of the membrane; in this respect, the entanglement membrane behaves like an interface in a zero-temperature classical lattice model.  We argue that entanglement membranes in {\it random} stabilizer circuits that produce volume-law entanglement are also effectively at zero temperature.

\end{abstract}

\maketitle
\section{Introduction}

In many strongly interacting quantum systems, the entanglement entropy of a spatial region has a geometrical interpretation in terms of a surface in a higher-dimensional space(-time). This interpretation dates back to the Ryu-Takayanagi conjecture in AdS/CFT~\cite{Ryu2006}, and has recently been explored in depth for random quantum circuits~\cite{Nahum2017op,Nahum2017entanglement,VonKeyserlingk2018,Jonay2018,Nahum2018,Zhou2019,Zhou2020,Nahum2022,Sierant2023,Li2021,Li2023,Sang2023,Liu2023noisy,Lovas2023}. In random circuits, the interpretation of the entanglement entropy as the surface tension of a spacetime membrane can be made explicit, using a mapping between random-circuit ensembles and statistical mechanics models. According to this mapping, the spacetime bulk of a random unitary circuit maps onto a ferromagnet, and the entanglement entropy of a region $A$ is the free energy cost of a domain wall in the ferromagnet separating $A$ from its complement---or equivalently, the {free energy} 
of a spacetime membrane whose boundary (at the final time) coincides with the boundary of $A$. While membrane theory has been developed for the typical behavior of ensembles of random circuits, many of its basic predictions are generic~\cite{Zhou2020}. However, some of its predictions---such as the fact that the large time and/or large distance scaling behavior of the entanglement membrane is in many cases that of a directed polymer in a random medium (DPRM)~\cite{Nahum2017entanglement}---rely on randomness. 

This paper is concerned with the behavior of the entanglement membrane in individual circuits, rather than ensembles, so it is desirable to have an interpretation of the entanglement membrane that does not implicitly rely on an ensemble. {A formal approach to entanglement membrane theory in chaotic, nonrandom systems was introduced in Ref.~\cite{Zhou2020}, which expresses the line tension as an expansion in effective spins that typically converges, but is not guaranteed to do so. Here we {explore a different approach,}  
interpreting the tension of an entanglement membrane as a flux of information across it.} This interpretation is most transparent if one thinks of a setting in which a reference qubit is initially maximally entangled with each qubit in the system, and then the system evolves under a unitary circuit. The entanglement between the reference qubits and the system is the amount of quantum information flowing through the circuit, which, by unitarity, is equal to the number of qubits in the system.\footnote{If we time-reverse the dynamics, it is still unitary and carries the same flux of information.  Thus the direction of this "flow" is set by our choice of the time direction.} 

In the membrane picture, this entanglement is the tension of a membrane running sideways across the circuit, separating the final state of the system from the reference qubits. Unitarity constrains the properties of the entanglement membrane in this setup: since quantum information flows without loss in unitary circuits, the entanglement membrane tension for spacelike membranes is independent of the depth of the circuit. This observation, interpreted in terms of a hypothetical statistical mechanics model, would translate to the statement that spacelike entanglement membranes in unitary circuits behave as if they were at zero temperature:  Increasing the depth of the circuit gives the membrane more places to sit and thus more entropy, yet its tension (a free energy) does not change, indicating that it is at zero temperature. 

In many other contexts---e.g., the growth of bipartite entanglement in shallow wide circuits---the membrane instead runs in a timelike direction.  In this case, when we instead consider the dynamics across the membrane we are choosing {to run time along a different direction through the circuit, so that this membrane is spacelike for this new time direction.}  
{When $v_{LC}=\infty$ in both directions then spacelike is only ``slope'' $v=\infty$.} In general, this dynamics across the membrane is not unitary, but can be expressed in terms of unitary gates interleaved with generalized measurements.  Such hybrid dynamics does not  in general conserve the rank or spectrum of a mixed state, and correspondingly the entanglement membrane need not be at zero temperature.  However, there are important classes of circuits in which the dynamics remains unitary for multiple distinct choices of the time directions~\cite{Gopalakrishnan2019,Bertini2019,Piroli2019,Bertini2020,Bertini2021,Jonay2021,Mestyan2022,Milbradt2023,Sommers2023}:  In 1+1d, these include dual-unitary (DU) circuits, which are unitary in the spatial direction~\cite{Gopalakrishnan2019,Bertini2019,Piroli2019,Bertini2020,Bertini2021}; tri-unitary circuits, with three distinct unitary arrows of time (six if we count time-reversals)~\cite{Jonay2021}; and the recently introduced hierarchy of generalized dual-unitary (gDU) circuits, which satisfy relaxed constraints and in some cases are unitary in only a subspace of the full space of many-body states~\cite{Yu2023,Foligno2023,Liu2023,Rampp2023}. In these cases, the dynamics across the membrane is constrained by (generalized) unitarity and the entanglement membrane is correspondingly at zero temperature.

The bulk of this paper is concerned with the behavior of the entanglement membrane in space-time translation invariant (STTI) stabilizer circuits, in both $1+1$ and $2+1$ dimensions, including nonunitary stabilizer circuits with (generalized) measurements.  We observe that STTI stabilizer circuits that generate volume law entanglement are multi-unitary in the generalized sense {for all of the many cases we have examined:} for any orientation of the entanglement membrane there is a subspace (which we call the ``plateau'') within which the dynamics is exactly unitary. The entanglement membrane tension exhibits zero-temperature behavior constrained by the unitary arrow(s) of time for which the membrane is spacelike.  As a result, the tension is piecewise linear, with cusps when the constraint switches between distinct unitary arrows of time. We demonstrate this behavior for various STTI circuits. In the case of two spatial dimensions, we investigate two examples that had not previously been identified as multi-unitary.

Our methods for studying the membrane tension $\mathcal{E}$ are largely complementary to those in recent works on gDU circuits, and constitute a second major thrust of this work. We present several equivalent ways of defining $\mathcal{E}$ and demonstrate the associated numerical methods in Clifford circuits, e.g. pinning the ends/edges of the membrane to probe many angles at once. These techniques suggest new types of questions about entanglement dynamics in spacetime. While some of the interesting cases -- e.g. the dynamics across membranes running at angles between the butterfly and light cones -- arise in non-stabilizer circuits where our methods are not efficient, Clifford dynamics are nonetheless important in their own right, such as in applications to fault-tolerant quantum computing~\cite{Campbell2017,Gottesman2022,Gottesman2024}.  Moreover, the zero-temperature behavior of membranes in STTI Clifford and in multi-unitary circuits is conjectured to extend to \textit{random} stabilizer circuits, owing to an "emergent plateau" defined below.

\gs{The remainder of the paper proceeds as follows. In~\autoref{sect:membranes} we define entanglement membranes and outline the different classes of dynamics across the membrane.~\autoref{sect:tension-2} presents the general theory for the line tension in "fully multi-unitary" models, several examples of which are presented in~\autoref{sect:examples}. We conclude in~\autoref{sect:conclude}.}
\section{Entanglement membranes}\label{sect:membranes}

\subsection{Membranes in unitary circuits}

Entanglement membranes were first introduced in the context of entanglement growth, starting from a product state, in random unitary circuits. In this context, one is interested in the entanglement of a region $A$---for example, the left half of an infinite system---after a circuit of depth $t$. The entanglement membrane originates at the boundary of $A$, and runs through the circuit so as to disconnect $A$ from its complement. In this case, the entanglement membrane runs in a timelike direction, i.e., within the light cone. In this theory one posits that the entanglement entropy of $A$ is the free energy of the membrane.\footnote{The entropy could be the $n$th Renyi entropy or the von Neumann entropy. In general, the line tension will depend on the Renyi index, and in the presence of conservation laws, the ballistic growth of the von Neumann entropy differs qualitatively from the diffusive growth of higher Renyi entropies~\cite{Rakovszky2019,Zhou2020diffusive,Znidaric2020,Huang2020}. However, we will just refer to "the" line tension. In the concrete examples we consider (Clifford circuits), the spectrum is flat and all Renyi entropies are equal. {Generally, we give entropies in bits.}}

For clarity, we will first discuss the case of 1+1-dimensional circuits, where the membrane is a line, generalizing to higher $d>1$ later.  \david{Our initial discussion here is for a general circuit which may be random, although it must be statistically homogeneous in space and time.}  If the membrane runs at an angle $v$, this free energy takes the form $S = \int \mathcal{E}(v) dt$, where the integral is over the membrane and the entanglement entropy per unit time $\mathcal{E}(v)$ is called the line tension.  In this setup the angle of the membrane is not fixed, since it can terminate anywhere at the bottom of the circuit and, in the limit of a large circuit, will adopt the configuration(s) that minimizes $S$.

To fix the angle of the membrane, one can pin its bottom end using the method introduced in~\cite{Jonay2018} (Fig.~\ref{fig:page-sketch}):  Initialize a long chain of $L$ qubits in a pure state that is a product of fully scrambled states on each half of the chain at time $t=0$, with the center of the chain at $x=0$, and distances defined so the qubits have unit spacing along the chain.  The initial entanglement is then $S(x,t=0)=|x|$ for $|x|<L/4$; this is the entanglement entropy between the qubits to the left of the cut at $x$ and those to right.  Then run the circuit and obtain $S(x,t)$ in the regime $t\gg 1$, $|x|<L/4$, and $L$ large enough so that the results do not depend on $L$.  $\mathcal{E}(v)$ is then obtained from $S(vt,t) \approx \mathcal{E}(v) t$ (or more generally an integral---this equation assumes a coarse-grained homogeneity in space and time). This allows us to gather data for many angles at once: one end of the membrane is pinned at $(x,t)=(0,0)$ by the initial state, and the other end is at $(x,t)$ since we are looking at $S(x,t)$ of the state at time $t$. 
Note that $S(x,t)$ can change by at most one if we change $x$ by one qubit.  This restricts $\mathcal{E}(v)$ to be continuous, although it can have discontinuities in its first derivative; we will call these discontinuities in $d\mathcal{E}/dv$ ``cusps''.

\begin{figure}[t]
\centering
\subfloat[\label{fig:page-sketch}]{%
          \includegraphics[width=0.85\linewidth]{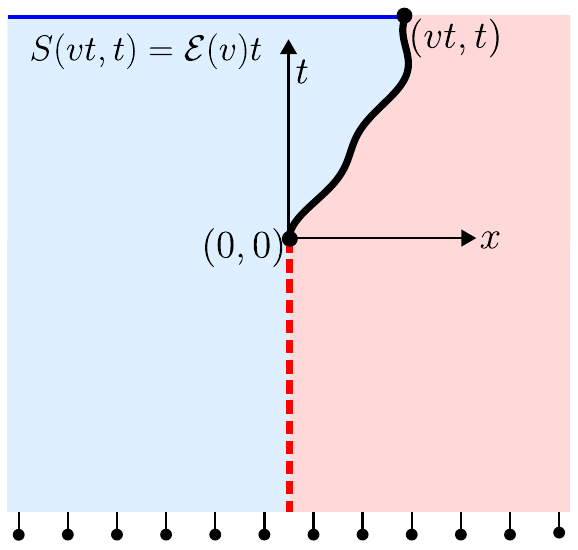}}\\
    \subfloat[\label{fig:spacelike1}]{%
          \includegraphics[width=0.85\linewidth]{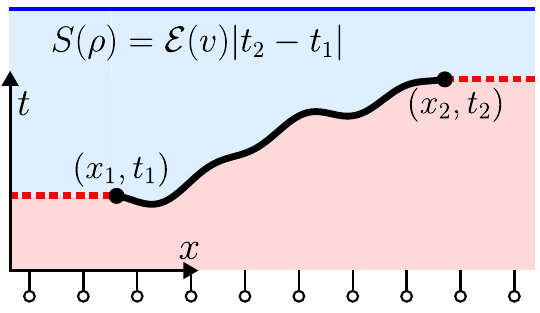}}
\caption{\gs{Probes of line tension. In each panel, the time direction is chosen to be vertical, the thick black curve shows a possible path of the membrane, and the blue line on the top boundary indicates the (sub)system whose entropy is measured at time $t$. (a) Protocol to determine $\mathcal{E}(v)$ by tracking subsystem entropies $S(x=vt,t)$ starting from two separately scrambled systems at $t=0$. This also illustrates one way to prepare the $t=0$ state: evolve a product initial \david{pure} state (closed circles on bottom boundary) under the chosen dynamics, but omit gates that cross the red dashed cut for $t<0$. 
(b) An alternative probe is to evolve a fully mixed initial state (open circles on bottom boundary) under the chosen dynamics, measure all qubits in the interval $[0,x_1]$ at time $t_1$ and the interval $[x_2,L]$ at time $t_2$, and compute the entropy of the full state at time $t>t_1,t_2$. }
} 
\end{figure}
An alternative approach to compute the line tension for spacelike membranes~\gs{is shown in~\autoref{fig:spacelike1}}: Initialize the circuit in a maximally mixed state (or equivalently, with each system qubit in a Bell state with a reference qubit). Make cuts (perform rank-1 projective measurements) on either side of a finite-size chain, i.e. measure all qubits in the interval $[0,x_1]$ at time $t_1>0$ and $[x_2,L]$ at time $t_2>0$, with $x_2>x_1$, $t>t_2,t_1$. The membrane separating the initial state from the final state will follow the cuts and thus run through the uncut part of the circuit from 
$(x_1,t_1)$ to $(x_2,t_2)$, so with $v=(x_2-x_1)/(t_2-t_1)$. The entropy of the final state (or equivalently the entanglement between the references and the final state) is then $\mathcal{E}(v)|t_2-t_1|$.  Again, we take the times, $(x_2-x_1)$, and $L$ all large enough so the resulting $\mathcal{E}(v)$ does not depend on them at a given $v$. While this is not {an ``efficient''} numerical method, since only one angle can be probed in each run, it is a natural way to define the line tension for spacelike membranes.

An alternative quantification of the tension that is not singular in the limit $v\rightarrow\infty$ ($t_1\rightarrow t_2$) is the entropy density per (Euclidean) unit length (per unit ``area'' in higher $d$) of the membrane $\ef(\vec{\hat{n}})=(1+v^2)^{-1/2}\mathcal{E}(v)$, where we denote the membrane orientation with the unit normal $\vec{\hat{n}}$. For a given orientation $\vec{\hat{n}}$ of the membrane, $\ef$ does not depend on the choice of time direction, unlike $\mathcal{E}(v)$. 

\subsection{Dynamics across the membrane}\label{sect:dam}

So far, we have mostly discussed circuits that are unitary along the conventional time direction, but the circuit need not possess any strictly unitary arrow of time. More generally, given a circuit made up of geometrically local (not necessarily unitary) gates, one can set any direction to be the ``time'' direction and construct a transfer matrix that evolves states along that direction. Using the polar decomposition, each gate appearing in this transfer matrix can be written  as a combination of a unitary and a generalized measurement with a (possibly) post-selected measurement outcome. For any orientation of the membrane, we can choose to run time along the normal direction. The membrane tension $\ef$ can then be interpreted as the entropy density of an initially fully mixed state under this hybrid dynamics (\autoref{fig:spacelike2}). If this entropy density settles to a constant value in the limit $L\rightarrow\infty$, followed by $t\rightarrow \infty$, one can interpret it as the density of information flowing along the direction normal to the membrane. \gs{This interpretation of entanglement entropy as information flux appears in the holography literature under the name of "bit threads": Bell pairs between a region and its complement carried by a vector field whose flux through the region yields its entropy~\cite{Agon2021}.}

\begin{figure}[t]
\centering
\includegraphics[width=0.9\linewidth]{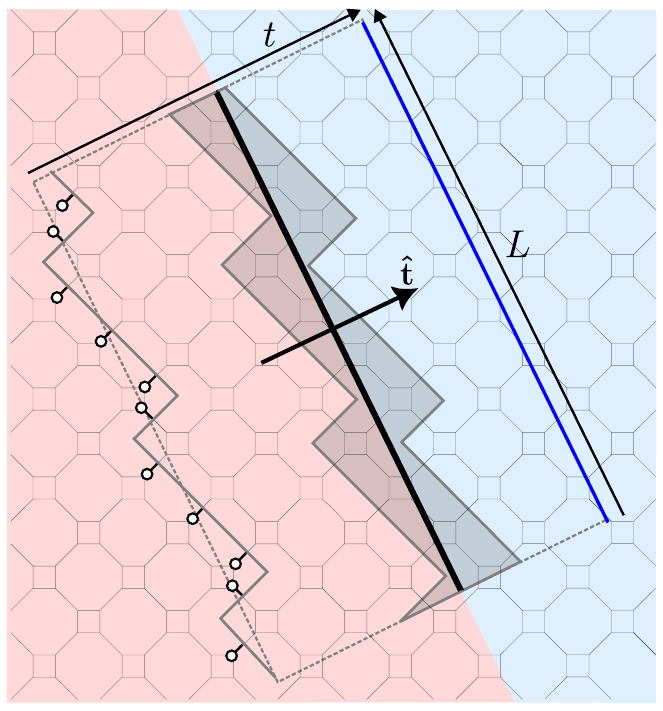}
\caption{\gs{Tensor network of not necessarily unitary \david{and possibly random} gates. For an arbitrary "time" direction $\vec{\hat{t}}$, one may evolve a fully mixed initial state (open circles) on $L$ qubits under $t$ time steps (one time step is shaded in gray). If the final entropy density converges to a constant value in the limit $L\rightarrow\infty$ \david{then} $t\rightarrow\infty$, this density is the free energy density of a membrane (thick black line) with normal vector $\vec{\hat{t}}$. \label{fig:spacelike2}}}
\end{figure}
Let's denote this dynamics across the membrane as the "\dam". 
The \dam~ appear to fall into four general categories. From most to least constrained, these are:

(i) The \dam~are precisely unitary. This is the case when there is some choice of the time direction $\vec{\hat{t}}$ that defines a unitary time evolution, as long as the membrane surface runs in a spacelike direction (outside the light cone) with respect to $\vec{\hat{t}}$. In this case $\mathcal{E}(v)$ is set by the size of the Hilbert space (the entropy density) that this unitary dynamics is acting on.  The dynamics across the membrane is qualitatively the same for all membrane orientations that are spacelike for this $\vec{\hat{t}}$.  We will call this set of membrane orientations a ``light sector''.  This gives $T=0$ behavior and $\mathcal{E}(v)$ linear in $v$ within this light sector. If there is no strict light cone (as in continuous time nonrelativistic lattice dynamics), this light sector consists of only a single membrane orientation, but the range of orientations with this unitary~\dam~widens upon Trotterization.

(ii) The \dam~are not strictly unitary, but have a permanently stable plateau, which is a subspace with extensive entropy whose spectrum is invariant in time. Starting in the maximally mixed state, the system evolves to this subspace in its steady state. Like the strict unitarity in case (i), this fixes the free energy density of the membrane across which the dynamics runs to the entropy density of the plateau.  This means that the interface is effectively at zero temperature $T$, since at $T\neq 0$, for finite $L$ the free energy would instead continuously evolve with increasing time as terms are added to the partition function.

In many circuits 
the dynamics is unitary and local within this subspace. \gs{Then the dynamics has a light cone, and this same dynamics applies over all membrane orientations for which the membrane lies outside the light cone. This range of orientations thus defines a light sector}, within which $\mathcal{E}(v)$ is linear in $v$.  

We will call a circuit multi-unitary if it has at least two distinct light sectors with 
unitary \dam~of types (i) or (ii).  The circuit is {\it fully} multi-unitary if the union of all of its light sectors covers all membrane orientations (except possibly the boundaries between sectors); \gs{see~\autoref{sect:tension-2} and \autoref{fig:sectors} below.}

This definition of full multi-unitarity 
applies to any dynamics composed of "multi-directional unitary gates"~\cite{Mestyan2022}, including dual-~\cite{Gopalakrishnan2019,Bertini2019,Piroli2019,Bertini2020,Bertini2021}, tri-~\cite{Jonay2021}, and ternary-unitary~\cite{Milbradt2023} circuits. It also applies to the first two levels of the hierarchy of generalized DU (gdU) circuits~\cite{Yu2023,Foligno2023,Rampp2023}, whose sideways dynamics is unitary within a subspace, and to many STTI Clifford circuits.  These multi-unitary circuits enjoy the further property that the steady state $\rho$ for any membrane orientation has a uniform spectrum, so the entropy is independent of the Renyi index. Thus, not only is the entanglement membrane at zero temperature, but it is described by a lattice model with integer energies. 

More generally, the three properties which allow us to define light sectors---(1) stable plateau, (2) unitarity within the plateau subspace, and (3) locality-preserving {unitary} dynamics within the subspace---need not occur together. While (3) implies (2) and (2) implies (1), the converses do not hold in general. One example is the mixed phase of certain non-Hermitian models, which has a plateau subspace within which the dynamics is not unitary due to dephasing of non-orthogonal eigenvectors~\cite{Gopalakrishnan2021}. (The spectrum is also not degenerate in this case, so the line tension would depend on Renyi index, but each tension will still be zero-temperature-like due to property 1.) Such models are outside the scope of the present paper.

In STTI Clifford circuits, the existence of a permanently stable plateau \textit{does} imply unitarity within a subspace~\cite{supp-ref}, so these circuits are gDU in the weak sense of having a set of membrane orientations for which the \dam~are unitary within a subspace. However, generalized dual-unitarity can also be defined in a stricter sense as a condition at the gate level, which makes the influence matrix factorize at least for some range of $v$. Explicitly, for the first two levels of the hierarchy, DU and DU2, it factorizes for all $v$; for higher levels of the hierarchy, the influence matrix only factorizes for $|v|>v^*$~\cite{Rampp2023}. This factorization imposes locality (property 3) on the effective unitary dynamics, confining two-point correlations to single-qubit channels along specific "flow directions" (described below), as also occurs in the type-(i) multi-unitary circuits. Several of the $1+1$d STTI Clifford circuits that we have studied can be expressed in terms of the DU2 conditions and thus provably possess these flow directions~\cite{supp-ref}.  For larger unit cells, we do not have an explicit construction that would identify the flow directions and light sectors, but for the cases that we have studied the line tension is always piecewise linear in $v$, and we have yet to find a 1+1d circuit with a line tension that has cusps at more than 3 different membrane orientations. 

(iii) The \dam~has an {\it emergent} plateau: a long-lived steady state {entropy} which purifies only occasionally.  This is the case, for example, in the mixed (volume-law) phase of random monitored Clifford circuits~\cite{Li2018,Li2019,Gullans2019,Gullans2020}, whose plateau is exponentially long-lived. Starting from a fully mixed state (a stabilizer state), the Renyi entropies remain equal at all times, and in particular are set by the rank (Hartley) entropy. The rank entropy takes only integer values, so in a finite system of length $L$ with an exponential purification time, the purification events which take the state to a lower entropy plateau are exponentially rare in $L$.  This dynamics again matches that of a $T=0$ membrane, whose free energy is set by the minimum energy and thus only changes when a new location for the membrane appears that sets a new record low for the energy~\cite{supp-ref}.  This is in accord with the observation in Ref.~\cite{Li2023} that the Hartley entropy of random monitored circuits in the volume-law phase is governed by the zero-temperature limit of the DPRM; for Cliffords, the Hartley entropy is the same as all the other entropies, again suggesting that the DPRM picture of the entanglement membrane for random Clifford circuits is at $T=0$.
 
Generalizing beyond hybrid Clifford circuits with one-site projective measurements, the dynamics across a membrane run at any arbitrary angle in a random Clifford circuit will likewise generically host an emergent plateau, provided that it does not purify completely. Indeed, given a circuit and an orientation of time with only Clifford gates and stabilizer projectors and/or measurements \gs{(e.g., the vertical direction in~\autoref{fig:spacelike2}, with each 4-legged tensor a Clifford gate)}, the dynamics of that circuit along any direction \gs{(e.g., across the membrane indicated in~\autoref{fig:spacelike2})}, can again be expressed in terms of Clifford gates, projectors, and/or measurements~\cite{Ippoliti2021}.  In the absence of (multi-)unitarity, 
the membrane is in a random medium which generally breaks the piecewise linearity of $\mathcal{E}(v)$, while remaining at $T=0$. 

We leave open the question of whether there are any non-Clifford and non-unitary circuits which contain membranes whose \dam~support an emergent plateau. These circuits would have to be highly fine-tuned (but not gDU).

(iv)  Finally there is the general case that is nonunitary, non-Clifford, and "generic", where we do not expect a plateau in the above sense, and thus none of these $T=0$-like features. For example, although the purification time is exponential in system size in the volume-law phase of Haar-random monitored circuits~\cite{Skinner2018}---\gs{so that the $L\rightarrow\infty, t\rightarrow\infty$ limit in~\autoref{fig:spacelike2} does converge}---the density matrix eigenvalues, being nondegenerate, evolve nontrivially rather than reach an emergent plateau. The entanglement and purification dynamics of such circuits is quite rich, requiring tools outside the scope of this paper, such as techniques from random matrix theory~\cite{Bulchandani2023,DeLuca2023,Gerbino2024}. 

In a generic unitary circuit, \gs{with a unitary arrow of time along $v=0$}, a membrane running at a slope between the light cone speed $v_{LC}$ and the butterfly speed $v_B$, may also be subject to thermal fluctuations, due to the exponential tails outside the butterfly cone. However, from general constraints~\cite{Jonay2018}, \gs{$\mathcal{E}(v) = s_{eq}|v|$ for all $|v|\geq v_B$}, {so some of the zero-temperature-like properties that the membranes have for $v>v_{LC}$ may remain for $v>v_B$.}\footnote{The functional form of $\mathcal{E}(v)$ in random unitary circuits continued from $|v|\leq v_B$ to $v_B < |v| \leq v_{LC}$ does relate to the exponential decay of the OTOC outside the butterfly cone in random circuits~\cite{Jonay2018,Khemani2018,Zhou2019}, but does not determine the leading order scaling of the entanglement across a membrane in this interval.}

\section{Membrane tensions and multi-unitarity}\label{sect:tension-2}
In the remainder of this paper, we focus on fully multi-unitary models {whose membranes are in class (i) and/or (ii).}

In general, circuits can have full multi-unitarity with a high symmetry, but the symmetry may be hidden by a particular microscopic representation of the circuit, or only emergent. \gs{An invertible linear transformation of the coordinates $(x,t)$---i.e., linearly scaling, shearing, and/or rotating the circuit---can often bring out the symmetry and a familiar functional form of the line tension:
\begin{align}\label{eq:transform}
    (x,t) &= (ax'+bt', cx'+dt') \notag \\
    \Rightarrow \mathcal{E}'(v') &= (cv'+d)\mathcal{E}\left(\frac{av'+b}{cv'+d}\right).
\end{align}
For concrete examples of these transformations, see~\autoref{sect:examples} and the Supplement~\cite{supp-ref}.} 

\begin{figure}[hbtp]
\subfloat[]{\includegraphics[width=0.48\linewidth]{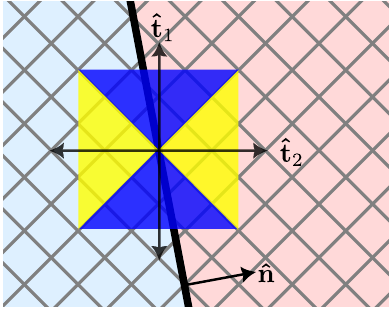}\label{fig:dual}}\hfill
\subfloat[]{\includegraphics[width=0.48\linewidth]{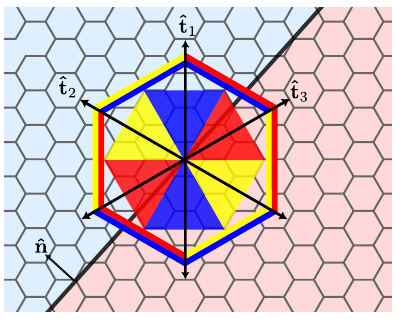}\label{fig:tri}}
\caption{\gs{Emergent (a) square and (b) hexagonal symmetries of models whose line tension can be brought into the form of~\autoref{eq:dual-tension} and~\autoref{eq:tri-tension} respectively. In each panel, the thick black line is a membrane that is spacelike with respect to $\vec{\hat{t}}_2$ (lies outside the yellow light cone), with normal vector $\vec{\hat{n}}$ pointing inside the corresponding yellow light sector, so that~\autoref{eq:normal-vector} holds with $\vec{\hat{t}} = \vec{\hat{t}}_2$. (a) has two pairs of light sectors, which are coincident with the light cones centered on (generalized) unitary directions $\vec{\hat{t}}_1$ and $\vec{\hat{t}}_2$. (b) has three pairs of light sectors (filled triangles) corresponding to three overlapping light cones (red, yellow, and blue borders).} 
\label{fig:sectors}}
\end{figure}

For instance, any 1+1d circuit for which $\mathcal{E}(v)$ is piecewise linear and has only two cusps can be scaled, sheared and rotated to give it a (possibly emergent) square symmetry, so that the two cusps are at $v=\pm 1$, with the standard dual-unitary behavior~\cite{Piroli2019,Zhou2020}:
\begin{equation}\label{eq:dual-tension}
\mathcal{E}(v) = \begin{cases}
|v| & |v| \geq 1 \\
1 & |v| \leq 1~.
\end{cases}
\end{equation}
The per-unit-length line tension then satisfies:
\begin{subequations}
\begin{align}
\ef(\theta) &= \cos\theta \quad (0\leq \theta \leq \pi/4) \label{eq:dual1} \\
\ef(\theta) &= \ef(\theta + \pi/2) = \ef(-\theta)~,
\end{align}
\end{subequations}
where the orientation of the membrane is denoted with $v=\cot{\theta}$; this explicitly exhibits the square symmetry, \gs{illustrated in~\autoref{fig:dual}.}

We can also re-examine tri-unitary 1+1d circuits (piecewise linear $\mathcal{E}(v)$ with 3 cusps) from this perspective.  Many such circuits can be linearly scaled, sheared and rotated to have a (possibly emergent) hexagonal symmetry, with
\begin{equation}\label{eq:tri-tension}
\mathcal{E}(v) = \begin{cases}
|v| & |v| \geq \sqrt{3} \\
\frac{|v|+\sqrt{3}}{2} & |v| \leq \sqrt{3}~,
\end{cases}
\end{equation}
so
\begin{subequations}
\begin{align}
\ef(\theta) &= \cos\theta \quad (0\leq \theta \leq \pi/6) \label{eq:tri1} \\
\ef(\theta) &= \ef(\theta+\pi/3) = \ef(-\theta)~.
\end{align}
\end{subequations}
Thus $\ef(\theta)$ is periodic with cusps at the boundaries between the six sectors of a "light hexagon" (\autoref{fig:tri}).\footnote{Note: angles $\theta$ and $(\pi + \theta)$ represent the same membrane orientation, so there are only 3 independent light sectors.}

\gs{The DU2 circuits of Ref.~\cite{Yu2023,Foligno2023,Rampp2023} provide a simple example of how the explicitly hexagonal-symmetric form of~\autoref{eq:tri-tension} can be obtained for tri-unitary circuits. All entangling DU2 circuits are either the standard dual-unitary or have a reflection-symmetric line tension with cusps at $v=\pm 1, v=0$ in the units of the brickwork circuit. The latter class of circuits are thus tri-unitary by our definition. In brickwork circuit units, 
\begin{equation}\label{eq:brickwork-tri}
    \mathcal{E}(v) = \begin{cases}
        v_E + (1-v_E)|v|  & |v| < 1 \\
        |v| & |v| \geq 1.
    \end{cases}
\end{equation}
When the entanglement velocity $v_E\equiv \mathcal{E}(0) =1/2$,~\autoref{eq:brickwork-tri} can be brought into the form of~\autoref{eq:tri-tension} just by rescaling $t\rightarrow t/\sqrt{3}$. 
On qubits, $v_E=1/2$ is the only allowed value of the entanglement velocity~\cite{Rampp2023}}, and the emergent hexagonal symmetry can be made explicit by expressing each gate as a dressed CNOT gate~\cite{Balakrishnan2011,Cohen2013} and representing it on a honeycomb lattice~\cite{Liu2023} using ZX calculus~\cite{VandeWetering2020,supp-ref}. On higher-dimensional qudits, however, there exist examples of DU2 gates with $v_E\neq 1/2$~\cite{Foligno2023,Rampp2023}, for which the rectangular symmetry of the original circuit cannot be promoted to hexagonal symmetry.\footnote{\gs{In principle, ``monoclinic'' tri-unitary circuits that cannot be scaled to have any extra symmetry should also be possible, although we have not yet found any such examples.}}

There is a third way of parameterizing the membrane tension that is useful for $d>1$:  For a given choice of the time and spatial directions, define the membrane surface as the function $t(\vec{r})$, where $\vec{r}$ is a $d$-component spatial vector such that the point $(t(\vec{r}),\vec{r})$ lies on the membrane. Then, we can quantify the membrane orientation by the gradient  $\vec{g}=\nabla t(\vec{r}) = \vec{v}/v^2$.  We then define the tension 
$\mathcal{E}_{\vec{r}}(\vec{g})$ as the free energy per unit area projected onto an \textit{equal time} surface (not in the surface defined by the membrane), that is,
\begin{equation}
\mathcal{E}_{\vec{r}}(\vec{g}) = \ef(\vec{\hat{n}})/(\vec{\hat{n}}\cdot \vec{\hat{t}}) = \ef(\hat{\vec{n}}) \sqrt{1 + (\nabla t (\vec{r}))^2}.
\end{equation}
$\mathcal{E}_{\vec{r}}(\vec{g})$ is the quantity that, for membrane dynamics in classes (i) and (ii), remains a piecewise linear function of its argument $\vec{g}$ within a light sector, no matter how we linearly deform spacetime, and in any dimension $d$.\footnote{In $d=1$, $g = 1/v$ and $\mathcal{E}_x(g) = \mathcal{E}(v)/v$, so if $\mathcal{E}_x(g)$ is piecewise linear in $g$, then $\mathcal{E}(v)$ is piecewise linear in $v$ and vice versa. However, the piecewise linearity of $\mathcal{E}(v)$ in $d=1$ does not simply generalize to $d>1$.}  

The commonality between~\autoref{eq:dual1} and~\autoref{eq:tri1} points to a more general form for the membrane tension, which holds in higher dimensions as well. Consider a single light sector.  For that light sector, there is a choice of the time direction $\vec{\hat{t}}$ that makes $\mathcal{E}_{\vec{r}}(\vec{g})$ independent of $\vec{g}$, with $\mathcal{E}_{\vec{r}} = \ef(\vec{\hat{t}})$, which is the plateau entropy density for this time direction.  Within that light sector, we then have
\begin{equation}\label{eq:normal-vector}
\ef(\vec{\hat{n}}) = (\vec{\hat{n}} \cdot \vec{\hat{t}} ) \ef(\vec{\hat{t}})~.
\end{equation}
Heuristically, it is as if the entanglement entropy is ``flowing'' along the direction $\vec{\hat{t}}$ for all membrane orientations within this light sector. \gs{In the language of bit threads~\cite{Agon2021}, $\ef(\vec{\hat{n}})$ is the flux density on the minimal surface, and the optimal vector field carrying the bit threads points along $\vec{\hat{t}}$ for the entire light sector.}\footnote{\gs{That is, in the notation of Ref.~\cite{Agon2021}, $w^{\mu} = \ef(\vec{\hat{t}})\vec{\hat{t}}$ when evaluated on the minimal membrane.}} The direction of this ``flow'' then changes discontinuously as the membrane orientation is rotated across the boundary between two light sectors.

In some of the fully multi-unitary models we have studied, \gs{there exist vectors $\vec{\hat{u}}_i$ with "capacity" $s(i)$ such that:
\begin{equation}\label{eq:entanglement-channels}
\ef(\vec{\hat{n}}) = \sum_i |\vec{\hat{n}} \cdot \vec{\hat{u}}_i| s(i).
\end{equation}
for any choice of $\vec{\hat{n}}$.
In such cases, we can alternatively interpret the entropy as flowing independently along these vectors, which we call "flow directions," with the $i$th direction carrying flow $s(i)$.} This constrains all possible membranes, rather than just a single light sector as in~\autoref{eq:normal-vector}. We do not have an information-theoretic understanding of this "flow," but empirically, when (i) a model possesses sets of unitary light cones such that the locus of points that are not spacelike with respect to any time direction are a discrete set of lines, then (ii)~\autoref{eq:entanglement-channels} is satisfied by taking the flow directions to be along these lines~\cite{supp-ref}.

While dual- and tri-unitary circuits have been studied before in some detail,~\autoref{eq:normal-vector}, along with its stronger form~\autoref{eq:entanglement-channels}, is a unifying principle that will prove useful in our treatment of more novel gDU circuits in 2+1d. 

\section{Case studies in $1+1$ and $2+1$d}\label{sect:examples}
We now highlight some STTI Clifford circuits in 1 and 2 spatial dimensions, situating them within the broad framework of the previous sections. All of our circuits have features that go beyond prior works (larger unit cell, measurements, higher dimensions) but fall within the broad umbrella of generalized multi-unitarity.

\subsection{$1+1$d Clifford circuits}
A $(b, a, d)$ STTI brickwork circuit is defined to have a unit cell that is of depth $b$ layers of gates in time and is $a$ qubits wide.  The lattice translations are: spatial translation by $a$ qubits, with $a$ even; and simultaneous time translation by $b$ layers and spatial translation by $d$ qubits, with $b$ and $d$ having the same parity.

\begin{figure}[t]
    \centering    \includegraphics[width=\linewidth]{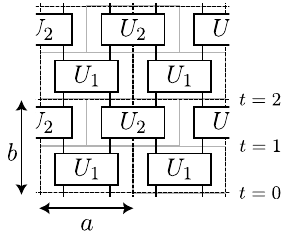}
    \caption{STTI brickwork circuit with ($b=2, a=2, d=0$) unit cell (indicated by dashed lines). Time is measured in the number of layers. \gs{If $U_1=U_2$, then the unit cell is $(b=1, a=2, d=1)$, outlined in light gray.}}
    \label{fig:cell}
\end{figure}

\autoref{fig:cell} shows a $(b=2,a=2,d=0)$ unitary circuit, whose unit cell contains two 2-qubit gates. Each gate can be decomposed into a 2-qubit "core", which for Clifford gates is either dual-unitary (SWAP, iSWAP) 
or not ($\mathbbm{1}$, CNOT), dressed by single-qubit gates; we will call the single-qubit gates $u_+$, $u_-$. 

\subsubsection{CNOT-core circuits}\label{sect:cnot}

As an extension of earlier works on DU2 circuits \gs{in which every gate is identical~\cite{Yu2023,Foligno2023,Rampp2023}, we first study a reflection-symmetric CNOT-core circuit with two gates per unit cell}. We set
\begin{equation}
U_1 = \cnot (u_+ \otimes u_-)~,
\end{equation}
and make the circuit left-right reflection-symmetric (modulo a space-time translation) by placing the spatially reflected gate in the layers after odd times:
\begin{equation}
U_2 = \mathrm{NOTC} (u_- \otimes u_+)~.
\end{equation}
The scrambling properties of this circuit \comment{and whether or not it's triunitary} depend upon $u_\pm$.  We focus on the class of good scrambling circuits with the densest operator spreading in the sense defined in Ref.~\cite{Sommers2023}, a representative of which has 
\begin{align}\label{eq:idx9}
u_+ &= \exp[-i \pi X/4] = R_X[\pi/2] \notag \\
u_- &= \exp[i\pi (X+Y+Z)/3\sqrt{3}] = R_{(1,1,1)}[-2\pi/3].
\end{align}
Using ZX calculus to express this circuit on the honeycomb lattice, as in Ref.~\cite{Liu2023}, we find that this choice of single-qubit gates makes the circuit tri-unitary, with line tension in the hexagonally symmetric honeycomb geometry given by~\autoref{eq:tri-tension}~\cite{supp-ref}.

While tri-unitarity allows us to determine $\mathcal{E}(v)$ without using any of the numerical probes discussed above, these numerical methods still provide insight into the entanglement dynamics. In the protocol depicted in~\autoref{fig:page-sketch}, the two Page curves at $t=0$ are prepared by feeding an initial pure product state into a scrambling circuit with open boundary conditions, omitting gates along the vertical cut. We can choose the scrambling circuit to be just a random Clifford circuit, thus preparing the left and right subystems in random stabilizer states; or we can use the STTI circuit itself to separately scramble the two halves. The latter implementation allows us to distinguish between "good scrambling" tri-unitary circuits such as that shown in~\autoref{fig:page}, whose $t>0$ entanglement dynamics reproduces the expected line tension after rescaling time, and non-ergodic poor scramblers with the same line tension whose entanglement generation for an initial product state (not a solvable state) is limited by the presence of "gliders"~\cite{Bertini2020soliton,Jonay2021}. {Using an alternative method, we also verify that the piecewise linearity of $S(x,t)$ is not just for Clifford initial states, but holds for typical Haar-random initial states as well~\cite{supp-ref}.} 

\begin{figure}[t]
\includegraphics[width=\linewidth]{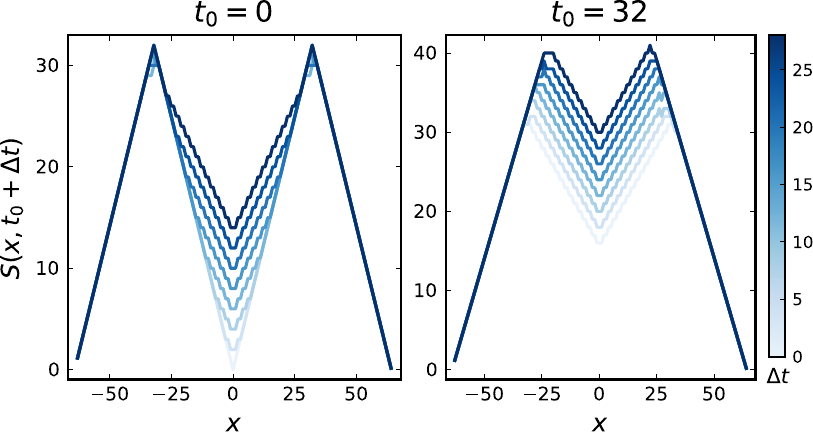}
\caption{Entanglement growth in a CNOT-NOTC circuit with single qubit gates given by~\autoref{eq:idx9}, for $L=128$, using the protocol in~\autoref{fig:page-sketch}. Times in increments of 4 layers are shown. \label{fig:page}}
\end{figure}

\subsubsection{Hybrid circuits}\label{sect:dual-hybrid}
Next, as a demonstration of how generalized dual-unitarity can appear in hybrid circuits, again consider a brickwork circuit with unit cell $(b=2,a=2,d=0)$ as in~\autoref{fig:cell}, where the left incoming leg to each $U_1$ gate is now measured in the $X$ basis. Take
\begin{equation}\label{eq:dual-hybrid-gate}
U_1 = U_2 = \iswap (R_X[\pi/2] \otimes R_{(1,1,1)}[-2\pi/3]).
\end{equation}
In the absence of measurements, this choice of gate would produce a good-scrambling dual-unitary circuit~\cite{Sommers2023}, with the line tension given in~\autoref{eq:dual-tension}.
When we add measurements, however, the entanglement $S(x,t)$ of this "dual hybrid" circuit grows as in~\autoref{fig:dual-hybrid}a.

\begin{figure}[t]
\includegraphics[width=\linewidth]{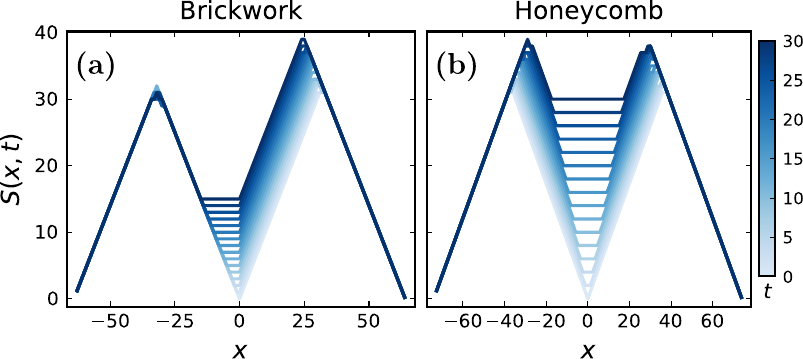}
\caption{(a) Entanglement growth starting from the product of random stabilizer states, in a hybrid circuit with dual-unitary gate~\autoref{eq:dual-hybrid-gate}. ($|x|$ is the number of unmeasured qubits left or right of the central cut.) (b) Entanglement growth in the same circuit deformed onto the honeycomb lattice.\label{fig:dual-hybrid} Times in increments of two layers are shown.}
\end{figure}
This circuit is a case in which the microscopic representation (as a brickwork circuit, on the square lattice) hides the underlying symmetry of a fully multi-unitary model. On the square lattice, the inclusion of measurements strongly breaks the left-right symmetry,
manifesting in an infinite butterfly velocity to the right, and finite butterfly velocity to the left. 
But by deforming the circuit onto a honeycomb lattice~\cite{supp-ref}---\gs{namely, applying a horizontal shear in the $(x,t)$ plane, while still running the dynamics along $t$}---we can restore symmetry between left and right and in fact obtain full hexagonal symmetry (\autoref{fig:dual-hybrid}b), with a line tension of
\begin{equation}\label{eq:honeycomb-spatial}
\mathcal{E}(v) = \begin{cases}
1 & |v| < \frac{1}{\sqrt{3}} \\
\frac{|v|\sqrt{3} + 1}{2} & |v| \geq \frac{1}{\sqrt{3}}.
\end{cases}
\end{equation}
This is precisely the tri-unitary line tension (\autoref{eq:tri-tension}) with the conventional time and space directions flipped. Indeed, a $\pi/2$ rotation in the honeycomb lattice representation transforms the dual hybrid circuit with the gate given in~\autoref{eq:dual-hybrid-gate} into the Clifford East model, a DU2 circuit in which every gate is a CNOT~\cite{Gopalakrishnan2018,Bertini2023,Yu2023}. 

More generally, whenever the time and space axes are chosen such that $v_B=\infty$, the method for extracting $\mathcal{E}(v)$ sketched in~\autoref{fig:page-sketch} can exhibit unwanted sensitivity to details of the boundary. This sensitivity also appears in the properties of the steady-state plateau group of hybrid circuits and recurrence times of pure states within the plateau subspace~\cite{supp-ref}. We emphasize that while these features may appear to violate the qualitative picture of an entanglement interface with a local free energy, the underlying (local, zero-temperature) membrane behavior can be uncovered by running the dynamics through a different angle so that no information flows along the space direction.

\subsection{$2+1$d Clifford circuits}\label{sect:2d}
To illustrate these concepts in higher dimensions, we next present two examples of generalized multi-unitary Clifford circuits in 2+1d. Detailed numerics, {along with a third model which is ternary-unitary~\cite{Milbradt2023}}, are left to the Supplementary Material~\cite{supp-ref}.
\subsubsection{Square lattice parity model}
In the square lattice parity model~\cite{Gopalakrishnan2018}, qubits are placed on the vertices of a square lattice. Denoting the two sublattices $A$ and $B$, in odd layers, a CNOT gate is applied from every $A$ vertex to each of its neighbors (which belong to the B sublattice):
\begin{equation}\label{eq:odd-slp}
\mathbb{U}_1=\prod_{v \in A}\prod_{n \in \mathcal{N}(v)} \cnot(v\rightarrow n)
\end{equation}
and in even layers, the control and target are reversed:
\begin{equation}\label{eq:even-slp}
\mathbb{U}_2 = \prod_{v\in A}\prod_{n \in \mathcal{N}(v)} \cnot(n\rightarrow v)
\end{equation}
Note that the product of the CNOTs acting on a given qubit from each of its neighbors will flip the target if and only if an odd number of the controls is in the state $\ket{1}$, hence the name "parity model." This model has no apparent local conserved quantities, and initially local operators spread with small highly structured gaps, within squares with corners at $(\pm t, 0, t), (0,\pm t,t)$~\cite{Gopalakrishnan2018}.

This model can be elegantly expressed as a ZX diagram~\cite{VandeWetering2020}, a tensor network representation of a quantum circuit in terms of Z (green) "spiders":
\begin{equation}\label{eq:z-spider}
\underbrace{\overbrace{\input{green-spider.tikz}}^{n}}_{m} = \overbrace{\ket{00...0}}^{n}\overbrace{\bra{0...0}}^{m} + e^{i\alpha} \overbrace{\ket{11...1}}^{n}\overbrace{\bra{1...1}}^{m}
\end{equation}
and X (red) spiders:
\begin{equation}\label{eq:x-spider}
\input{red-spider.tikz} = \ket{++...+}{\bra{+...+}} + e^{i\alpha} \ket{--...-}\bra{-...-}.
\end{equation}
In this representation, the square lattice parity model is a simple cubic lattice with X spiders on one sublattice and Z spiders on the other (\autoref{fig:square-lattice-parity}), all with phase $\alpha=0$, and $n+m=6$.  These spiders have cubic symmetry: e.g. the Z spider imposes the constraint that either the quantum states of the qubits on all six legs of that spider are $\ket{0}$ or they are all $\ket{1}$. 
Thus, while the unitarity of~\autoref{eq:odd-slp} and (\ref{eq:even-slp}) is no longer explicit in this simple cubic lattice ``ZX spider'' representation, the cubic symmetry is manifest. In particular, rotational symmetry implies that the model is likewise unitary if we instead run time along the $x$ or $y$ directions.  Therefore, the model is multi-unitary, with three pairs of light cones centered on these unitary arrows of time (\autoref{fig:light-cube}). Each light cone is a square pyramid within which the operators spread when time is oriented somewhere within that pyramid.


A smooth patch of membrane at any orientation will lie entirely outside, or on the boundary of, one of these three pairs of light cones, and thus its tension can be computed from unitarity. This results in 3 independent light sectors, pairs of pyramids that together form the "light cube" shown in~\autoref{fig:light-cube}. Like the light cones, the light sectors are centered on the unitary arrows of time, such that the surface tension of a membrane $\mathcal{M}$ with normal vector $\hat{\vec{n}}$ is determined by the inner product with the closest unitary arrow of time (cf.~\autoref{eq:normal-vector}):
\begin{equation}\label{eq:normal-slp}
\ef(\vec{\hat{n}}) = \max_i |\vec{\hat{n}} \cdot \vec{\hat{e}}_i|.
\end{equation}
It should be emphasized that while there are gaps between the light cones (at the corners of cube), the light \textit{sectors} covers all membrane orientations, making this model fully multi-unitary.  One can also consider the dynamics with the time direction pointed towards a corner of this light cube; that dynamics, which consists of spiders with $n=m=3$, is nonunitary within the full Hilbert space but unitary within a subspace [type (ii)] and has an infinite butterfly speed, so its ``light cones'' are all directions that are not orthogonal to this time direction.

\begin{figure}[hbtp]
\subfloat[]{
\includegraphics[width=0.49\linewidth]{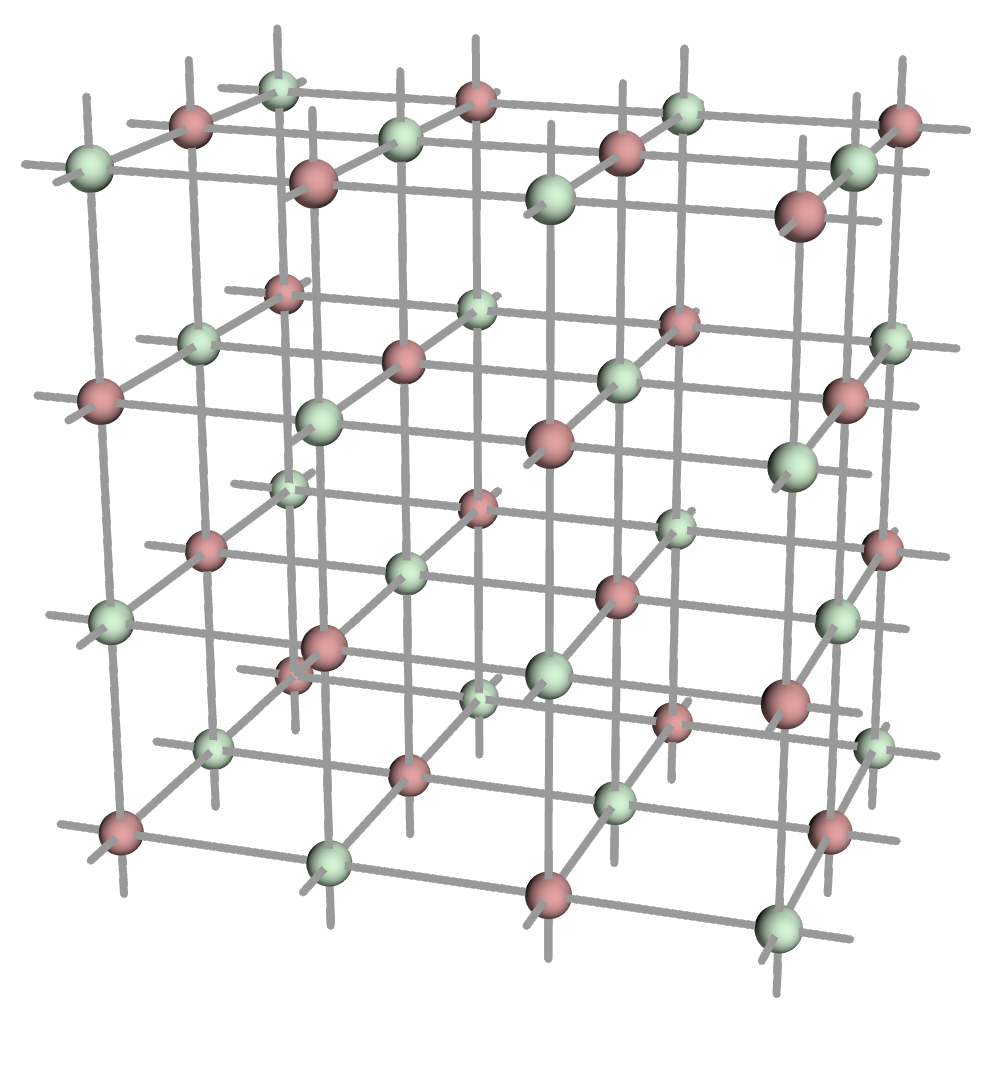}
\label{fig:square-lattice-parity}}
\subfloat[]{
\includegraphics[width=0.49\linewidth]{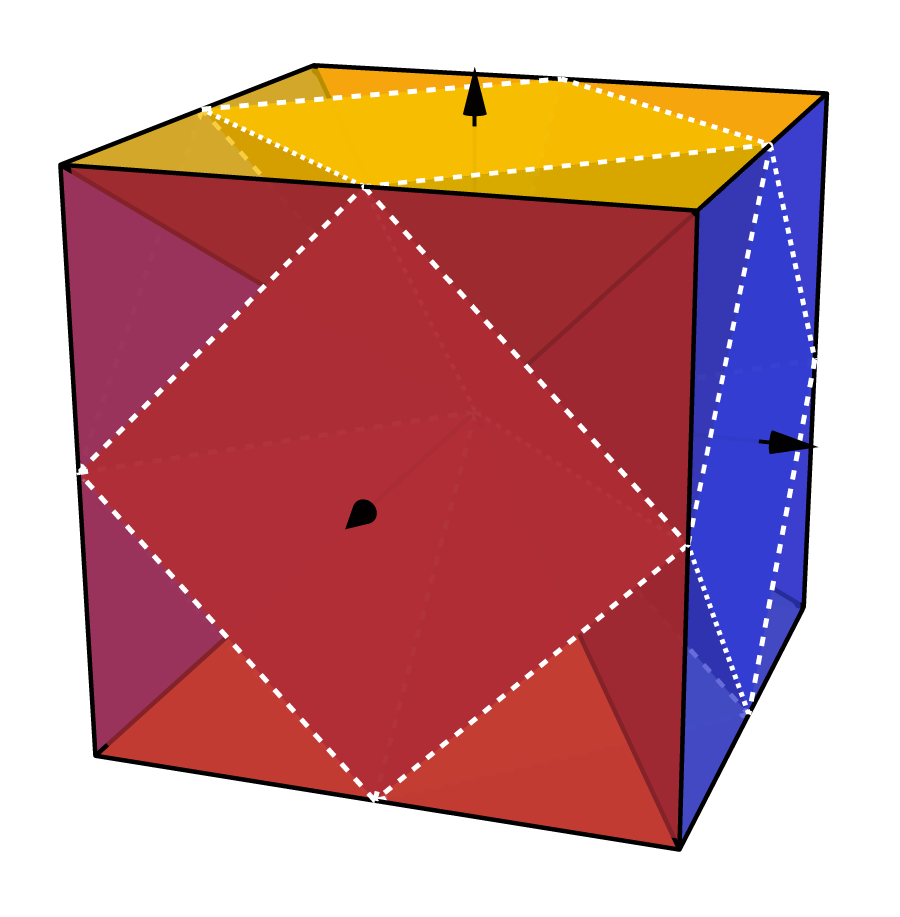}
\label{fig:light-cube}}\\
\caption{(a) Square lattice parity model as a simple cubic lattice of X (red) and Z (green) ``spiders''. A 2d input state is fed in at the bottom surface, and the final state is defined at the top surface. (b) Light cones (outlined in white dashed lines) and light sectors (outlined in black lines) of the square lattice parity model. Red, blue, and yellow pyramids are centered on the $x$, $y$, and $z$ "time arrows", respectively.}
\end{figure}

\begin{figure}[hbtp]
\subfloat[]{
\includegraphics[width=0.49\linewidth]{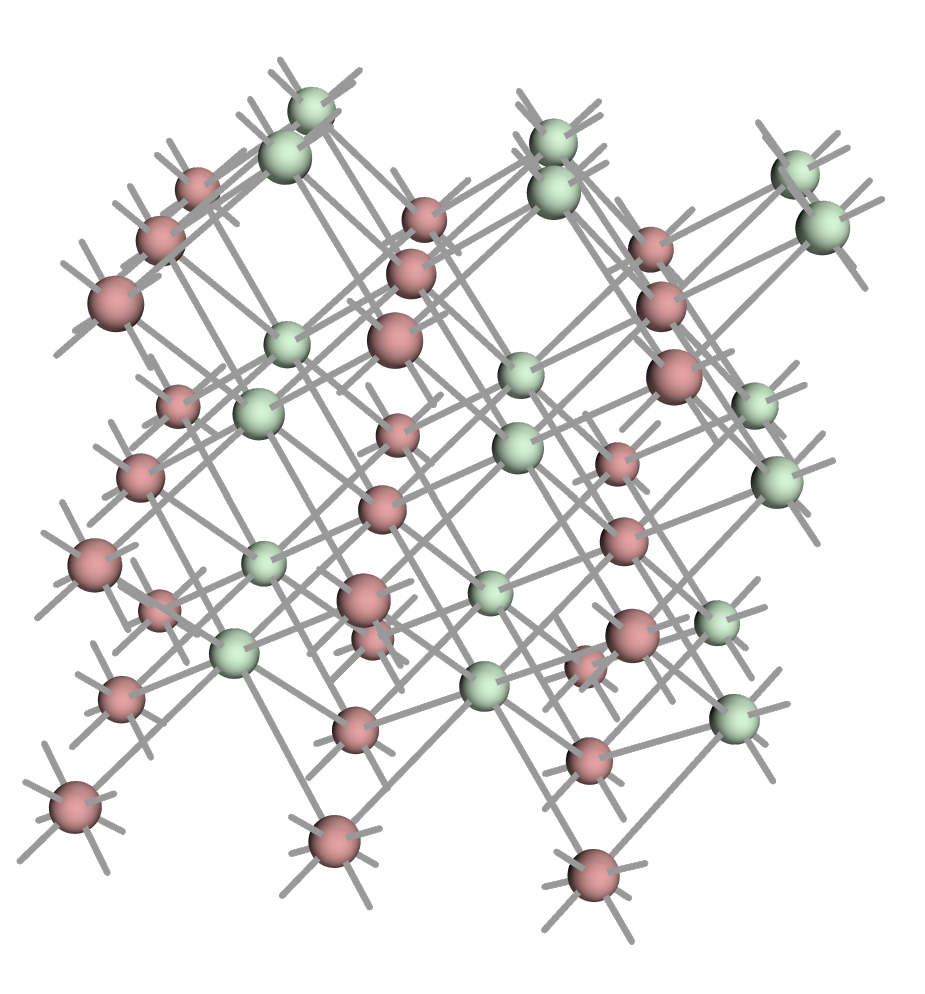}
\label{fig:bcc-lattice}}
\subfloat[]{
\includegraphics[width=0.49\linewidth]{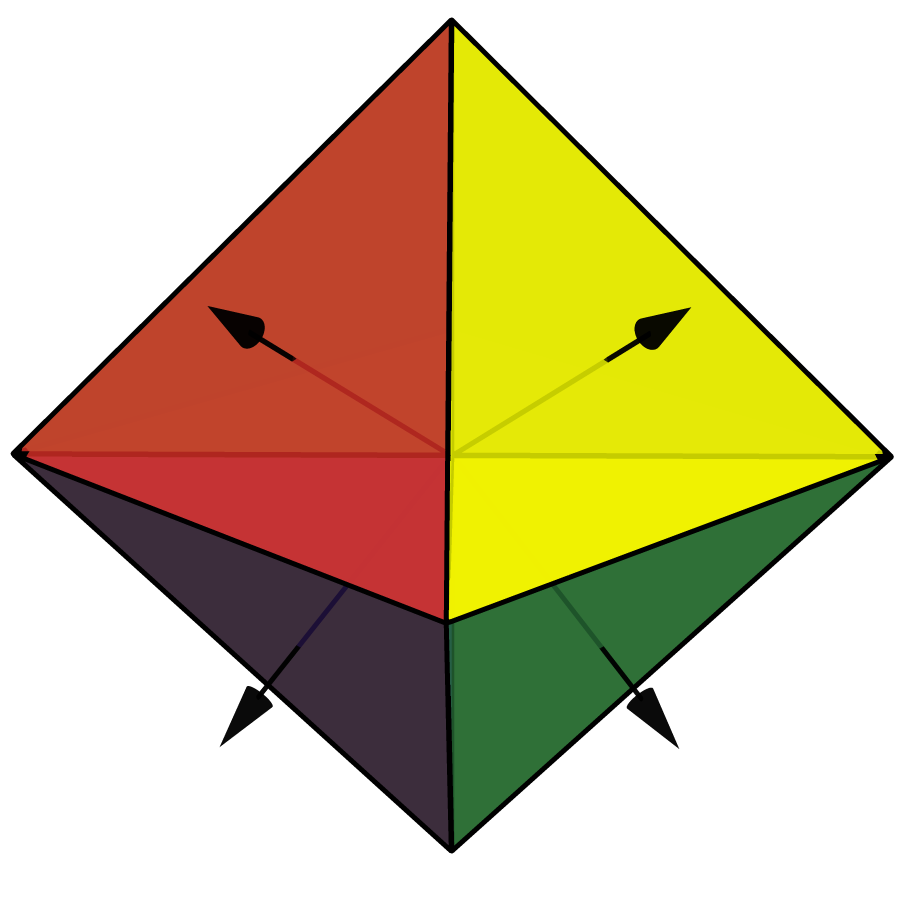}\label{fig:light-octahedron}}
\caption{(a) ZX diagram for the 2+1d model based on the bcc lattice. (b) Light octahedron. The black arrows at the centers of the yellow, blue, red, and green tetrahedra indicate the directions $\vec{t}_1,-\vec{t}_2,\vec{t}_3,\vec{t}_4$, respectively, along which the plateau entropy density is maximal. }
\end{figure}

\subsubsection{The bcc model}
The simple cubic form of the square lattice parity model in terms of X and Z spiders inspires a variation in which the spiders are placed on the vertices of a bcc lattice (\autoref{fig:bcc-lattice}).  Although this model does not have any obviously unitary interpretation, it is fully type-(ii) multi-unitary, with maximal plateau entropy density achieved by running time along the lines with $|x|=|y|=|z|$.  Running time along any of these directions, each spider has $n=m=4$, corresponding to a forced measurement of 3 stabilizers on 4 qubits.  Like the gDU circuits in 1+1d, the bcc model possesses discrete directions for information propagation/entropy flow along the $x$, $y$, and $z$ axes, in the sense that any spreading operator has an edge along one or more of these axes. These flow directions lie on the boundaries between light cones, which, unlike in the square lattice parity model, coincide with the light sectors: as shown in~\autoref{fig:light-octahedron}, the sectors are the four pairs of tetrahedra which compose an octahedron. 
In a given light sector $\mathcal{E}_{\vec{r}}(\vec{g})$ is constant if we choose $\vec{\hat{t}}$ to point along the maximal entropy density direction in that sector, that is,
\begin{equation}\label{eq:normal-bcc}
\ef(\vec{\hat{n}}) = \max_i |\vec{\hat{n}} \cdot \vec{t}_i|
\end{equation} 
where
\begin{align}\label{eq:octa-dir}
\vec{t}_1 = \begin{pmatrix} 1 \\ 1 \\ 1 \end{pmatrix},
\vec{t}_2 = \begin{pmatrix} -1 \\ 1 \\1 \end{pmatrix},
\vec{t}_3 = \begin{pmatrix} 1 \\ -1 \\ 1 \end{pmatrix},
\vec{t}_4 = \begin{pmatrix} 1 \\ 1 \\ -1 \end{pmatrix}. 
\end{align}
\section{Discussion and open questions}\label{sect:conclude}
In this paper, we have provided an interpretation of the entanglement membrane in terms of the flow of entropy across the membrane, and argued that when the dynamics across the membrane has a permanently stable plateau, or an emergent plateau with rare purification events, its free energy is consistent with zero temperature. We present a set of numerical techniques for evaluating the membrane tension and illustrate how these methods recover recent results~\cite{Foligno2023,Liu2023,Rampp2023} on membrane tensions in generalized multi-unitary circuits. Our explorations of new classes of solvable models in 2+1d extend to higher dimensions the geometrical interpretation of multi-unitary circuits as possessing a set of {"light sectors"; the different light sectors are all equivalent under symmetries in the two examples we discuss above.}

Our focus in this work has been on volume-law STTI Clifford circuits. All 1+1d circuits of this type that we have encountered have only two cusps (as in dual-unitary circuits~\cite{Piroli2019,Zhou2020}) or three cusps (as in DU2 or tri-unitary circuits~\cite{Foligno2023,Liu2023,Rampp2023}) in the line tension {as a function of orientation,} even when a mapping to an explicitly dual- or tri-unitary circuit is not known. Can we always construct such a mapping, i.e. do our STTI Clifford circuits always possess 
multi-unitarity in the sense of discrete light sectors and flow directions? A precise information-theoretic definition of these flow directions 
eludes us, although we offer various interpretations in the Supplement. To this end, it might also be interesting to construct STTI constructions with higher point-group symmetries that admit more than three cusps in the line tension.

While the existence of these light sectors and the corresponding piecewise linearity of the tension is a rather special phenomenon---the domain of multi-unitary models, such as the STTI Clifford circuits considered in this paper---the broad interpretation of information flowing across the entanglement membrane applies to a far more general setting: the volume-law phase of any geometrically local dynamics. A quantum channel that can take a pure state to a mixed state, such as decoherence, generally destroys the volume-law entangled phase, so we might restrict to dynamics that takes pure states to pure states. An exception is if the decoherence acts on a circuit that implements a good quantum error-correcting code. However, error correction is successful precisely when the decoder can recover a pure state despite noise, in which case the error channel followed by recovery is proportional to a projector onto the code space~\cite{nielsen2010}. Thus, the general setting is any quantum circuit (following Trotterization, if the original dynamics is continuous) that can be decomposed into unitary gates and POVMs, possibly with feedback and feeding in fresh qudits.

When the circuit is a stabilizer circuit, we have argued that the membrane at any orientation is at zero temperature, not just in the STTI circuits that are generalized dual- and tri-unitary, but in any (volume-law-entangling) Clifford circuit, with or without randomness. In the 1+1-dimensional random case, the corresponding statistical mechanics model is a $T=0$ DPRM with integer energies. An open question is whether this has any distinctive features from the generic $T>0$ DPRM that appears in the volume-law phase of Haar-random circuits~\cite{Li2023}. Since both the temperature and the details of the energy distributions are irrelevant at the large-system fixed point of the DPRM~\cite{Huse1985,Kardar1985}, there is no obvious distinguishing behavior {that we are aware of.}  We leave a more thorough investigation of this question to future work.  

Generalizing beyond Cliffords, the precise relationships between the existence of a permanently stable plateau, unitarity within the subspace defined by the plateau, and locality-preserving dynamics within the subspace remains open.  All three properties are necessary to define an entire light sector within which the plateau entropy density is unchanged, resulting in a constant membrane tension. Relating the first two properties, Ref.~\cite{Gullans2020} proves that a stable plateau in monitored dynamics, with trajectories sampled according to the Born rule, allows for perfect recovery. However, running a circuit at an arbitrary angle results in post-selected dynamics, e.g. projectors rather than measurements~\cite{Ippoliti2021}. In that setting, the states are not normalized and the plateau dynamics is nonunitary if the normalization depends on the input state.  A simple proof shows that the necessary condition that the normalization $\mathrm{Tr}(K \rho K^\dag)$, {where $K$ is the Kraus operator for one time step,  be independent of the input state for states within the plateau is also a sufficient condition for unitarity within a subspace~\cite{supp-ref}}.  
But we do not know of explicit non-Clifford examples with this property beyond DU2 circuits. Such examples would be useful as they retain some features of stabilizer circuits while generalizing beyond them.

{The role of the flatness of the plateau spectrum also merits clarification. Uniformity of the spectrum appears to be an independent property from unitarity, but again, the only class (ii) dynamics known to us with unitarity and locality preservation within a subspace also have a flat spectrum {when the initial state is maximally mixed}~\cite{Foligno2023,Rampp2023}. It would be interesting to find, or disprove the existence of, models where the membrane tension depends on Renyi index but retains light sectors for every Renyi index.} 

\begin{acknowledgments}
We thank Sagar Vijay for helpful discussions.  This research was supported in part by NSF QLCI grant OMA-2120757 and an Institute for Robust Quantum Simulation (RQS) seed grant. Numerical work was completed using computational resources managed and supported by Princeton Research Computing, a consortium of groups including the Princeton Institute for Computational Science and Engineering (PICSciE) and the Office of Information Technology's High Performance Computing Center and Visualization Laboratory at Princeton University. The open-source \texttt{QuantumClifford.jl} package was used to simulate Clifford circuits~\cite{QuantumClifford}. 
\end{acknowledgments}

\comment{
\begin{figure*}[hbtp]
    \centering
    \begin{tabular}{cc}
    \adjustbox{valign=b}{\subfloat[]{%
          \includegraphics[width=0.45\linewidth]{timelike-sketch.pdf}}}
    &      
    \adjustbox{valign=b}{\begin{tabular}{@{}c@{}}
    \subfloat[]{%
          \includegraphics[width=0.35\linewidth]{spacelike1.pdf}} \\
    \subfloat[]{%

\includegraphics[width=0.35\linewidth]{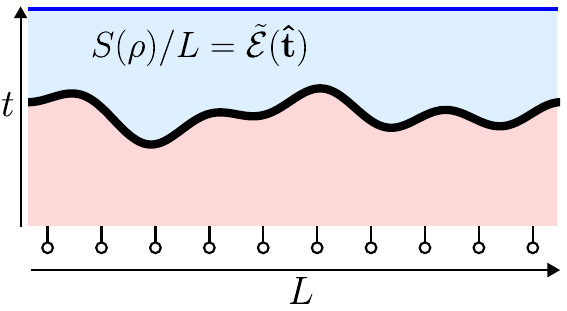}}
    \end{tabular}}
    \end{tabular}
\caption{\gs{ALTERNATE TO FIG. 1: Probes of line tension. In each panel, the time direction is chosen to be vertical, the thick black curve shows a possible path of the membrane, and the blue line on the top boundary indicates the subsystem whose entropy is measured at time $t$. (a) Protocol to determine $\mathcal{E}(v)$ by tracking subsystem entropies $S(x=vt,t)$ starting from two separately scrambled systems at $t=0$. This also illustrates one way to prepare the $t=0$ state: evolve a product initial state (closed circles on bottom boundary) under the chosen dynamics, but omit gates that cross the red dashed cut for $t<0$. Reintroducing these gates at $t=0$ terminates the red cut, making it the preferred endpoint for entanglement membranes within the butterfly cone at $t>0$. (b) Probe of spacelike membranes by evolving a fully mixed initial state (open circles on bottom boundary) under the chosen dynamics, measuring all qubits in the interval $[0,x_1]$ at time $t_1$ and the interval $[x_2,L]$ at time $t_2$, and computing the entropy of the full state at time $t>t_1,t_2$. (c) For an arbitary "time" direction $\vec{\hat{t}}$, evolve a fully mixed initial state on $L$ qubits under $t$ time steps. If the final entropy density converges to a constant value in the limit $L\rightarrow\infty$, $t\rightarrow\infty$, this density is the free energy density of a membrane with normal vector $\vec{\hat{t}}$.}
} 
\end{figure*}
}

%

\onecolumngrid
\setcounter{figure}{0}
\setcounter{section}{0}
\setcounter{page}{0}
\let\oldthefigure\thefigure
\renewcommand{\thefigure}{S\oldthefigure}
\setcounter{equation}{0}
\renewcommand{\theequation}{S\arabic{equation}}
\newpage

\title{Supplemental Information: Zero-temperature entanglement membranes in quantum circuits}
\maketitle
\onecolumngrid
In the supplemental information, we elaborate upon the concrete examples presented in the main text and report on several numerical experiments/surveys of Clifford circuits. 
\begin{itemize}
\item~\autoref{sect:channel} discusses the interpretation of light sectors in terms of "flow directions."
\item~\autoref{app:ZX} introduces the ZX calculus which is used to represent the circuits in the subsequent three sections.
\begin{itemize}
\item~\autoref{app:DU2} discusses the 1+1d Clifford circuits composed of CNOT cores, including examples which are tri-unitary on the honeycomb lattice (\autoref{fig:page}) and sheared dual-unitary.
\item~\autoref{app:dual} discusses the "sideways dynamics" of DU2 circuits, including the hybrid circuit composed of dual-unitary gates presented in the main text (\autoref{fig:dual-hybrid}).
\item~\autoref{app:2d} provides more detail on the 2+1d models introduced in~\autoref{sect:2d} of the main text and their relation to prior works.
\end{itemize}
\item Stepping away from explicitly multi-unitary circuits,~\autoref{app:hybrid} is a broad survey of hybrid STTI Clifford circuits, with attention to the aspects that do and do not conform to the membrane picture. \gs{We also discuss the emergent plateau in hybrid random Clifford circuits.}
\item~\autoref{app:probes} demonstrates techniques for probing the entanglement membranes using ancillas.
\item~\autoref{app:unbinding} offers further evidence of the zero-temperature nature of the domain walls in the good-scrambling CNOT-NOTC circuit, which undergoes a first-order unbinding transition when subject to quasiperiodic edge dissipation, and contrasts it with the transition in a random Clifford circuit.
\item~\autoref{sect:haar} analyzes the entanglement growth in STTI stabilizer circuits, acting upon non-stabilizer initial states. In particular, we verify that for the STTI CNOT-core circuit considered in the main text, starting from the product of two Haar-random states, the state-averaged purity grows in the same piecewise-linear fashion as~\autoref{fig:page}, and the variance decays exponentially with system size.
\end{itemize}

\section{Light sectors and flow directions}\label{sect:channel}

~\autoref{eq:normal-vector}  for the membrane tension \textit{within} a light sector can be interpreted as all {of the entanglement entropy ``flowing''} along the direction $\vec{\hat{t}}$, which is the direction of maximal plateau entropy density.~\gs{\autoref{eq:entanglement-channels} provides a different description, in terms of "flow directions" $\vec{\hat{u}}_i$. In the latter interpretation, entanglement flows independently along each flow direction, and the magnitude of the flow normal to the membrane contributes to $\ef(\vec{\hat{n}})$.} 
The boundaries between light sectors are then the set of orientations
\begin{equation}\label{eq:channel-boundaries}
    \{\vec{\hat{n}} \, | \, \exists \, i \, : \, \vec{\hat{n}} \cdot \vec{\hat{u}}_i = 0\}
\end{equation}
{where at least one of these flow directions lies within the membrane.}

Not only are the $\vec{\hat{u}}_i$ vectors for entropy flow, but in many previously studied examples, two-point correlation functions of one-site operators can be nonzero only along these directions, and are expressible in terms of single-qubit channels~\cite{Bertini2019,Jonay2021,Milbradt2023}. Such channels lie along the locus of points that are timelike or lightlike with respect to all time directions, since correlations must vanish along spacelike directions. In other words, when a model possesses unitary arrows of time, we can associate a light cone to each unitary arrow, and typically find that the points common to all light cones fall along a discrete set of lines.\footnote{Of all the fully multi-unitary models we have considered, the only circuit that lacks these flow directions is the square lattice parity model (see \autoref{sect:slp} of this Supplement).} In non-ergodic models, some correlations persist or oscillate to infinite time owing to conserved charges ("gliders") that travel along the channels~\cite{Bertini2020soliton}. In contrast, all correlations decay exponentially in time in ergodic models, or in maximally chaotic circuits~\cite{Aravinda2021}, vanish for all $t>0$. Clifford circuits are a pathological case where the correlator can only take the values $(-1,0,1)$, so in good scrambling gDU Clifford circuits, the correlators vanish at finite time. Even in such instances, the presence of flow directions can still manifest in the OTOC, as the left and right edges of initially local operators move at the velocity of one of the flow directions.

Certain information-theoretic probes, such as Hayden-Preskill recovery along the light cone~\cite{Rampp2023a,Rampp2023}, are only sensitive to flow directions at the maximal butterfly velocity $\pm v_B=\pm 1$. On the other hand, related probes, such as the growth of a single logical operator in an otherwise pure state, are only sensitive to the entanglement velocity $v_E$. At present, we do not know of a concrete information-theoretic interpretation of flow directions that are strictly timelike with respect to certain unitary arrows of time, e.g., the $v=0$ flow direction in tri-unitary circuits. However, in some ergodic tri-unitary Clifford models, including the Clifford East model~\cite{Gopalakrishnan2018}, the $v=0$ flow direction manifests as certain operators spreading only to the left or to the right, with one endpoint sticking to $x=0$.

In all models where we have been able to identify flow directions in the sense described above,~\autoref{eq:entanglement-channels} is consistent with~\autoref{eq:normal-vector} within each light sector, i.e., the light sectors are indeed bounded by the orientations in~\autoref{eq:channel-boundaries}. We address two 2+1d models, the BCC spider model and ternary-unitary models, in~\autoref{app:2d} of this Supplement, and here discuss the simpler 1+1d examples.

\gs{In dual-unitary circuits, the flow directions are along the light cone~\cite{Bertini2019} ($\theta=\pm \pi/4$ where $v=\cot\theta$), so that:
\begin{equation}
\ef(\theta) = (|\cos(\theta-\pi/4)| + |\cos(\theta + \pi/4)|)/\sqrt{2} = (|\cos\theta + \sin\theta| + |\cos\theta - \sin\theta|)/2,
\end{equation}
with the normalization chosen to match~\autoref{eq:dual1}. In this case, the boundaries of the light sectors are also the edges of the light cones, since the two flow directions are orthogonal.} 

\gs{In tri-unitary circuits with (emergent) hexagonal symmetry, the flow directions are along $\pm \pi/6, \pi/2$~\cite{Jonay2021}, so that:
\begin{equation}
\ef(\theta) = (|\cos(\theta+\pi/3)| + |\cos\theta| + |\cos(\theta - \pi/3)|)/2,
\end{equation}
in agreement with~\autoref{eq:tri1}. The boundaries of the light sectors are where $\vec{\hat{n}}$ points along $\theta = \pm \pi/3, 0$, as shown in~\autoref{fig:tri}.}

{We note in passing that even in non-multi-unitary models which do not admit light sectors in the simple form of~\autoref{eq:normal-vector}, if the butterfly velocity is maximal, that should manifest as a cusp in the line tension at $v=\pm 1$, associated with the finite density of information transported along the light cone~\cite{Claeys2020,Rampp2023a}.}

\section{Notation and ZX Calculus}\label{app:ZX}
A helpful graphical tool for decomposing the dynamics of our circuits is the ZX calculus. For a thorough introduction we refer the reader to Ref.~\cite{VandeWetering2020}. Here we just present the minimal amount of required notation.

ZX diagrams are a tensor network representation of quantum circuits in terms of X and Z spiders (\autoref{eq:z-spider} and~\autoref{eq:x-spider}),
along with the Hadamard gate (yellow box), which can be (nonuniquely) expressed in terms of spiders as
\begin{equation}\label{eq:hadamard}
\input{had.tikz}=\begin{tikzpicture}
	\begin{pgfonlayer}{nodelayer}
		\node [style=hadamard] (0) at (0, 0) {};
		\node [style=none] (7) at (0, 1.75) {};
		\node [style=none] (8) at (0, -1.75) {};
	\end{pgfonlayer}
	\begin{pgfonlayer}{edgelayer}
		\draw (7.center) to (8.center);
	\end{pgfonlayer}
\end{tikzpicture}
.
\end{equation}
In all of our diagrams, we use the convention that time flows upward. When no angle is indicated, $\alpha=0$.

Specializing the definition of Z spiders (\autoref{eq:z-spider}) to the case with $m=0$ input legs and 1 output leg, we obtain an X eigenstate:
\begin{equation}\label{eq:x-state}
\begin{tikzpicture}
	\begin{pgfonlayer}{nodelayer}
		\node [style=z spider] (0) at (0, -0.75) {};
		\node [style=none] (1) at (0, 0.5) {};
	\end{pgfonlayer}
	\begin{pgfonlayer}{edgelayer}
		\draw (1.center) to (0);
	\end{pgfonlayer}
\end{tikzpicture}
 \simeq \ket{+}
\end{equation}
where we will use the notation $\simeq$ to indicate "equal up to a scalar [in this case $\sqrt{2}$] and possibly a Pauli". Similarly, for X spiders:
\begin{equation}
\begin{tikzpicture}
	\begin{pgfonlayer}{nodelayer}
		\node [style=x spider] (0) at (0, -0.75) {};
		\node [style=none] (1) at (0, 0.5) {};
	\end{pgfonlayer}
	\begin{pgfonlayer}{edgelayer}
		\draw (1.center) to (0);
	\end{pgfonlayer}
\end{tikzpicture}
 \simeq \ket{0}.
\end{equation}

In terms of spiders and Hadamards, the CNOT and CZ gates take simple forms:
\begin{equation}
\cnot = \input{cnot.tikz}, \quad \mathrm{CZ} = \input{cz.tikz}.
\end{equation}
(Since the Hadamard gate is symmetric, it remains well-defined when placed on a horizontal leg.)

CNOT + arbitrary single-qubit rotations, which can be expressed in terms of 2-legged X and Z spiders, implement universal quantum computation. Restricting to Clifford circuits, any Clifford linear map~\cite{VandeWetering2020} can be written as a ZX diagram where the phase on each spider is a multiple of $\pi/2$. Up to multiplication by Paulis, there are six unique single-qubit Clifford gates:
\begin{equation}\label{eq:single-qubit}
\begin{tikzpicture}
	\begin{pgfonlayer}{nodelayer}
		\node [style=none] (7) at (0, 1.75) {};
		\node [style=none] (8) at (0, -1.75) {};
	\end{pgfonlayer}
	\begin{pgfonlayer}{edgelayer}
		\draw (7.center) to (8.center);
	\end{pgfonlayer}
\end{tikzpicture}
 \, , \quad \begin{tikzpicture}
	\begin{pgfonlayer}{nodelayer}
		\node [style=none] (7) at (0, 1.75) {};
		\node [style=none] (8) at (0, -1.75) {};
		\node [style=z spider] (9) at (0, 0) {};
		\node [style=none] (11) at (0, 0) {$\scriptstyle{\frac{\pi}{2}}$};
	\end{pgfonlayer}
	\begin{pgfonlayer}{edgelayer}
		\draw [in=90, out=-90] (7.center) to (8.center);
	\end{pgfonlayer}
\end{tikzpicture}
\, , \quad \input{rx.tikz}\, , \quad \input{rplus.tikz}\, ,\quad \input{rminus.tikz}\, ,\quad \begin{tikzpicture}
	\begin{pgfonlayer}{nodelayer}
		\node [style=hadamard] (0) at (0, 0) {};
		\node [style=none] (7) at (0, 1.75) {};
		\node [style=none] (8) at (0, -1.75) {};
	\end{pgfonlayer}
	\begin{pgfonlayer}{edgelayer}
		\draw (7.center) to (8.center);
	\end{pgfonlayer}
\end{tikzpicture}
 \, .
\end{equation}

A central tenet of ZX calculus is that "only connectivity matters" (wires can be bent, stretched, etc., as long as we preserve which legs are inputs and which are outputs). A specific instance of this general rule is that spiders are symmetric tensors, so we can swap any pair of inputs or any pair of outputs without changing the result~\cite{VandeWetering2020}.

Another property that we will exploit is the ability to commute Hadamards through spiders at the cost of changing their color:
\begin{equation}\label{eq:hadamard-rule}
\input{hadamard-rule.tikz}=\input{red-spider.tikz}
\end{equation}
which has the consequence that any ZX-diagram equation still holds if all the X and Z spiders are exchanged (the "color inverse").

Finally, we will use the "spider fusion rules":
\begin{equation}\label{eq:fusion}
\input{fusion.tikz}
\end{equation}
and cancel pairs of connections between opposite spiders:
\begin{equation}\label{eq:pairs}
\input{pair.tikz}.
\end{equation}
\section{Explicit DU2/Tri-unitary}\label{app:DU2}
In this section and the next we demonstrate the explicit connection to generalized dual-unitarity for three examples: (1) good scrambling CNOT-NOTC circuits, (2) reflection-asymmetric CNOT-core circuits, (3) hybrid circuits composed on dual-unitary gates. These circuits all go beyond the one-gate-per-unit-cell constructions of prior works.
\subsection{CNOT-core circuits: Tri-unitarity conditions}
The ability to split CNOT-core gates into pairs of 3-junctions motivates the representation of CNOT-core circuits on the honeycomb lattice, as shown in~\autoref{fig:honeycomb-cnot}.

\begin{figure}
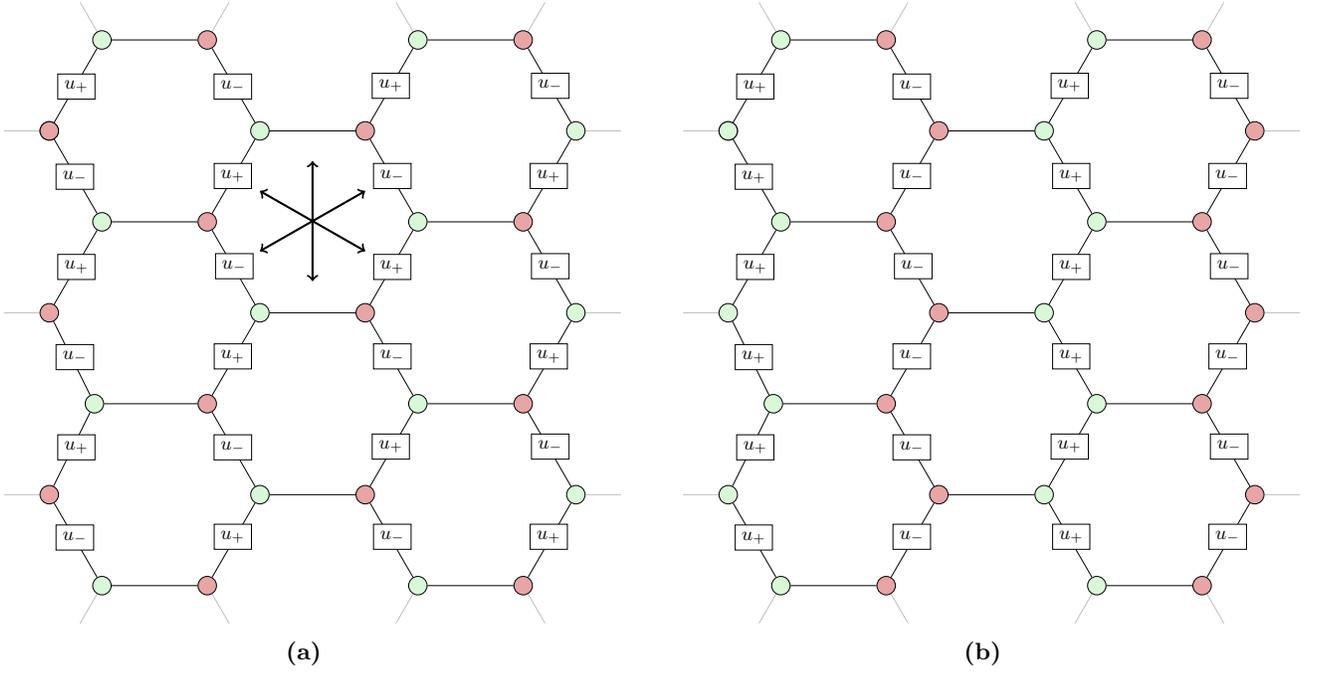

\subfloat[]{
\scalebox{0.8}{\input{cnot-lattice.tikz}}\label{fig:cnot}}
\hfill
\subfloat[]{
\scalebox{0.8}{\input{cnot-notc.tikz}}\label{fig:cnot-notc}}
\caption{(a) CNOT-core circuit ($a=2, b=1,d=1$). (b) Reflection-symmetric CNOT-NOTC-core circuit ($a=2,b=2,d=0$).\label{fig:honeycomb-cnot} Both circuits are networks of red (X) and green (Z) "spiders" interspersed with the single-qubit gates $u_\pm$ (white squares).}
\end{figure}

In this representation, we can then determine the condition for unitarity along six different directions, indicated in~\autoref{fig:cnot}. (These come in pairs related by time reversal, so it suffices to consider just three.) 

Note that
\begin{equation}\label{eq:cz-core-h}
\input{cz-core.tikz} = \sqrt{2} \begin{pmatrix}
u_{00} &&& \\
& u_{01} && \\
&& u_{10} & \\
&&& u_{11}
\end{pmatrix}
\end{equation}
(Here, since $u$ might not equal its transpose, we have distorted to indicate which way time flows; this is a nonstandard diagram.) This two-qubit gate is therefore unitary (up to an overall normalization) if
\begin{equation}\label{eq:unitary-cz}
|u_{ij}| = \frac{1}{\sqrt{2}} \, \forall i,j.
\end{equation}
\gs{i.e., $u$ is a complex Hadamard matrix (again up to normalization), a unitary matrix in which every element has modulus $1$~\cite{Tadej2006,Banica2023}}.\gs{\footnote{The appearance of a complex Hadamard matrix sandwiched between delta tensors (green spiders) in~\autoref{eq:cz-core-h} is reminiscent of the dual Hadamard lattice construction for dual-unitary circuits in Ref.~\cite{Claeys2024}.}.}

First let us consider a CNOT-core circuit in which each gate is identical (\autoref{fig:cnot}), i.e. unit cell $(b,a,d)=(1,2,1)$. The original time direction, which points upward, is unitary by design. Running time along the left lightlike direction, the two-qubit gate is now 
\begin{equation}
\input{cnot-plus.tikz}
\end{equation}
while along the right lightlike direction,
\begin{equation}
\input{cnot-minus.tikz}.
\end{equation}
In both of these equations we have used the color change rule of the Hadamards (\autoref{eq:hadamard-rule}).

Thus, unitarity along these two directions are independent conditions, which via~\autoref{eq:unitary-cz} can be expressed as
\begin{equation}\label{eq:unitary-cnot}
|(u_+ H)_{ij}| = \frac{1}{\sqrt{2}} \textrm{   [left edge of light cone]}, \quad |(H u_-)_{ij}| = \frac{1}{\sqrt{2}} \textrm{    [right edge of light cone]}.
\end{equation}
These conditions are equivalent to those previously derived in Ref.~\cite{Yu2023}. Within the set of Clifford gates, some choices of $u_+, u_-$ satisfy these conditions and others do not. Thus, with Clifford gates alone we can obtain one, two, or three unitary directions. The case of two unitary directions is discussed in more detail in~\autoref{app:cnot-asymm} below. A notable example with three unitary directions is the Clifford East model~\cite{Gopalakrishnan2018,Bertini2023}, where $u_+=u_-=\mathbbm{1}$. This is manifestly tri-unitary~\cite{Yu2023} and in fact self-tri-unitary, since the gates along the three light cone directions are all the same. 

The honeycomb lattice representation also makes it straightforward to generalize beyond the case where each gate is identical. As a minimal demonstration, consider the reflection-symmetric circuits with unit cell $(b,a,d)=(2,2,0)$, composed of alternating CNOT and NOTC cores, as defined in the main text (\autoref{sect:cnot}). Then, as shown in~\autoref{fig:cnot-notc}, the "left light cone direction" is comprised of the gates
\begin{equation}\label{eq:uplus}
\input{left1.tikz} \quad \textrm{and} \quad \input{left2.tikz} = \input{left3.tikz}
\end{equation}

The conditions for unitarity are then
\begin{equation}\label{eq:unitary-cnot-notc}
|(u_+)_{ij}|=\frac{1}{\sqrt{2}}, \quad |(H u_- H)_{ij}| = \frac{1}{\sqrt{2}} \, \forall i,j.
\end{equation}
\comment{
\begin{subequations}\label{eq:unitary-cnot-notc}
\begin{align}
\label{eq:uplus-unitary}
|(u_+)_{00}| &= |(u_+)_{10}| = |(u_+)_{01}| = |(u_+)_{11}| = \frac{1}{\sqrt{2}} \\
\label{eq:uminus-unitary}
|(u_-)_{++}| &= |(u_-)_{-+}| = |(u_-)_{+-}| = |(u_-)_{--}| = \frac{1}{\sqrt{2}}.
\end{align}
\end{subequations}}
If \autoref{eq:unitary-cnot-notc} is satisfied, then tri-unitarity immediately follows, since the evolution along the right light cone is just the mirror image. The circuit in the main text with $u_+$ and $u_-$ given by~\autoref{eq:idx9} is one such example. Graphically, the single-qubit gates are
\begin{equation}
u_+ = \input{rx.tikz}, \quad u_- = \input{rminus.tikz}.
\end{equation}
With time oriented along the left edge of the light cone, this becomes the unitary circuit with gates
\begin{equation}
U_1 = \cnot \, (R_X[\pi/2] \otimes R_{(1,1,1)}[2\pi/3]), \quad U_2 = \mathrm{NOTC} \, (R{(1,1,1)}[-2\pi/3] \otimes H).
\end{equation}

More generally, we can construct \textit{random} circuits with exactly 2 or 3 unitary directions in the following manner: starting from a brickwork circuit, let each two-qubit gate be CZ dressed with single-qubit gates on each input. Randomly choose the one-site gates on the left inputs from the ensemble satisfying (violating)~\autoref{eq:unitary-cz} to make the time evolution along the left light cone unitary (nonunitary), and similarly for the right inputs. Each ensemble contains non-Clifford gates, but when restricted to Cliffords, 4 of the 6 unique one-site gates (\autoref{eq:single-qubit}) satisfy~\autoref{eq:unitary-cz}.

\subsection{Consequences of unitarity on the honeycomb}
In general, unitarity along the original time direction implies

\begin{equation}\label{eq:unitary}
\mathcal{E}(v) = s_{eq}|v|, \quad |v|\geq v_{LC} \quad \textrm{[unitarity].}
\end{equation}
General constraints on the coarse-grained line tension also impose $\mathcal{E}(v_B^R) = s_{eq} v_B^R$ and $\mathcal{E}(v_B^L) = -s_{eq} v_B^L$~\cite{Jonay2018}, so that~\autoref{eq:unitary} in fact holds through the entire interval $(-\infty,v_B^L] \cup [v_B^R, \infty)$. 

On the honeycomb lattice, $v_{LC}=\sqrt{3}$, and $s_{eq}=2/(3a)$ in units of lattice spacing $a$. Unitarity along the right light cone translates into the identity~\cite{Yu2023,Liu2023}
\begin{equation}\label{eq:unitary-right}
\scalebox{0.5}{\input{right-unitary.tikz}} = \scalebox{0.5}{\input{right-unitary2.tikz}}
\end{equation}
where the color purple indicates the folded gate $U \otimes U^*$, and
\begin{equation}\label{eq:identity}
\scalebox{0.6}{\input{trace.tikz}}.
\end{equation}
~\autoref{eq:unitary-right} allows for analytic computation of the line tension in the interval $-v_{LC}<v<0$:
\begin{equation}\label{eq:right-unitary}
\mathcal{E}(v) = s_{eq} \frac{\sqrt{3}-v}{2}, \quad -\sqrt{3} < v < 0 \quad \textrm{[unitarity along right light cone].}
\end{equation}
And analogously, unitarity along the left light cone, which translates into
\begin{equation}\label{eq:left-unitary}
\scalebox{0.5}{\input{left-unitary.tikz}} = \scalebox{0.5}{\input{left-unitary2.tikz}}
\end{equation}
allows for analytic computation of the line tension in the interval $0<v<v_{LC}$:
\begin{equation}
\mathcal{E}(v) = s_{eq}\frac{\sqrt{3}+v}{2}, \quad 0 < v < \sqrt{3} \quad \textrm{[unitarity along left light cone].}
\end{equation}

Thus, for our representative circuit in the main text [\autoref{eq:idx9}], we recover the tri-unitary line tension reported in~\autoref{eq:tri-tension} by choosing lattice spacing $a=2/3$. The same line tension is obtained for any tri-unitary CNOT (\autoref{eq:unitary-cnot}) or CNOT-NOTC circuit (\autoref{eq:unitary-cnot-notc}). However, some of these circuits are non-ergodic, possessing operators that move along the flow directions $v=0,v=\pm \sqrt{3}$ without spreading ("gliders")~\cite{Jonay2021}. These non-ergodic circuits fail to generate maximal entanglement from an initial pure product state, which can be diagnosed using the protocol in~\autoref{fig:page-sketch}. 

Finally, a note on units: with the symmetry of the honeycomb it is natural to say that the flow directions have velocities $v=0, v=\pm \sqrt{3}$. But the spacing between qubits on an equal-time slice is $\sqrt{3}$ times the spacing between consecutive layers in time. Thus if we measure time in number of layers, and space in number of qubits, the nonzero-velocity flow directions are at $v=\pm 1$. With the latter convention, which also yields $s_{eq}=1$, we recover the line tension derived for DU2 circuits in Ref.~\cite{Foligno2023,Rampp2023}. 

\subsection{Reflection-asymmetric CNOT-core circuit}\label{app:cnot-asymm}
The left-right reflection symmetry of the $(b,a,d)=(2,2,0)$ CNOT-NOTC circuits in the main text means that these circuits contain precisely one or three unitary directions. Without this symmetry, not only can the timelike domain wall energies become asymmetric, but $\mathcal{E}(v)$ can have a minimum at nonzero $v$, i.e., the lowest energy interface has nonzero slope.

A simple setting in which this strong asymmetry manifests is the CNOT-core $(b,a,d)=(1,2,1)$ circuits, particularly those where \textit{two} of the three directions on the honeycomb are unitary. A representative circuit with this behavior has the two-qubit gate
\begin{equation}\label{eq:cnot-asymm}
U = \cnot (R_{(1,1,1)}[-2\pi/3]\otimes R_{(1,1,1)}[-2\pi/3])
\end{equation} 
which is unitary along the original time direction and right light cone direction, but nonunitary along the left light cone direction. In this section we will show how these two directions of unitarity manifest as a piecewise-linear line tension that is "sheared dual-unitary."

The two directions of unitarity (\autoref{eq:unitary} and~\autoref{eq:right-unitary}) give the piecewise form for $\mathcal{E}(v)$ outside the interval $v\in[0,v_B^{R}]$. Using the brickwork circuit units with $v_{LC}=1$ yields:

\begin{equation}
\mathcal{E}(v) = \begin{cases}
    |v| & v \in (-\infty,-1] \cup [v_B^R,\infty)  \\
    \frac{1-v}{2} & v \in(-1,0)
\end{cases}
\end{equation}
\comment{
\begin{equation}
\mathcal{E}(v) = s_{eq} \begin{cases}
|v| & v \in (-\infty, -\sqrt{3}] \cup [v_B^R,\infty) \\
\frac{\sqrt{3}-v}{2} & v \in [-\sqrt{3}, 0]
\end{cases}
\end{equation}
}
The cusp at $v=-1$ is associated with operators spreading with left endpoints along the left light cone, with velocity $v_B^L = -1$. The Schr{\"o}dinger evolution of Pauli operators in Clifford circuits is particularly simple because Clifford gates normalize the Pauli group: a single Pauli string transforms into a single Pauli string, rather than a superposition of many strings. As shown in~\autoref{fig:asymm-operator}, in the circuit with local gate~\autoref{eq:cnot-asymm} initially local Paulis spread within a region bounded by $v_B^L = -1, v_B^R = 1/3$. This fixes the line tension 
$\mathcal{E}(1/3) = 1/3$. Convexity of $\mathcal{E}(v)$ further fixes the slope in the interval $[0,1/3]$ to be $-1/2$, so in fact $\mathcal{E}(v)$ is continuous within the butterfly cone:
\comment{
\begin{equation}
\mathcal{E}(v) = s_{eq} \begin{cases}
    \frac{\sqrt{3}-v}{2} & v \in [-\sqrt{3},1/\sqrt{3}]  \\
    |v| & \mathrm{otherwise}
\end{cases}
\end{equation}
}
\begin{equation}
\mathcal{E}(v) = \begin{cases}
    \frac{1-v}{2} & v \in [-1,1/3]  \\
    |v| & \mathrm{otherwise}
\end{cases}
\end{equation}
as seen in~\autoref{fig:asymm}. That is, the only cusps in the line tension are along the (asymmetric) butterfly cone directions.

\begin{figure}[t]
\centering
\includegraphics[width=0.8\linewidth]{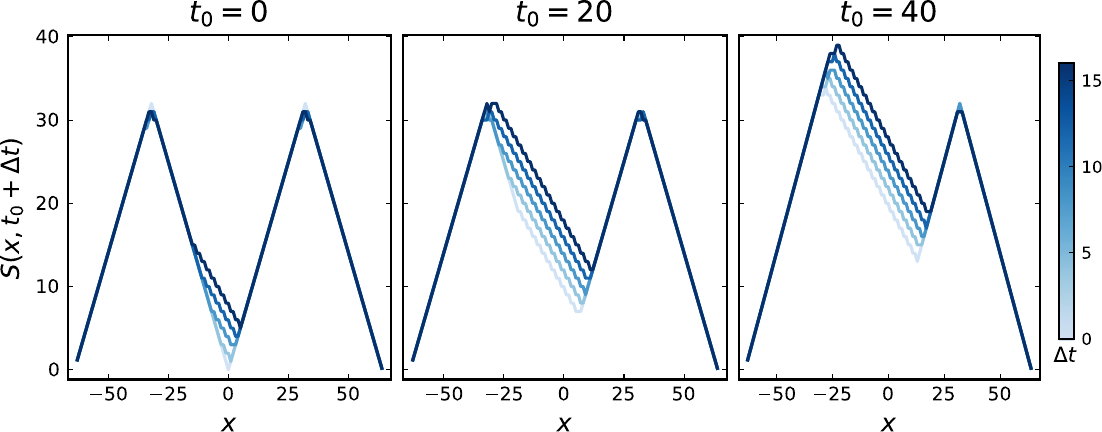}
\caption{Entropy growth in an asymmetric dressed CNOT circuit belonging to the SDKI class, on $L=128$ qubits. Each gate is $(R_{(1,1,1)}[-2\pi/3] \otimes R_{(1,1,1)}[-2\pi/3]) \cnot$. The entanglement after every four layers is plotted, with darker curves correspond to later times. 
\label{fig:asymm}}
\end{figure}

\begin{figure}[hbtp]
\subfloat[]{\includegraphics[width=0.41\linewidth,height=3in,keepaspectratio]{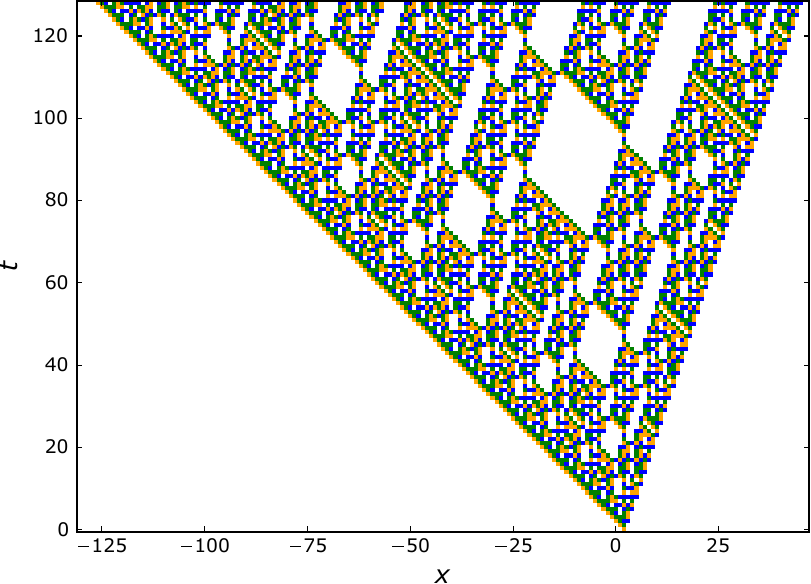}
\label{fig:asymm-operator}}
\subfloat[]{
\includegraphics[width=0.55\linewidth,height=3in,keepaspectratio]{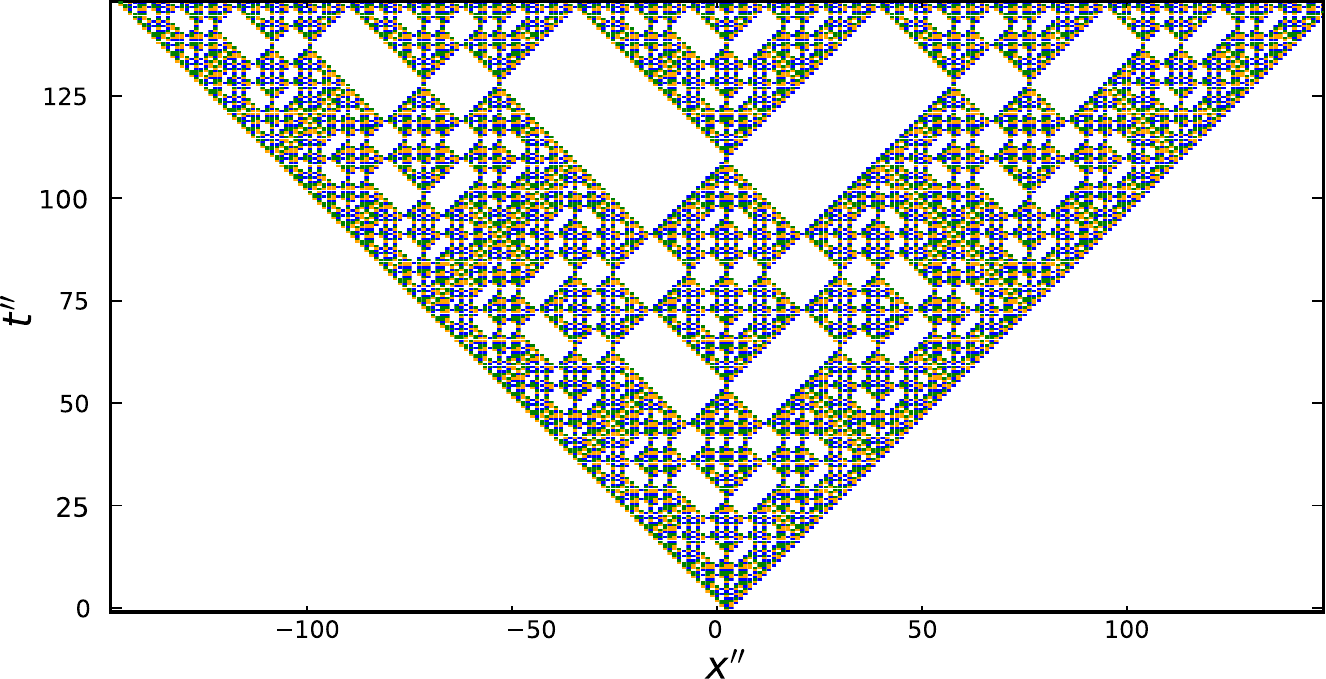}
\label{fig:sheared}}
\caption{Operator spreading in a CNOT-core brickwork circuit where every gate is given by~\autoref{eq:cnot-asymm}. A $Z$ operator is initialized on the right leg entering a gate ($x=2$) and evolved in the Schrodinger picture. (a) shows its spacetime "footprint" in the original coordinates, with a white, blue, green, or orange pixel at the location $(x,t)$ indicating that the time-evolved Pauli string at time $t$ has an $I$, $X$, $Y$, or $Z$, respectively, at position $x$. The same spacetime footprint is shown in (b) in terms of the transformed coordinates $(x'',t'')$ [\autoref{eq:transformed}].}
\end{figure}

It turns out there is a deeper reason for why $\mathcal{E}(v)$ takes such a simple form. While the edges of the spreading operators tell us the butterfly velocities, the bulk tells us the fractal dimension $d_f$, that is, the scaling of the non-identity spacetime volume $V\sim t^{d_f}$ within the image of the operator. For this particular class of circuits, $d_f=\log_2[(3+\sqrt{17})/2]=1.832...$, the same fractal dimension as that of the (dual-unitary) self-dual kicked-Ising (SDKI) circuit~\cite{Gutschow2010fractal,Akila2016,Bertini2018,Gopalakrishnan2019,Sommers2023}. In fact, a shearing + rescaling transformation yields the same operator spreading as in the SDKI model: first shear the lattice as
\begin{equation}
(x,t) \rightarrow (x',t')=(x + t/3, 2t/3)
\end{equation}
so that $|v_B^L| = v_B^R = 1$, and then stretch/shrink along the light cone axes:
\begin{equation}
(t'+x',t'-x') \rightarrow (t''+x'', t''-x'')=((t'+x')/\sqrt{3}, \sqrt{3}(t'-x'))
\end{equation}
so that
\gs{
\begin{equation}\label{eq:transformed}
x = \frac{\sqrt{3}}{2} x'', \quad t = \frac{\sqrt{3}}{2} (x''+2t'').
\end{equation}
}
This yields~\autoref{fig:sheared}. \gs{By~\autoref{eq:transform},} this same transformation turns the asymmetric line tension into the flat line tension characteristic of dual-unitary circuits:
\begin{equation}
\mathcal{E}''(v'') = \frac{\sqrt{3}}{2} \begin{cases} 1 & |v| < 1 \\
|v| & |v| \geq 1
\end{cases}.
\end{equation}

\section{Duals of Tri-Unitary Circuits}\label{app:dual}
In this section we provide more detail on the spatial dynamics of the tri-unitary/DU2 circuits and the relation to the hybrid circuit with dual-unitary gates introduced \autoref{sect:dual-hybrid} of the main text.

\subsection{Running the honeycomb circuit sideways}

\begin{figure}[t]
\subfloat[]{
\includegraphics[width=0.55\linewidth,height=3in,keepaspectratio]{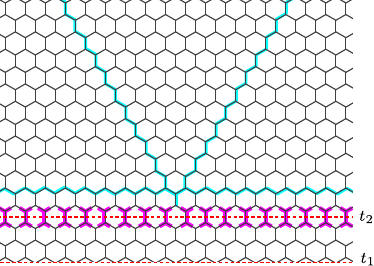}\label{fig:honeycomb-spatial}}
\hfill
\subfloat[]{
\includegraphics[width=0.4\linewidth,height=3in,keepaspectratio]{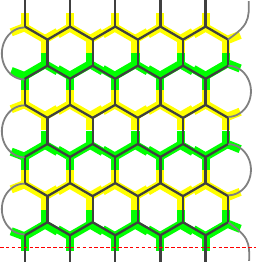}\label{fig:honeycomb-obc}}
\caption{CNOT-core circuit run in the "spatial" direction. (a) Zoom-out view of the honeycomb lattice. Highlighted in blue are the three flow directions running at $v=\pm 1/\sqrt{3}$ and $|v|=\infty$. Highlighted in purple is a row of non-unitary gates obtained by taking the dual of a CNOT-core gate. If we cut time along the red dashed line labeled $t_1$, then the spatial circuit is a brickwork circuit of these purple gates. The dashed line labeled $t_2$ is an alternate way to cut time which reduces the Hilbert space dimension by 1/2. (b) The spatial circuit with OBCs chosen so that, cutting time along $t_2$, the dynamics is explicitly unitary: alternating cascade of yellow gates and green gates.}
\end{figure}

The \dam~of CNOT-core tri-unitary circuits can be interpreted in two ways. From one perspective, since the circuit defined on the honeycomb lattice is tri-unitary, any angle along which we run the membrane will fall outside the light cone of some unitary time arrow, making the \dam~strictly unitary, group (i) of~\autoref{sect:dam}. From a second point of view, in terms of the gDU hierarchy, these circuits belong to the second level (DU2) and their spatial dynamics is only unitary within a subspace. In the language of entanglement membranes, the dimension of this subspace relative to the full Hilbert space is given by the entanglement velocity $v_E=\mathcal{E}(0)$, that is, the entropy density of the plateau group running time in the spatial direction. This entropy density is set by the Schmidt rank $\mathcal{R}$ of $U$~\cite{Rampp2023}, that is, the number of nonzero eigenvalues of
\begin{equation}\label{eq:u-utilde}
\tilde{U}\tilde{U}^\dag = \scalebox{0.5}{\input{Urep.tikz}}
\end{equation}
where $\tilde{U}$ is the dual gate $\tilde{U}^{o_2 i_2}_{o_1 i_1} = U^{o_1 o_2}_{i_1 i_2}$. In DU2 circuits,~\autoref{eq:u-utilde} has a flat spectrum~\cite{Foligno2023}, and the entanglement velocity is:
\begin{equation}
v_E = \frac{\log \mathcal{R}}{\log q^2}.
\end{equation} 
For a generic unitary, $\mathcal{R}$ can range from $1$ (e.g. for the identity gate) to $q^2$ (e.g. for a dual-unitary gate); the DU2 conditions restrict $\mathcal{R}$ to certain values (although the constructions in Ref.~\cite{Rampp2023} are not exhaustive). On qubits, CNOT has $\mathcal{R} = q = 2$, and we recover the entanglement velocity of $1/2$ in brickwork circuit units.

To reconcile the full tri-unitarity on the honeycomb with the unitarity only within a subspace for the spatial dynamics, consider rotating the honeycomb lattice in~\autoref{fig:honeycomb-cnot} by $\pi/2$ to obtain~\autoref{fig:honeycomb-spatial}. 

In the figure, we show two constant time slices on which to cut the legs. If we cut time along the slice labeled $t_1$, then the spatial dynamics is a brickwork circuit of nonunitary gates, one row of which is shown in purple in~\autoref{fig:honeycomb-spatial}. Each nonunitary gate is the dual of a CNOT-core, which can be decomposed into a forced measurement of $ZZ$ (projector onto +1 eigenvalue) followed by a unitary~\cite{Ippoliti2021}. Thus, starting from a fully mixed $L$-qubit state, the first layer of nonunitary gates purifies the state by $L/2$ bits. In fact, the state stops purifying after this first layer, reaching a steady state stabilizer group with entropy density $1/2$. This "plateau group" is the subspace within which the spatial evolution is unitary, with half as many logical qubits as physical qubits.

If we instead slice time along the dashed line labeled $t_2$, however, the Hilbert space only contains $L/2$ physical qubits to begin with. On this Hilbert space, the spatial dynamics can be made explicitly unitary by choosing open boundary conditions as in~\autoref{fig:honeycomb-obc}. The resulting circuit can be expressed as alternating "cascades" of gates along the left and right light cones of the original honeycomb lattice, which are both unitary if the underlying circuit is tri-unitary. 

 To summarize, with the appropriate choice of boundary conditions, the \dam~of DU2 circuits on qubits composed of entangling gates fall into group (i). An open question is whether this is true of all gDU circuits: can we tune the boundary conditions such that for any membrane orientation, the \dam~is strictly unitary, not just unitary within a subspace? If so, then fully multi-unitary circuits of type (i) and (ii) coincide.
 
Highlighted in blue in~\autoref{fig:honeycomb-spatial} are the flow directions of the $\pi/2$-rotated circuit, at $v=\pm 1/\sqrt{3}$ and $v=\pm \infty$---the latter implying an infinite "butterfly velocity" with respect to the rotated arrow of time. These flow directions are precisely those observed exhibited in the entanglement growth of the dual hybrid circuit (\autoref{eq:honeycomb-spatial} and \autoref{fig:dual-hybrid}b). In ~\autoref{app:dual-hybrid} we show where this connection comes from.

\subsection{Boundary conditions for the membrane tension}

Before proceeding, let us note a third way to define "boundary conditions" that aligns more closely with tensor network calculations of the line tension.\footnote{\gs{See also the "depolarizing" or "light cone" boundary conditions of Ref.~\cite{Ippoliti2021}.}} To obtain the $n$th Renyi tension, $\mathcal{E}_n(v)$ of a membrane running from $(0,0)$ to $(x=vt,t)$ in a circuit that is unitary with respect to the standard time direction, one would evaluate the $n$-replica object~\cite{Zhou2020,Foligno2023,Rampp2023}:
\begin{equation}
    e^{-(n-1)\mathcal{E}_n(v) t} = \frac{1}{q^{nt}}\scalebox{0.5}{\input{membrane.tikz}}
\end{equation}
where the orange shading indicates $n$ copies of the folded gate, and the mark in the corner indicates that the standard unitary time direction flows to the right. On the legs terminated with circles, each folded gate is contracted with itself (identity permutation, generalizing~\autoref{eq:identity} to $n$ copies), while on the legs terminated with squares, a cyclic permutation is performed:
\begin{equation}
\scalebox{0.6}{\begin{tikzpicture}
	\begin{pgfonlayer}{nodelayer}
		\node [style=end] (0) at (0, -1.25) {};
		\node [style=none] (1) at (0, 1.25) {};
	\end{pgfonlayer}
	\begin{pgfonlayer}{edgelayer}
		\draw (1.center) to (0);
	\end{pgfonlayer}
\end{tikzpicture}
} = \overbrace{\scalebox{0.6}{\input{identity-perm2.tikz}}}^{2n}, \qquad \scalebox{0.6}{\begin{tikzpicture}
	\begin{pgfonlayer}{nodelayer}
		\node [style=square] (0) at (0, -1.25) {};
		\node [style=none] (1) at (0, 1.25) {};
	\end{pgfonlayer}
	\begin{pgfonlayer}{edgelayer}
		\draw (1.center) to (0);
	\end{pgfonlayer}
\end{tikzpicture}
} = \overbrace{\scalebox{0.6}{\input{swap-perm2.tikz}}}^{2n}.
\end{equation}
 That is, we are feeding a fully mixed state on $t$ qudits into a circuit composed of the dual gates:
 \begin{equation}\label{eq:dual-dynamics}
 \rho(t) = \frac{1}{q^t} \scalebox{0.5}{\input{membrane-v3.tikz}}
 \end{equation}
 then calculating the Renyi entropy of the output state $\rho(t)$:
 \begin{equation}
     \mathcal{E}_n(v) t = - \frac{1}{1-n} \log \mathrm{Tr}(\rho(t)^n).
 \end{equation}

 In the main text, when we define the tension of the membrane by the plateau entropy density of the dynamics running across it, in the limit $L\rightarrow\infty, t\rightarrow\infty$, we are essentially asserting that this plateau entropy matches the output entropy density of~\autoref{eq:dual-dynamics}, which contains $(t-x)(t+x)/4$ gates. For $v=0$, this corresponds to saying that the plateau entropy density of the "sideways" brickwork circuit with the appropriate boundary conditions matches the output entropy density of the $t$-qubit, \gs{$(t-1)$-layer} diamond-shaped circuit in~\autoref{eq:dual-dynamics}. 
 
 \gs{More generally, for $|v|<1$ in DU2 circuits,~\autoref{eq:dual-dynamics} reaches the plateau as soon as the circuit has reached its maximal width, of $t-|x|$ qubits. This can be seen from the factorization of the right influence matrix $\ket{R}$~\cite{Rampp2023}: after $(t-|x|)/2$ layers (bottom dashed line), the state evolves to:}
 \begin{equation}\label{eq:rho-input}
     \rho((t-|x|)/{2}) = \frac{1}{q^t} \scalebox{0.5}{\input{membrane-input.tikz}} = \frac{1}{q^t} \scalebox{0.5}{\input{membrane-input2.tikz}}
 \end{equation}
\gs{which has entropy $|x| + v_E (t - |x|)$, consistent with~\autoref{eq:brickwork-tri}. The spectrum is invariant under the remaining $(t-|x|)/2-1$ layers. In particular, for each layer of width $t-|x|$ in between the two dashed lines, time evolution just implements a shift:}
 \begin{align}
     \rho((t-|x|)/{2}+1) &= \frac{1}{q^t} \scalebox{0.5}{\input{membrane-transfer.tikz}} = \frac{1}{q^t} \scalebox{0.5}{\input{membrane-output.tikz}}.
 \end{align}
\gs{Finally, when evaluating the entropy, the last $(t-|x|)/2$ layers simplify in a complementary fashion to~\autoref{eq:rho-input}, owing to the factorization of the left influence matrix~\cite{Rampp2023}.}

The structure of~\autoref{eq:dual-dynamics} also lends insight into the behavior of the \dam~in a generic (non-Clifford) unitary circuit, when the slope of the membrane is between $v_B$ and $v_{LC}$. 
Setting $s_{eq}=1$, the general constraints of convexity, $\mathcal{E}(v_B) = v_B$, and $\mathcal{E}(v_{LC}) = v_{LC}$ ~\cite{Jonay2018} impose that the coarse-grained line tension of $\mathcal{E}(v) = |v|$, which follows from strict unitarity outside the light cone, extends all the way down to $|v|=v_B$. Comparing to~\autoref{eq:dual-dynamics}, we see that this line tension implies that the $t$-qubit initial mixed state fed in at the bottom reaches a "plateau entropy" of $|x| < t$ bits. {However, if the circuit is non-Clifford, this plateau is not permanently stable due to exponential tails outside the butterfly cone, and if the spectrum evolves continuously, we might expect that as the slope of the membrane is taken through $v_{LC}$, the membrane changes from being at $T=0$, to $T\neq 0$, although the functional form of the line tension is unchanged.}

\subsection{From Dual Hybrid to Tri-unitary}\label{app:dual-hybrid}

Using tools from ZX calculus, we can bring the dual hybrid circuits into a similar form to the CNOT-core circuits of the previous section. First, we write the iSWAP gate in terms of spiders and Hadamards as
\begin{align}
\iswap \simeq \mathrm{SWAP} \, \mathrm{CZ} \, (R_Z[\pi/2]\otimes R_Z[\pi/2]) = \input{iswap.tikz}.
\end{align}
where we have used the fusion rule (\autoref{eq:fusion}) to absorb the Z rotations into the Z spiders of the CZ gate.

In the dual-hybrid circuit described in the main text, each gate is either immediately preceded by a measurement on the left leg, or immediately followed by a measurement on the right outgoing leg. Performing measurements in the $X$ basis, we choose gates of the form:
\begin{equation}
    \iswap (R_X[\pi/2] \otimes u)
\end{equation}
Since $X$ rotations commute with $X$ measurements, we can omit $R_X[\pi/2]$ on layers immediately after measurements. We can also assume that the measurement outcomes are always +1 (i.e., post-select/apply feedback), since the sign of the measurement outcome just affects the signs on the stabilizer generators and not the entanglement properties. 

Using~\autoref{eq:x-state}, iSWAP followed by measurement is:
\begin{equation}
(\mathbbm{1} \otimes \bra{+}) \iswap = \input{iswap-out.tikz}
\end{equation}
and iSWAP preceded by measurement is:
\begin{equation}
\iswap (\ket{+} \otimes \mathbbm{1}) = \input{iswap-in.tikz}
\end{equation}

Next, using the identity~\autoref{eq:hadamard},
we can deform the hybrid brickwork circuit into a circuit on a honeycomb lattice where each hexagon is:
\begin{equation}
\input{hexagon-v2.tikz}
\end{equation}

We can now run the circuit with the spatial direction along any of the edges of the hexagon and ask whether it is unitary. With space along the direction labeled "1", each two-qubit gate is CZ dressed with single-qubit gates, which is unitary. The dynamics along the "2" and "3" directions are both unitary if $u$ satisfies~\autoref{eq:unitary-cz}. 

Two striking examples deserve mention. First, when $u=H$, the dynamics are self-tri-unitary; any of the three directions is a brickwork circuit with the two-qubit gate
\begin{equation}
\input{cz-h.tikz} = \mathrm{CZ} (R_{(1,1,1)}[-2\pi/3] \otimes R_{(1,1,1)}[-2\pi/3]).
\end{equation}
While this circuit is DU2/tri-unitary, it has gliders that suppress entanglement growth. Indeed, it is precisely the $(\pi/2,\pi/2)$ alternating kicked Ising model (AKIM) discussed in Ref.~\cite{Liu2023}, which has a non-unique steady state.

To obtain better scrambling properties, we can instead set $u = R_{(1,1,1)}[-2\pi/3] \simeq \input{rminus.tikz}$. Then we can choose to cut the legs such that the dynamics along the "1" direction have the two-qubit gate
\begin{equation}
\input{direction1.tikz} = \input{cnot.tikz}
\end{equation}
i.e., the evolution in this direction is just that of the Clifford East model. Since the Clifford East model is also self-tri-unitary, we obtain the same circuit along the other five directions.

From the line tension~\autoref{eq:honeycomb-spatial} on the honeycomb, we can also derive the line tension in the original coordinates on the square lattice. \gs{The honeycomb coordinates $(x,t)$ are related to the brickwork circuit coordinates $(x',t')$ through the transformation:
\begin{equation}
(x,t) = \left(\frac{4x'+t'}{2\sqrt{3}}, \frac{t'}{2}\right)
\end{equation}}
where $\Delta x' = 1$ is the spacing between unmeasured qubits and $t'$ is the number of layers as in~\autoref{fig:cell}. Then the line tension~\autoref{eq:honeycomb-spatial} transforms into \gs{[cf.~\autoref{eq:transform}]}:
\begin{equation}
\mathcal{E}'(v') = \begin{cases}
|v'| & v' < -\frac{1}{2} \\
\frac{1}{2} & -\frac{1}{2} \leq v' \leq 0 \\
|v'| + \frac{1}{2} & v' > 0
\end{cases}
\end{equation}
in agreement with the entanglement growth shown in the left panel of~\autoref{fig:dual-hybrid}. It should be noted, though, that to obtain this growth, we took the initial state at $t=0$ to be the product of two random stabilizer states. If we instead prepare the $t=0$ state by running the STTI hybrid circuit omitting gates along the cut, then the amount of entanglement generated by the circuit is significantly reduced. This is another indication of the fragility of the protocol when $v_B=\infty$; see~\autoref{sect:gpp} of this Supplement for a related example.

\section{Details on 2+1d}\label{app:2d}

In this section we elaborate on the two 2+1d STTI Clifford circuits introduced in the main text (\autoref{sect:2d}) and compare the resulting light sector geometries---light cube for the square lattice parity model, light octahedron for the BCC spider model---to the 2+1d ternary-unitary~\cite{Milbradt2023} and tri-unitary~\cite{Jonay2021} circuits previously introduced.

\subsection{Square lattice parity model}\label{sect:slp}
To represent the square lattice parity model as a simple cubic lattice, we use the ZX calculus fusion rule (\autoref{eq:fusion}) to fuse CNOT gates with the same control qubit. Conversely, to analytically derive the surface tension of an entanglement membrane with normal vector $\hat{\vec{n}}$ (\autoref{eq:normal-slp}), we perform the fusion rule in reverse along to break up the lattice into unitary gates along any angle. 

Running time along the $z$ direction, let us restate~\autoref{eq:normal-slp} in terms of the tension $\mathcal{E}_{\vec{r}}(\vec{g})$ defined in~\autoref{sect:tension-2}. Consider a membrane containing the points $(x,y,t(x,y))$, so it has the unnormalized normal vector

\begin{equation}
\vec{n} = \begin{pmatrix}
    \frac{\partial t}{\partial x} & \frac{\partial t}{\partial y} & -1
    \end{pmatrix} = \begin{pmatrix} g_x & g_y & -1 \end{pmatrix}.
\end{equation}
For any choice of membrane orientation, the~\dam~is unitary in the sense of class (i) owing to unitarity along one of the three axes.
For example, when $\mathrm{argmax}_i(|n_i|) = 1$ (so $\vec{n}$ points inside the red light sector), the membrane is spacelike with respect to the unitary evolution along the $x$ direction, i.e. it lies outside the pair of red light cones shown in~\autoref{fig:light-pyramids}. Thus, the entanglement across the membrane is given by the size of the Hilbert space in the $yz$ plane. Similarly, when $\mathrm{argmax}_i(|n_i|) = 2$ (belongs to the blue light sector), the membrane is spacelike with respect to unitary evolution along the $y$ direction, i.e. lies outside the pair of blue light cones; and when $\mathrm{argmax}_i(|n_i|)=3$ (yellow light sector), the membrane is spacelike to the unitary evolution along the $z$ direction, i.e. lies outside the pair of yellow light cones.

Altogether, this implies a membrane tension of: 
\begin{equation}
\mathcal{E}_{\vec{r}}(\vec{g}) = \max_i(|n_i|) = \max(|g_x|, |g_y|, 1),
\end{equation}
which is equivalent to~\autoref{eq:normal-slp} for $\ef(\vec{\hat{n}})$. This is a natural higher-dimensional analog to the dual unitary line tension in 1+1d, where $\mathcal{E}_r(g_x) = \max(1, |dt/dx|) = \max(1, 1/|v|)$ (more commonly written as $\mathcal{E}(v) = \max(1, |v|))$. The orientations of $\vec{n}$ for which $\mathcal{E}_{\vec{r}}$ has a cusp define the boundaries of the light \textit{sectors}. In this model, the light sectors are the six pyramids shown in~\autoref{fig:light-cube}; note that the light sectors are larger than the light cones, a point of contrast with 1+1d dual-unitary circuits.

As a test of our Page curve protocol against the analytic result, we study membranes with ends pinned to the line $x + my = 0$ at $z=0$ and $x + my = vt$ at $z=t$. The cut at $z=0$ is effected by preparing a state on the lattice  $[-L_x/2,L_x/2] \times [-L_y/2,L_y/2]$, initialized as the product of random stabilizer states on the subsystems with $f_m(x,y)=x + my < 0$ and $f_m(x,y)=x+my > 0$. After evolving under $t$ layers with open boundary conditions, we pin to the line $f_m(x,y) = vt$ by computing the entropy of the subsystem to the left of this line. 

\begin{figure}[t]
\subfloat[]{
\includegraphics[width=0.4\linewidth]{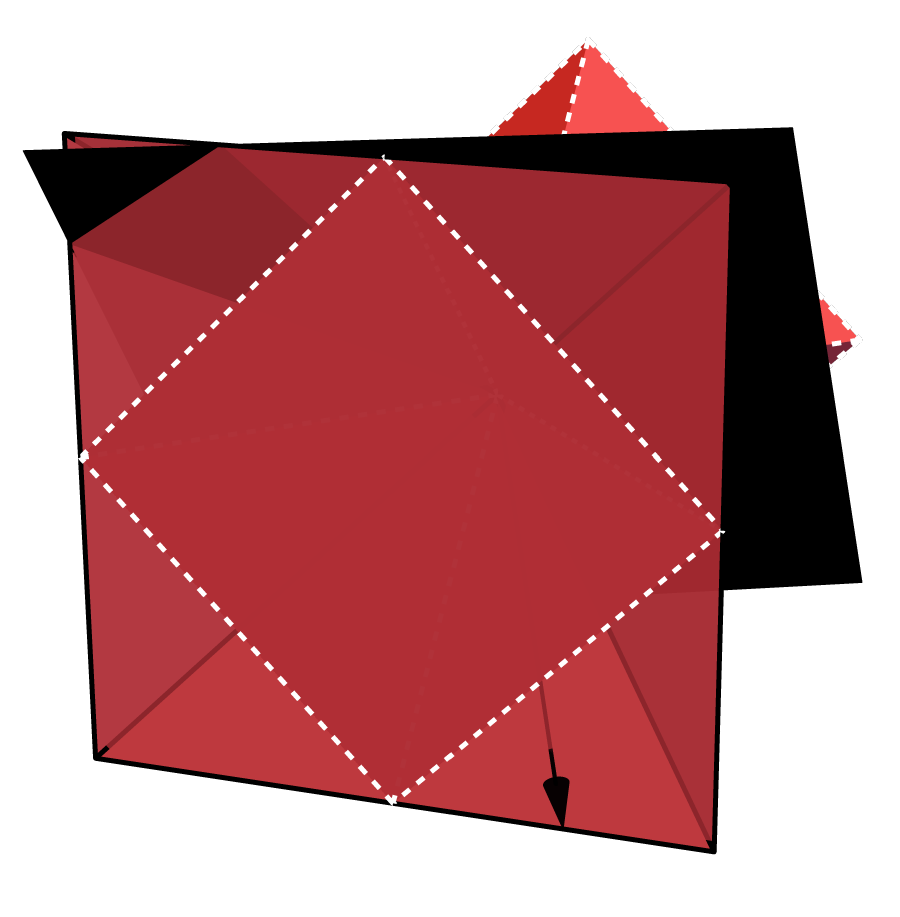}
\label{fig:light-pyramids}}
\subfloat[]{\includegraphics[width=0.5\linewidth]{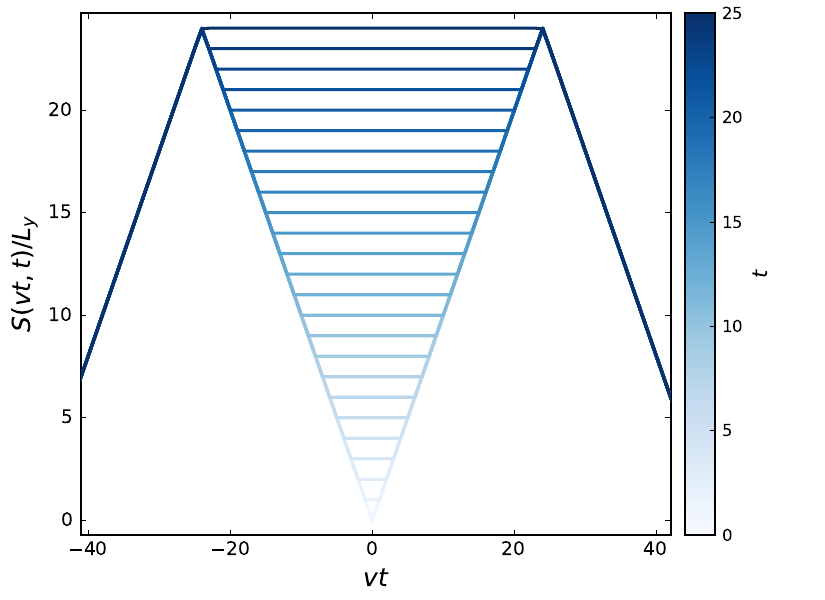}\label{fig:entropy-slp}}
\caption{(a) Portion of the membrane parameterized by $m=0.5,v=0.75$ in the Page curve protocol described in the text, shown in black, with the red light cones of the square lattice parity model outlined in white for reference. The membrane lies outside both red light cones; its normal vector (black arrow) belongs to a red light sector (outlined in black). (b) $S_m(vt,t)/L_y$ for $m=0.5$, up to $t=24$. Time runs in the $z$ direction, and the qubits occupy a $L_x \times L_y$ lattice with $L_x=96$, $L_y=24$ and open boundary conditions.}
\end{figure}

For a given $m$, we take $L_x$ to be large enough compared to $L_y$ that within the light cone ($|v|<1$), there are $L_y$ qubits within the slice $[f_m(x,y) = vt, f_m(x,y) = vt + 1]$. The Page curve protocol gives us access to all "timelike" membranes with $|v| < 1$. An $(m,v)$ membrane has the parameterization
\begin{equation}
    t(x,y) = \frac{x + my}{v} \rightarrow \vec{g} = \left(\frac{1}{v}, \frac{m}{v}\right)
\end{equation}

~\autoref{fig:light-pyramids} shows a patch of the membrane with $m=0.5,v=0.75$; it cuts through both the yellow and the red light cones but is spacelike with respect to the blue light cones. Correspondingly, its normal vector belongs to a light sector centered on the $x$ axis. Indeed, taking $m\in[0,1], |v|<1$ samples all normal vectors within the lower half of the $+x$ light sector (or the upper half of $-x$ light sector), i.e.
\begin{equation}
\mathcal{E}_{\vec{r}}(m,v) = g_x = \frac{1}{v}.
\end{equation}

Since the area of the membrane projected onto an equal-time slice is $L_y |v| t$, we conclude that the entanglement entropy across the cut grows as
\begin{equation}\label{eq:entropy-slp}
S_m(vt,t) = L_y t \quad [|v|<1].
\end{equation}
We find good agreement with this equation for a sampling of $m\in [0,1]$; the entropy growth for $m=0.5$ is shown in
~\autoref{fig:entropy-slp}.

Unlike in the other models we have studied, the membrane tension in the square lattice parity model does not admit a decomposition into flow directions \`a la~\autoref{eq:entanglement-channels}. This is because \textit{every} point is spacelike with respect to some unitary direction: the three pairs of light cones centered on unitary arrows of time have gaps in between them. The lines along $\begin{pmatrix} \pm 1 & 1 & 0 \end{pmatrix}$, $\begin{pmatrix} \pm 1 & 0 & 1 \end{pmatrix}$, $\begin{pmatrix} 0 & \pm 1 & 1 \end{pmatrix}$ are lightlike with respect to two of the three unitary direction, but spacelike with respect to the third.

\subsection{Ternary-unitary models}
Like the square lattice parity model, the ternary-unitary circuits of Ref.~\cite{Milbradt2023} are a class of models in which the time evolution along the $x$, $y$, and $z$ directions is unitary. In that construction, each layer of gates decomposes into 
\begin{equation}
\mathbb{U} = \prod_{oo} U_{oo} \prod_{ee} U_{ee}
\end{equation}
where $U_{oo}$ is a four-qubit unitary gate acting on each plaquette with coordinates $(2n+1, 2m+1)$, and $U_{ee}$ acts on each plaquette with coordinates $(2n,2m)$. As depicted in~\autoref{fig:square-lattice-layer}, the same decomposition holds for each of the square lattice parity model, with $U_{oo}^{(1)}=U_{ee}^{(1)}=U$ where
\begin{equation}
    U(i,j) = \prod_{s=\pm 1/2}\cnot((i+s, j+s)\rightarrow (i-1/2,j+1/2)) \cnot((i+s,j+s)\rightarrow (i+1/2,j-1/2))
\end{equation}
and $U_{oo}^{(2)} = U_{ee}^{(2)}$. But unlike in Ref.~\cite{Milbradt2023}, where each plaquette operator $U_{ee}, U_{oo}$ satisfies the "ternary unitarity" conditions
\begin{equation}
\scalebox{0.5}{\input{ternary-unitary-z.tikz}}, \quad \scalebox{0.5}{\input{ternary-unitary-x.tikz}}, \quad \scalebox{0.5}{\input{ternary-unitary-y.tikz}},
\end{equation}
only the first condition (unitarity in the $z$ direction) is satisfied at the plaquette level for the square lattice parity model. Unitarity along all three directions is restored by the fact that in even layers, $\mathbb{U}_2$, the controls and targets are flipped, i.e. ${U_{oo}^{(1)}}^{o_1 o_2 o_3 o_4}_{i_1 i_2 i_3 i_4} = {U_{oo}^{(2)}}^{o_2 o_3 o_4 o_1}_{i_2 i_3 i_4 i_1}$.

\begin{figure}[t]
\begin{minipage}[c]{0.5\linewidth}
\input{square-lattice.tikz}
\caption{Odd layer of square lattice parity model. With time running out of the page, each spider has an input leg (into the page) and output leg (out of the page), not shown. In even layers, the X and Z spiders are switched, or equivalently, the lattice is shifted one place to the right. \label{fig:square-lattice-layer}}
\end{minipage}
\hfill
\begin{minipage}[c]{0.4\linewidth}
\includegraphics[width=\linewidth]{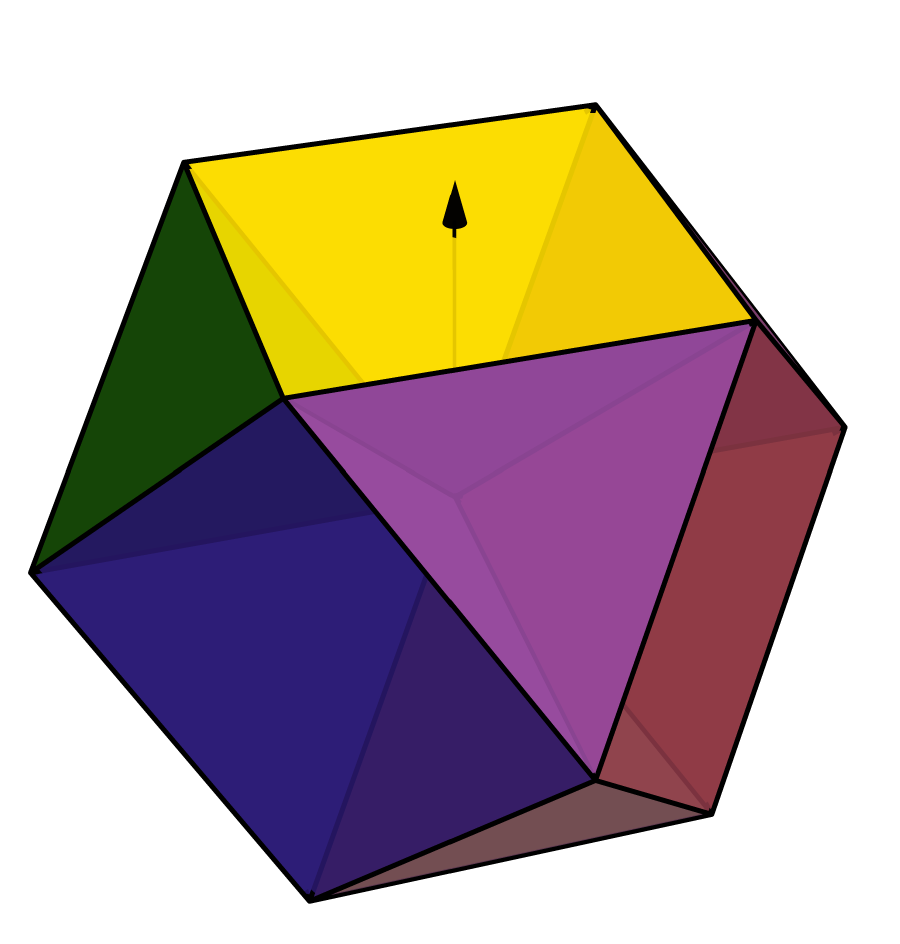}
\caption{Cuboctahedron of light sectors in the ternary-unitary models introduced in Ref.~\cite{Milbradt2023}. The black arrow points in the conventional time direction $\begin{pmatrix} 0 & 0 & 1 \end{pmatrix}$, and the blue face is centered on $\begin{pmatrix} 1 & 0 & 0 \end{pmatrix}$. \label{fig:cuboctahedron}}
\end{minipage}
\end{figure}

A key distinction between the square lattice parity model and the ternary-unitary circuits is in the structure of the light cones and light sectors. In ternary-unitary circuits, running time in the $z$ direction, operators spread within squares with corners at $(\pm t, t, t), (t,\pm t, t)$, i.e. the light cones of ternary-unitary models are the light \textit{sectors} of the square lattice parity model. The lines $|x|=|y|=|z|$ are common to all three pairs of light cones, hence the conclusion in Ref.~\cite{Milbradt2023} that two-point correlation functions of one-site observables can be nonzero along these directions (and only these directions). These are the flow directions $\vec{\hat{u}}_i$ of~\autoref{eq:entanglement-channels}. In contrast, as mentioned above, the light cones of the square lattice parity model intersect only at the origin, ruling out all such correlations beyond $t=0$.

We therefore conjecture that in ternary-unitary models, the free energy per unit area of a membrane with normal vector $\vec{\hat{n}}$ is:
\begin{equation}\label{eq:ternary-u-normal}
\ef(\vec{\hat{n}}) = (|\vec{\hat{n}} \cdot \begin{pmatrix} 1 & 1 & 1 \end{pmatrix}| + |\vec{\hat{n}} \cdot \begin{pmatrix} -1 & 1 & 1\end{pmatrix}| + |\vec{\hat{n}} \cdot \begin{pmatrix} 1 & -1 & 1 \end{pmatrix}| + |\vec{\hat{n}} \cdot \begin{pmatrix} 1 & 1 & -1 \end{pmatrix}|)/4
\end{equation}
where the normalization comes from considering the purely spatial membrane with $\vec{\hat{n}} = \vec{\hat{e}}_3$.

According to this conjecture,  the boundaries on the light sectors of the ternary-unitary circuits are the loci of $\vec{\hat{n}}$ for which one of the flow directions runs inside the membrane normal to $\vec{\hat{n}}$ [cf.~\autoref{eq:channel-boundaries}]. The sectors are therefore the six square pyramids and eight tetrahedra that together form a cuboctahedron (\autoref{fig:cuboctahedron}). When $\vec{\hat{n}}$ points inside one of the square pyramids, such as the conventional time arrow $(001)$, the membrane is spacelike with respect to one of the three unitary directions, so the \dam~are fully unitary, e.g., $\ef((001)) = 1$. When $\vec{\hat{n}}$ points inside one of the tetrahedra, e.g. $\vec{\hat{n}} = (111)/\sqrt{3}$, the membrane intersects all pairs of light cones, and the~\dam~are only unitary within a subspace. Conjecturing that entropy flows along the channel directions allows us to deduce the membrane tension within these tetrahedral sectors, and we can read off the dimension of the unitary subspace from the free energy density at the center of the sector, i.e.
\begin{equation}
\ef(\begin{pmatrix} 1 & 1 & 1 \end{pmatrix}/\sqrt{3}) = \sqrt{3}/2.
\end{equation}
Indeed, rewriting~\autoref{eq:ternary-u-normal} as:
\begin{equation}
\ef(\vec{\hat{n}}) = \max(\max_i |\vec{\hat{n}} \cdot \vec{\hat{e}}_i|, \max_i|\vec{\hat{n}}\cdot \vec{t}_i|/2)
\end{equation}
we can loosely interpret the light sectors as the intersection of the light cube of the square lattice parity model (maximal entropy directions $\vec{\hat{e}}_i$, cf.~\autoref{eq:normal-slp}) and the light octahedron of the the BCC spider model (maximal entropy directions $\vec{t}_i$, cf.~\autoref{eq:normal-bcc} and~\autoref{eq:octa-dir}).

Making the same cuts as in the numerics of the square lattice parity model, with time flowing along the $z$ axis,~\autoref{eq:ternary-u-normal} implies
\begin{equation}\label{eq:theory-ternary}
S_m(vt, t) = L_y t \begin{cases}
1 & |v| \leq 1 - m \\
\frac{|v|-1}{2m} + \frac{3}{2} & 1-m < |v| < 1+m \\
|v| & |v| \geq 1+m
\end{cases}.
\end{equation}
The numerics for a good scrambling ternary-unitary circuit based on iSWAP gates roughly confirm this conjecture, although the numerically observed entanglement consistently falls slightly below the theoretical prediction~\autoref{fig:ternary-u}.

\begin{figure}[t]
\includegraphics[width=\linewidth]{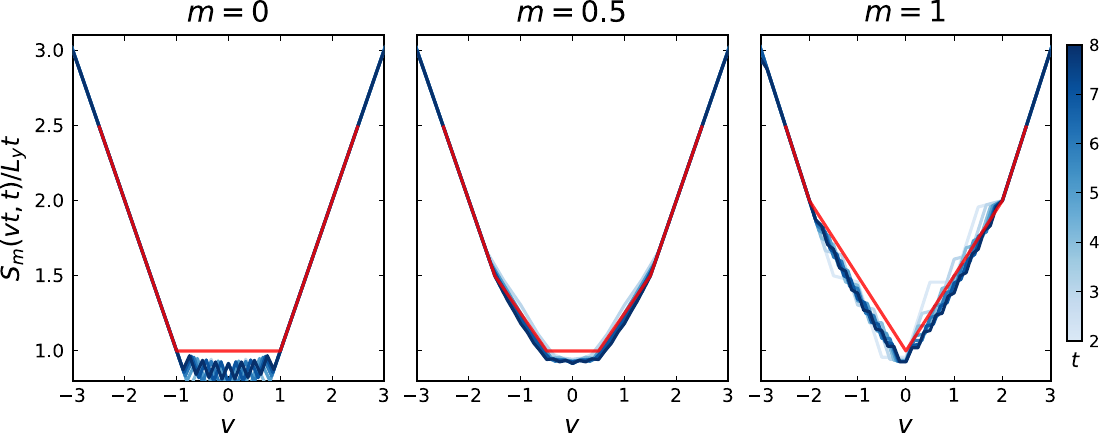}
\caption{\label{fig:ternary-u} Entropy growth in a ternary-unitary circuit, with time running in the $(001)$ direction. Red lines are the theoretical predictions (\autoref{eq:theory-ternary}).}
\end{figure}

\subsection{BCC Spider Model}\label{sect:bcc}
Turning to the BCC model depicted in~\autoref{fig:bcc-lattice}, we identify the flow directions along the $x$, $y$ and $z$ axes by comparing the plateau entropy density running time along two different, high-symmetry directions: $(001)$ and $(111)$. In both cases, the ZX diagram can be interpreted as a measurement-only model via the definitions~\autoref{eq:z-spider} and~\autoref{eq:x-spider} of X and Z spiders. We use units where the X spiders are at positions $(i,j,k)$ and Z spiders are at positions $(i+1/2,j+1/2,k+1/2)$ where $i,j,k$ are integers.

\begin{figure}[t]
\subfloat[]{
\includegraphics[width=0.4\linewidth,height=3in,keepaspectratio]{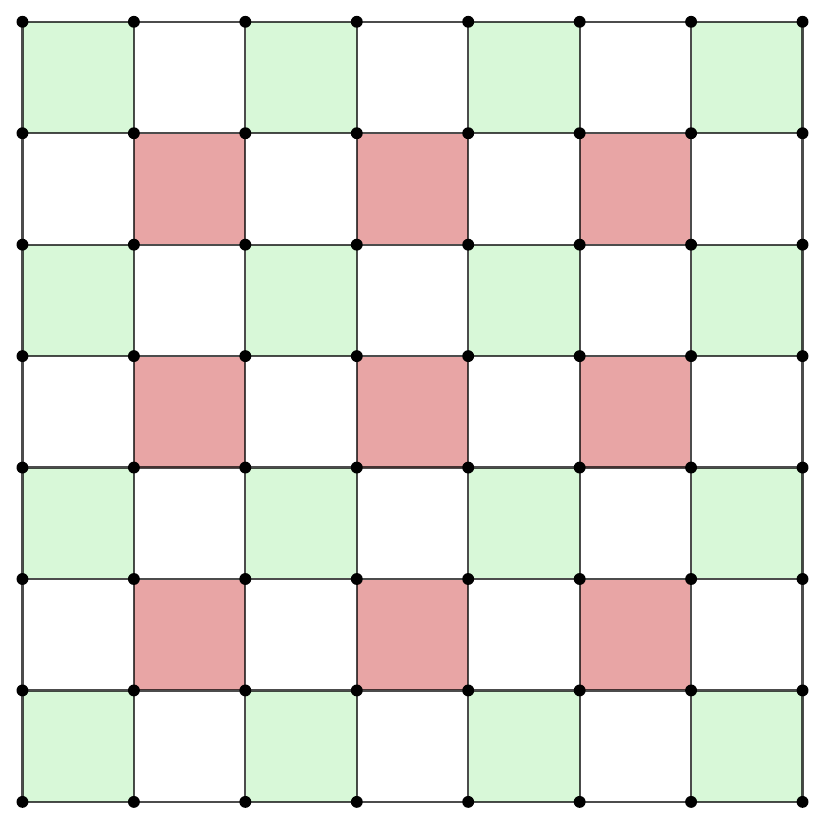}
\label{fig:bcc-z}
}
\subfloat[]{\includegraphics[width=0.6\linewidth,height=3in,keepaspectratio]{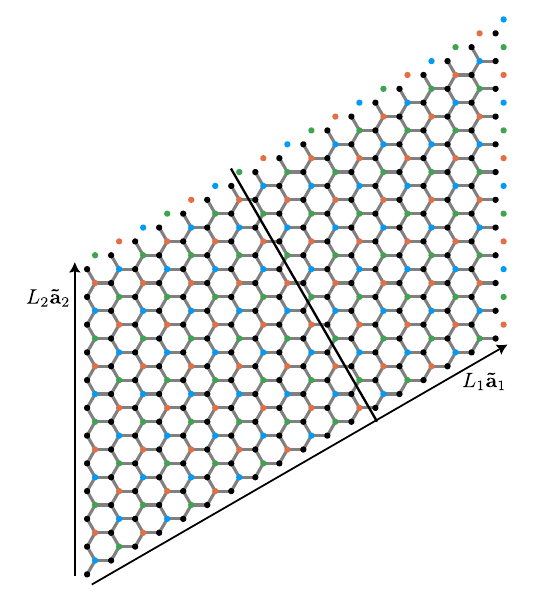}
\label{fig:layers-111}
}
\caption{ Qubit lattice in the BCC spider model running time in the (a) $(001)$ direction. Vertices are qubits, and green (red) shading of a plaquette indicates a Z (X) spider acting on the qubits around the plaquette. Z spiders and X spiders act on alternating layers in time. (b) $(111)$ direction. Black vertices are qubits involved in a gate in every layer, while blue, orange, green vertices are qubits involved only in the gate at that location, which occurs every three layers. System size is $(L_1, L_2) = (18,12)$ with boundaries aligned with $\vec{\tilde{a}}_1,\vec{\tilde{a}}_2$. Also shown is the cut with $m=0.5$ used to pin the membrane.} 
\end{figure}
Running time in the $z$ direction, the circuit alternates between layers of Z spiders at half-integer times and X spiders at integer times. Each spider has four inputs and four outputs, so it is natural to assign qubits to the lattice $(i+1/4,j+1/4)$; the spiders act as plaquette operators on the qubit lattice (\autoref{fig:bcc-z}).

Each spider projects a 16-dimensional input Hilbert space onto the subspace spanned by
\begin{equation}
\ket{0000}_{x,z}, \, \, \ket{1111}_{x,z},
\end{equation}
i.e. the spins are aligned in the $X$ ($Z$) basis for X (Z) spiders. In terms of stabilizers, a 4-input, 4-output X spider is a forced measurement of $X_1 X_2, X_2 X_3, X_3 X_4$, post-selecting on +1 for all three measurements. Likewise, a 4-input, 4-output Z spider measures $Z_1 Z_2, Z_2 Z_3, Z_3 Z_4$.

After one layer of measurements, a fully mixed state is reduced to entropy density 1/4, placing an upper bound on the entropy density of the plateau. In the limit $L_x, L_y \rightarrow \infty$, the plateau entropy density converges toward this bound. For example, with periodic boundary conditions, if $L_x$ and $L_y$ have the same parity, then we numerically observe that the plateau entropy is
\begin{equation}
    S_{plat}(L_x,L_y) = (\lceil L_x/2 \rceil -1) (\lceil L_y/2 \rceil - 1) 
\end{equation}
The plateau is reached after just one layer if $L_x$ and $L_y$ are both odd, and two layers otherwise. Since there are 4 qubits/gate and the gate density is 1, the asymptotic entropy density per unit area is 1. This is the surface tension of a membrane running parallel to the xy plane, i.e.
\begin{equation}\label{eq:tension-001}
\ef(\begin{pmatrix} 0 & 0 & 1 \end{pmatrix}) = 1.
\end{equation}

Running time in the $(111)$ direction, three of the four outputs of a given spider go on to a gate in the next layer, while the fourth output skips the next two layers. The symmetry of the spiders with respect to inputs and outputs gives us freedom to choose which qubit to assign to the fourth output leg. We choose to have half of the qubits involved in a nonunitary gate in every layer, and the other half involved in a gate every third layer.

With respect to the orthonormal basis
\begin{equation}
    \vec{\tilde{e}}_1 = \begin{pmatrix} 0 & 1 & -1 \end{pmatrix}/\sqrt{2}, \quad \vec{\tilde{e}}_2 = \begin{pmatrix} -2 & 1 & 1 \end{pmatrix} / \sqrt{6}, \quad \vec{\tilde{e}}_3 = \begin{pmatrix} 1 & 1 & 1 \end{pmatrix} / \sqrt{3} \equiv \vec{\hat{t}},
\end{equation}
the gates on a constant time layer form a triangular lattice in the $\vec{\tilde{e}}_1, \vec{\tilde{e}}_2$ plane with lattice spacing $\sqrt{2}$. Concretely, let
\begin{equation}
\vec{a}_1 = \sqrt{2} \vec{\tilde{e}}_1 = \begin{pmatrix} 0 & 1 & -1 \end{pmatrix}, \quad
\vec{a}_2 = \frac{1}{\sqrt{2}} (\vec{\tilde{e}}_1 + \sqrt{3} \vec{\tilde{e}}_2) = \begin{pmatrix} -1 & 1 & 0 \end{pmatrix}, \quad \vec{a}_3 = \frac{1}{2 \sqrt{3}} \vec{\tilde{e}}_3 = \begin{pmatrix} 1 & 1 & 1 \end{pmatrix}/6.
\end{equation}
Then letting $2(x+y+z) = t$, 
\begin{equation}
\begin{pmatrix} x & y & z \end{pmatrix} = (t/6-z) \vec{a}_1 + (t/6-x) \vec{a}_2 + t \vec{a}_3.
\end{equation}

In these coordinates, the spacing in time between consecutive layers is $1$ (since $x+y+z$ is an integer for X spiders, and half-integer for Z spiders), and the triangular lattice on a given layer has primitive lattice vectors $\vec{a}_1, \vec{a}_2$. That is, the gates are located at the coordinates $(n + o_1) \vec{a}_1 + (m + o_2) \vec{a}_2$ for integer $n,m$, where the origin of the lattice $(o_1,o_2)$ evolves in time with a period of 3:

\begin{equation}
o_1 = o_2 = - (t \, \mod \, 3) / 3
\end{equation}

Placing our qubits at the vertices of the three triangular lattices defined at times $0, 1, 2$, as well as the three triangular lattices with origins $(o_1,o_2) = (1/3,0), (0,2/3), (2/3,1/3)$, we obtain the honeycomb lattice depicted in~\autoref{fig:layers-111}. We then assign to each gate the qubit at the location of the gate and the three neighboring qubits on the honeycomb.

\begin{figure}[t]
\includegraphics[width=\linewidth]{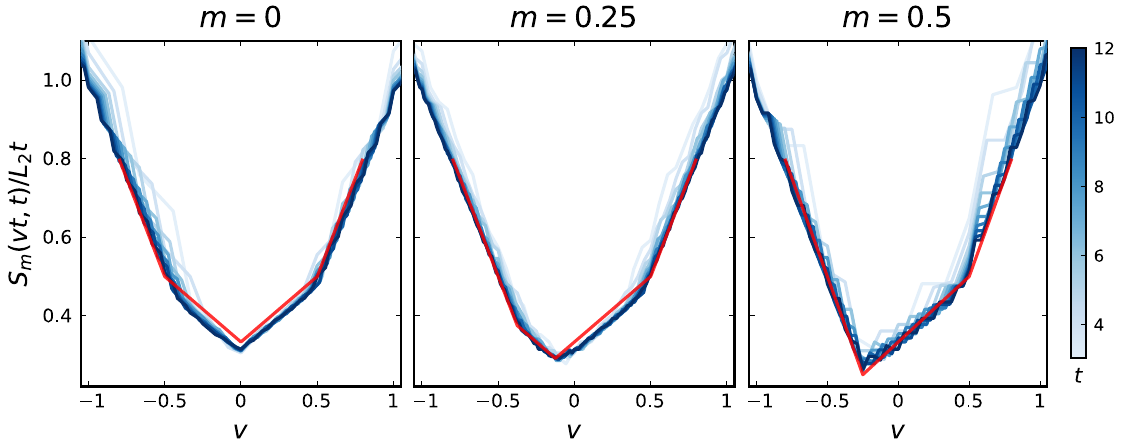}
\caption{Entropy growth in the BCC spider model, with time running along the (111) direction. Red lines are the theoretical predictions (\autoref{eq:Sm}).\label{fig:bcc-page} 
}
\end{figure}

With periodic boundary conditions, a fully mixed state on $N$ qubits reaches the plateau after exactly one time step, with an entropy of $N/2$. 
Since the qubit density is $2\sqrt{3}$, this corresponds to a surface tension of
\begin{equation}\label{eq:tension-111}
\ef(\begin{pmatrix} 1 & 1 & 1 \end{pmatrix}/\sqrt{3}) = \sqrt{3}.
\end{equation}
Anticipating that the flow directions $\vec{\hat{u}}_i$ and the maximal entropy directions $\vec{\hat{t}}$ will both be lines of high symmetry, the ratio of $\sqrt{3}$ between~\autoref{eq:tension-111} and~\autoref{eq:tension-001} leads us to conjecture that the time vector $(001)$ is along a flow direction whereas the time vector $(111)$ is along a direction of maximal entropy. This in turn implies the surface tension~\autoref{eq:normal-bcc}, since the boundaries of the tetrahedra in~\autoref{fig:light-octahedron} are precisely the orientations $\vec{\hat{n}}$ for which the $x$, $y$, or $z$ axis runs inside the membrane: the light sectors and light cones are the same. 
If we had instead found a plateau entropy of $N/6$ (surface tension $1/\sqrt{3}$) along the $(111)$ direction, we would have deduced that the light sectors form the light cube found for the square lattice parity model.

In terms of operator spreading, the $x$, $y$, and $z$ axes are the "edges" of spreading operators when time runs along a maximal entropy direction such as $(111)$. That is, within the unitary subspace defined by the plateau group, logical operators spread within the light cones shown in~\autoref{fig:light-octahedron}. The $x$, $y$, and $z$ axes are common to all three pairs of light cones (are not spacelike with respect to any time arrow), hence their designation as flow directions. Running time along the $(001)$ direction, we instead find an infinite butterfly velocity owing to the flow directions along the $x$ and $y$ axes. 

Rotational symmetry imposes that the $x$, $y$, and $z$ flow directions have equal capacity, and indeed, we can rewrite~\autoref{eq:normal-bcc} in the form of~\autoref{eq:entanglement-channels} by taking $s=1$, $\vec{\hat{u}}_i = \vec{\hat{e}}_i$:
\begin{equation}\label{eq:bcc-tension}
    \ef(\vec{\hat{n}}) = \max_{i} |\vec{\hat{n}} \cdot \vec{t}_i| = \sum_{i=1}^3 |n_i| = \sum_{i=1}^3 |\vec{\hat{n}} \cdot \vec{\hat{e}}_i|.
\end{equation}
To verify~\autoref{eq:normal-bcc}/~\autoref{eq:bcc-tension}, we probe membranes with normal vectors in between (111) and (001) by again using the Page curve protocol. Running the protocol with time in the (001) direction is hazardous, owing to the infinite butterfly velocity. Thus, we run time in the (111) direction, with open boundary conditions as depicted in~\autoref{fig:layers-111}. The black qubits occupy a triangular lattice with $\vec{\tilde{a}}_1 = (\vec{a}_1 + \vec{a}_2)/3, \vec{\tilde{a}}_2 = (2\vec{a}_2-\vec{a}_1)/3$. We track entanglement growth on the lattice spanned by $L_1 \vec{\tilde{a}}_1, L_2 \vec{\tilde{a}}_2$, with the initial pinning line from the origin to $-mL_2\vec{\tilde{a}}_1 + L_2 \vec{\tilde{a}}_2$, and final pinning line from $vt \vec{\tilde{a}}_1  + t \vec{a}_3$ to $(vt - mL_2) \vec{\tilde{a}}_1 + L_2\vec{\tilde{a}}_2 + t\vec{a}_3$.

For a given $m$,~\autoref{eq:bcc-tension} implies:
\begin{equation}\label{eq:Sm}
S_m(v,t) = \frac{t L_2}{6} \left(|m+2v| + |2v-1| + |2v+1-m|\right). 
\end{equation}
The numerics at early times are roughly consistent with this equation, as shown for three examples ($m=0,1/4,1/2$) in~\autoref{fig:bcc-page}. Slight discrepancies can be attributed in part to details at the boundaries.

\subsubsection{Tri- and Octahedral-unitarity}
Like our BCC spider model, the 2+1d tri-unitary circuits composed of three-qubit gates in Ref.~\cite{Jonay2021} have flow directions along the $x, y, z$ axes. However, in that work there are only three strictly unitary directions: $\begin{pmatrix} 1 & 1 & 1 \end{pmatrix}, \begin{pmatrix} -1 & 1 & 1 \end{pmatrix}, \begin{pmatrix} -1 & -1 & 1 \end{pmatrix}$. Unitarity along the fourth symmetry-related direction, $\begin{pmatrix} 1 & -1 & 1 \end{pmatrix}$, imposes an additional condition on the gates, resulting in "octahedral-unitary" gates which are a strict subset of tri-unitaries~\cite{Mestyan2022}. The BCC spider model can be viewed as a "generalized octahedral-unitary" circuit as it is unitary within a maximal-entropy-density subspace along all four directions, but we have not found an explicitly unitary formulation.

Interestingly, tri-unitarity is enough to constrain the flow directions to lines along $x, y, z$. To recapitulate the argument in Ref.~\cite{Jonay2021}: each of their three arrows of time makes the correlations vanish outside two octants (forward and backward light tetrahedra). The only lines common to all three sets of octants intersect are the $x, y, z$ axes. This argument follows from unitarity within the full Hilbert space, which imposes a strict light cone outside of which infinite-temperature one-site correlations vanish---or equivalently, if an entanglement membrane falls entirely outside the light cone, its tension can be diagrammatically evaluated just using unitarity. 

While unitarity along the fourth direction would impose a redundant constraint on the channels for two-point correlations, it is necessary to be able to deduce $\ef(\vec{\hat{n}})$ when $\vec{\hat{n}}$ points inside the fourth pair of tetrahedra. A membrane with such an $\vec{\hat{n}}$ intersects all three other pairs of tetrahedra, so tri-unitarity alone is not enough to simplify the surface tension. Thus, tri-unitary circuits in 2+1d are not "fully multi-unitary", and~\autoref{eq:bcc-tension} only applies within three pairs of light sectors.

\section{Broad survey of hybrid circuits}\label{app:hybrid}
The dual hybrid circuit in~\autoref{app:dual-hybrid} of this Supplement is a special case of a hybrid circuit that is explicitly tri-unitary. In this section we move away from this particular limit, considering larger unit cells populated by random Clifford gates (possibly a mix of iSWAP and CNOT cores) and one or more projective Pauli measurements.

In a stabilizer circuit, measurement outcomes are either deterministic or uniformly random~\cite{Gottesman1998,Aaronson2004}. In the latter case, the two equally likely outcomes only differ with respect to the sign on one stabilizer generator, which does not affect entanglement observables. Therefore, if the gates and measurement locations are STTI, the dynamics of the stabilizer group "modulo signs" will also be STTI. (Equivalently, to fix the signs, we can implement classical feedback: following every measurement that yields a $-1$ outcome, apply a Pauli operator that flips the sign to $+1$.

An important consequence of this independence of measurement outcomes is that the entanglement dynamics in hybrid Clifford circuits where the measurements are sampled according to Born probabilities is equivalent to the entanglement dynamics in Clifford circuits with forced projectors (such as those those that arise when running a unitary Clifford circuit at some angle). The latter kind of tensor network is what we have in mind for the~\dam~of a generic circuit, and for non-Clifford evolution, differs in a meaningful way from the former, a point we elaborate upon below.

In Clifford circuits, the STTI time evolution from a fully mixed initial state always leads to a steady-state "plateau" stabilizer group. (We refer to this as a "plateau group" rather than "plateau state" because we do not keep track of signs on stabilizer generators.) For concreteness, we consider one single-qubit Pauli measurement per unit cell, in which case the plateau entropy density $s^*\leq (a-1)/a$. As discussed in the main text, this entropy density tells us the line tension of spacelike interfaces, which separate the initial time slice from the final time slice~\cite{Li2021,Li2023,Sang2023}; if the plateau is reached within a time that scales more slowly than $L$ as we take $L\rightarrow\infty$, then the interface in runs along the spatial direction. The fully mixed boundary condition makes these interfaces sensitive to the system size in a fractal way.

The sensitivity to boundaries manifests in several ways, as we now describe. In the following, let $m=L/a = j 2^n$ denote the number of unit cells in the spatial direction, where $n$ is an integer and $j$ is odd. When we refer to the "dual" of a circuit, we mean the circuit with time and space swapped in the bulk, possibly with different boundary conditions.

\subsection{Plateau proofs}

We begin by proving some results about the purification dynamics in Floquet Clifford circuits starting from the fully mixed state. These results do not require the spatial translation invariance of our STTI circuits; only time translation invariance is necessary.\footnote{However, to apply these results to the dynamics run along any angle {commensurate with the lattice} would require full spacetime translation invariance.} Thus, when we refer to a "layer of measurements," this can be a translation-invariant or spatially random set of commuting single- or multi-qubit measurements all on the same time slice. These results also hold for $N>1$ measurement layers per period.

Let $t^*$ denote the time step of the last purification event, and $s^*$ the corresponding entropy density. In the absence of purifications, the action of the brickwork circuit from time $t_1$ to $t_2$ can be described by an effective unitary $U(t_1,t_2)$, which maps the stabilizer generators and logical operators at time $t_1$ to those at time $t_2$. Note that $U$ does not necessarily inherent time translation invariance from the circuit, since it depends on the input and output stabilizers, i.e. $U(n,n+1)$ is, in general, not equal to $U(m,m+1)$ for $m\neq n$.

Now let $\mathcal{S}(t)$ denote the stabilizer group (ignoring signs) at time $t$. For a finite system size, there is only a finite (albeit exponentially large) number of stabilizer groups with this entropy. Therefore, there must be a time $t_{plat} \geq t^*$ such that $\mathcal{S}(t_{plat}) = \mathcal{S}(t_{plat}+\tau)$ for some integer period $\tau$. Then for $t_1\geq t_{plat}$, since the inputs at time $t_1$ and $t_1+\tau$ are the same, in this case we \textit{can} deduce that $U(t_1,t_1+\tau)=U(t_1 + n\tau,t_1+(n+1)\tau)$ for all $n\geq 0$. That is, for $t\geq t_{plat}$, the system belongs to a "plateau" of period $\tau$.

In general, $\tau$ can be exponentially large in $L$, and $t_{plat}$ may exceed $t^*$ by $O(L)$. However, evolution from a fully mixed initial state enjoys two special properties: 

\begin{theorem}\label{theor:plateau}
Purification under an STTI (or simply TTI) Clifford circuit starting from the fully mixed state arrives at a static plateau group ($\tau=1$). This group is reached immediately after the last purification event, i.e. $t_{plat}=t^*$.
\end{theorem}
\begin{proof}
To prove the first property, suppose for contradiction that $\tau > 1$, so $\mathcal{S}(t_{plat}), \mathcal{S}(t_{plat}+1),..., \mathcal{S}(t_{plat}+\tau-1)$ are distinct stabilizer groups. We can then define a higher-entropy "plateau group" $\mathcal{S}'$ consisting of this ensemble of groups. The maximally mixed state would therefore have stopped purifying at this higher-entropy group $\mathcal{S}'$ instead of $\mathcal{S}(t^*)$.

Now consider a purification of the maximally mixed state using $\mathcal{S}(t_{plat})$ and its logicals as a basis; i.e., we can choose the logical $Z$ (logical $X$) operators at time $0$ to be the stabilizer (destabilizer) and logical $Z$ (logical $X$) generators of the plateau group. In this basis, a given logical is tagged as either “destined to remain a logical” or “destined to become a stabilizer/destabilizer.” The stabilizer group and destabilizer group is preserved under one time step of the dynamics, and the purification events just convert these “future (de)stabilizers” into “actual (de)stabilizers”, while the logical generators evolve according to the unitary $U$. Since the only stabilizers added during the evolution are those belonging to the plateau group, once the state has finished purifying at $t=t^*$, it must be the plateau group, so $t^*=t_{plat}$.
\end{proof}

Henceforth the phrase "the plateau" will refer to the steady-state plateau group $\mathcal{S}^*=\mathcal{S}(t^*)$ reached from the fully mixed initial condition. The proof of~\autoref{theor:plateau} also implies that any initial state with a stabilizer group that is a subgroup of $\mathcal{S}^*$ will likewise stop purifying at $s=s^*$. Meanwhile, since all initial states are contained within the fully mixed initial state, a generic initial condition will evolve to an ensemble of $\tau$ stabilizer groups that all contain $\mathcal{S}^*$ as a subgroup (i.e. the corresponding states are lower-entropy states that "belong" to the mixed state stabilized by $\mathcal{S}^*$).

A related property of the purification dynamics holds for the fully mixed initial condition, or any state in the sequence of stabilizer groups from fully mixed to $\mathcal{S}^*$:
\begin{theorem}\label{theor:subgroup}
Let $\mathcal{S}(t,s)$ denote the stabilizer group at a fixed instant $s$ within the time step $t$, starting from the fully mixed state (trivial group) at $(t,s)=(0,0)$. Then $\mathcal{S}(t-1,s)$ is a subgroup of $\mathcal{S}(t,s)$.
\end{theorem}
\begin{proof}
We will choose the instant of time $s$ to be immediately after a layer of measurements, and let $\mathcal{S}(t) \equiv \mathcal{S}(t,s)$. To simplify the proof, we assume there is only $N=1$ layer of measurements per Floquet period.

Let $U_F$ denote the unitary evolution under one time step of the brickwork circuit, from just after a layer of measurements to just before the next layer of measurements. $U_F$ is independent of $t$ since the circuit is Floquet. Also let $\overline{Z}(0)$ denote the stabilizer group generated by the measurement operators, which, in the absence of measurements, would evolve unitarily as
\begin{equation}
    \overline{Z}(t) = U_F^t \overline{Z}(0) (U_F^\dagger)^t
\end{equation}
Finally, as shorthand, define:
\begin{equation}\label{eq:normal1}
    M(t)[G] = \mathcal{N}_{\langle G, \overline{Z}(t)\rangle}(\overline{Z}(t))
\end{equation}
where $\mathcal{N}_B(A)$ is "everything in $B$ that commutes with everything in $A$". That is, $M(t)[G]$ is the group generated by the measurement(s) time-evolved by $t$ steps, along with subgroup of the group $G$ that commutes with those time-evolved measurements. This is a form of "code-switching", i.e. the map $S \rightarrow \mathcal{N}_{\langle S, F\rangle} (F)$~\cite{ECZoo,Aasen2023}.

Starting from the trivial group $\mathcal{S}(0) = \langle \emptyset \rangle$, and $\mathcal{S}(1) = \overline{Z}(0)$, the stabilizer group after the layer of measurements in time step $t$ is:
\begin{equation}\label{eq:S_tn}
\mathcal{S}(t) = M(0)[M(1)[M(2)[...M(t-2)[\overline{Z}(t-1)]]]]
\end{equation}
This follows from induction: After the second round of measurements,
\begin{equation}
    \mathcal{S}(2) = M(0)[\overline{Z}(1)].
\end{equation}
Now suppose~\autoref{eq:S_tn} holds for $t-1$. The stabilizer group then evolves unitarily under $U_F$:
\begin{align}\label{eq:unitary-S1}
U_F \mathcal{S}(t-1) U_F^\dag = M(1)[M(2)[...M(t-2)[\overline{Z}(t-1)]]]
\end{align}
where we have used the fact that commutators of Paulis are preserved under Clifford gates:
\begin{align}
    U M(t-1)[G] U^\dagger = M(t) [U G U^\dagger].
\end{align}
Then, the system undergoes a layer of measurements, implementing the code-switching map:
\begin{equation}
    M(1)[M(2)[...M(t-2)[\overline{Z}(t-1)]]] \rightarrow M(0)[M(1)[M(2)[...M(t-2)[\overline{Z}(t-1)]]]],
\end{equation}
in agreement with~\autoref{eq:S_tn}.

The theorem then follows by observing that for $t'=0,...,t-3$, the argument of $M(t')$ in the expression for $\mathcal{S}(t-1)$ is a subgroup of the argument of $M(t')$ for $\mathcal{S}(t)$, and that
\begin{equation}
    G \leq H \rightarrow M(t)[G] \leq M(t)[H]
\end{equation}
\end{proof}

An immediate consequence of this theorem is that the entropy $S(t,s)$ is a convex function of $t$:
\begin{equation}\label{eq:convex}
    S(t-1, s) - S(t, s) \geq S(t,s) - S(t+1,s)
\end{equation}
In particular, this means that there are no "false plateaus" in the entropy density: if $S(t+1,s)=S(t,s)$, then $s^*=s(t,s)$. 

\autoref{eq:convex} follows from comparing the stabilizer groups immediately before the measurements at time $(t-1,s)$ and $(t,s)$. Measurement of the operator $P$ can lead to three outcomes~\cite{Aaronson2004}:
\begin{enumerate}
    \item $P$ commutes with all of $\mathcal{S}$, but does not belong to $\mathcal{S}$, in which case a purification event occurs ($\pm P$ is added as a generator).
    \item $P$ commutes with all of $\mathcal{S}$ and $\pm P \in \mathcal{S}$, so the measurement is deterministic and $\mathcal{S}$ is unchanged.
    \item $P$ anticommutes with some element of $\mathcal{S}$, so half of $\mathcal{S}$ gets replaced and the measurement outcome is random.
\end{enumerate}
If outcome (2) or (3) happens at time $t$, then by~\autoref{theor:subgroup} it must occur in all future Floquet periods. Hence once a measurement at a particular site fails to produce a purifying event, all subsequent measurements of that site (at integer time steps) will also not lead to purification.

Generalizing beyond the fully mixed initial condition, \autoref{theor:subgroup} and its corollary~\autoref{eq:convex} also hold for any initial state on the ramp from the fully mixed state to the plateau. That is, any state of the form~\autoref{eq:S_tn}, if fed into the STTI circuit starting immediately after a measurement layer, purifies in a concave-up fashion until reaching the plateau. 

For a generic initial state, the theorem and its corollary need not hold. However, the convexity of $S(t,s)$ for the fully mixed initial state does place an important constraint on the purification dynamics from any initial state. In Clifford circuits, the entropy is discretized to integer values. Starting from the maximally mixed state with entropy $L$, for $1<t\leq t^*$, at least one stabilizer is added per time step, so $t^*\leq L$. For an STTI circuit with $m=L/a$ measurements per layer, all $m$ measurements in the first round produce purifications, so in fact $t^*\leq 1 + m(a-1)$.  Since the fully mixed initial state contains all lower-entropy initial states, the time for an arbitrary initial state $\mathcal{S}_0$ to reach a state within the plateau, $t_{plat}(\mathcal{S}_0)$, cannot exceed $t^*$, which therefore imposes a linear upper bound on the plateau time for all initial states.

As a final remark, let us compare~\autoref{theor:subgroup} to~\autoref{theor:plateau}. The proof of the former does not even invoke the existence of a plateau group, but given that $\mathcal{S}(t-1,s)$ is a subgroup of $\mathcal{S}(t,s)$, the latter theorem comes for free: at every time step, the stabilizer group must be a subgroup of the plateau group (uniquely defined at a particular point in the Floquet period), so as soon as the initially fully mixed state stops purifying, it reaches the static plateau group.

One result from the proof of~\autoref{theor:plateau} that does not naturally follow from~\autoref{theor:subgroup} regards initial states that are subgroups of $\mathcal{S^*}$ but do not take the form of~\autoref{eq:S_tn}. Such initial states are not guaranteed to purify convexly, but \textit{do} satisfy $\mathcal{S}(t,s) \leq \mathcal{S^*}(s)$ for all $t$, and thus do not purify below the plateau.

\subsection{Unitarity within the plateau}
Before characterizing the plateau group of Clifford circuits, let us pause to remark upon the unitary evolution within the plateau subspace. 

Clifford circuits preserve stabilizerness, so clearly if we feed in a fully mixed state it will remain a (mixed) stabilizer state under the STTI Clifford evolution; in particular, it will have a flat spectrum. In other words, the state is fully mixed on the subspace stabilized by the stabilizer group, and each purification event cuts the dimension of the subspace in half. We can interpret this in two ways. If the measurement of an operator $P$ is a "true" measurement (sampled according to the Born rule, and followed, say, by classical feedback to correct any random $-1$ outcomes to $+1$), then each purification maps two orthogonal subspaces (those with eigenvalues $\pm 1$ under $P$) to the same $+1$ eigenspace, preserving the norm of the state. If the measurement of $P$ is a forced projector (i.e., $(\mathbbm{1} + P)/2$), then it annihilates the $-1$ eigenspace and preserves the $+1$ eigenspace, reducing the norm of an input mixed state by a factor of 2. Thus, we can view the purification dynamics in the Clifford circuit as either uniformly projecting out states outside of the plateau subspace, or uniformly mapping them onto the plateau. 

By definition, once the state belongs to the plateau, it cannot purify any further, so only (2) and (3) above are possible outcomes of a measurement. In case (2), the measurement just implements an identity gate. In case (3), the measured operator $P$ must anticommute with some $Q \in \mathcal{S^*}$, and we can again interpret this in two ways. If the measurement obeys the Born rule and is followed by classical feedback, then it implements the unitary gate $(Q+P)/\sqrt{2}$ on all states within the plateau~\cite{Yoganathan2019}. If the measurement is forced, then it implements $(Q + P)/2$, which is proportional to a unitary gate, and reduces the norm of all states within the plateau by the same amount, $\mathrm{Tr}(\rho_{out}) = \mathrm{Tr}(\rho_{in})/2$. Either way, the inner product between states belonging to the plateau subspace is preserved up to constant factor \textit{independent} of the particular states. This unitarity is a necessary condition for having \textit{local} dynamics on the plateau, i.e. a finite butterfly cone, in turn allowing us to define entire light sectors within which the line tension is linear. 

{More generally, if $K\rho K^\dag / [\mathrm{Tr}(K \rho K^\dag)$ has the same spectrum as $\rho$, where $K$ is the Kraus operator for one time step, then we say that $\rho$ belongs to the plateau. But unless $\mathrm{Tr}(K \rho K^{\dag})$ is independent of $\rho$ within the plateau, the plateau dynamics cannot be unitary. This should be contrasted with the theorem in Ref.~\cite{Gullans2020}, which states that the emergent plateau in the volume-law phase of a \textit{monitored} (not post-selected) circuit is approximately unitary. (When the plateau is permanently stable, the recovery is perfect and the dynamics within the code space is precisely unitary.) The proof of this theorem relies on the fact that the "unraveled channel" which contains all possible trajectories, weighted by their Born probabilities, is normalized, whereas the post-selected dynamics along a single trajectory is not. The results of Ref.~\cite{Gullans2020} can presumably be extended to typical trajectories, but rare trajectories such as the non-Hermitian evolution in Ref.~\cite{Gopalakrishnan2021} manifestly violate the independence condition.}

{A simple proof shows that the independence condition is also sufficient for unitarity within the plateau subspace: suppose $\mathrm{Tr}(K \rho K^{\dag}) = C$ for all $\rho$ belonging to the plateau. Letting $K_c = \sqrt{C} K$, $\bra{\psi} K_c^\dag K_c \ket{\psi} = 1$ for any normalized pure state $\ket{\psi}$ within the plateau, so $K_c$ acts as an isometry within this subspace.}

\subsection{Properties of the plateau group}
Returning to Clifford circuits, having proven the existence of a plateau group reached in at most $O(L)$ time steps, we now specialize to STTI circuits with single-qubit (non-entangling) measurements and classify them based on their purification dynamics and entanglement on the plateau. Unless otherwise noted, periodic boundary conditions are used to preserve the full translation invariance.

Starting from a fully mixed state, the purification process occurs in two stages. In the first $O(1)$ steps, all $m=L/a$ measurements (or a finite fraction of them) lead to purification events. In the second stage, which can last up to $O(m)$ steps, each measurement layer only results in $O(1)$ purifications.

There are two basic classes of circuits when PBCs are used. We use the term "gapped" to refer to circuits which stop purifying after the first stage, and "gapless" to refer to those which continue through the second stage.

When a circuit is gapped for all $m$, both $t^*$ and $s^*$ are independent of system size, and $s^*$ is simply the entanglement velocity of the dual circuit. An extreme example is when only the first layer of measurements leads to purifications, in which case the circuit is "pinned" to the maximal entropy density $(a-1)/a$. The code obtained by throwing away the just-measured qubits is then a trivial code with code rate 1 and code distance 1. (Recall that the distance of a code is the minimum weight of a nontrivial logical operator, i.e. the lowest weight of an undetectable Pauli error~\cite{nielsen2010}.) At the other extreme are circuits for which each subsequent measurement also leads to a purification event, so that an area-law-entangled pure steady state is reached in $a$ time steps. These circuits are the STTI counterparts of random monitored Clifford circuits in the pure/area-law phase, and their dual has $v_E=0$.


The purification dynamics in "gapless" circuits is more complex, with a striking sensitivity to the system size $m=j 2^n$. Generally, the time to reach the plateau grows linearly with $2^n$ but is independent of $j$, while $s^*$ depends on $j$ but not on $n$, as shown in~\autoref{fig:gapless} for unit cell $(b,a,d)=(4,4,0)$. An extreme case is where $s^*=0$ for $j=1$ but is nonzero for $j>1$ (see next subsection and Ref.~\cite{Sommers2023}). 

The mechanism behind this "fractal gaplessness" is that the second stage of purification is "fragile": the stabilizer generators added to the stabilizer group during this stage are of length $O(m)$, requiring a particular alignment with the system size in order to be added. The simplest trend is that the second stage ends after $1$ step for $n=0$ (odd $m$), 2 steps for $n=1$, and so on, such that as $j\rightarrow\infty$ the plateau entropy density asymptotes toward the entropy density of the state after the first purification stage. For example, in the gapless $(4,4,0)$ circuits surveyed, the plateau entropy density asymptotes toward either $s_1=3/4$, $s_1=1/2$, or $s_1=1/4$. It is this entropy density---an integer multiple of $1/a$---that quantifies the line tension of a membrane run along the spatial direction, rather than $s^*$. Indeed, even for $j=1$ where the state continues to purify well below $s_1$, in many cases $s_1$ is still lurking within the entanglement structure of the plateau group. For example, pure states sampled randomly within the mixed plateau exhibit Page curves with slope $s_{eq} = s_1$, rather than slope $s^*$. The "true" entropy density $s_1$ may also be recovered by switching to open boundary conditions, although even with OBCs the purification dynamics have some residual sensitivity to the system size, and the second stage of purification can still last for $O(m)$ steps.

From this perspective, gapless purification is only a finite-size phenomenon which, if we take the appropriate sequence of system sizes (i.e., odd $m$) in the $m\rightarrow\infty$ limit, disappears in the infinite system size limit. Nevertheless, the extra structure present in the plateau group of $j=1$ gapless circuits points to some interesting features beyond the entanglement membrane theory, which we briefly catalog here.

\begin{figure}[t]
\centering
\subfloat[]{
\includegraphics[width=0.5\linewidth,height=2.7in,keepaspectratio]{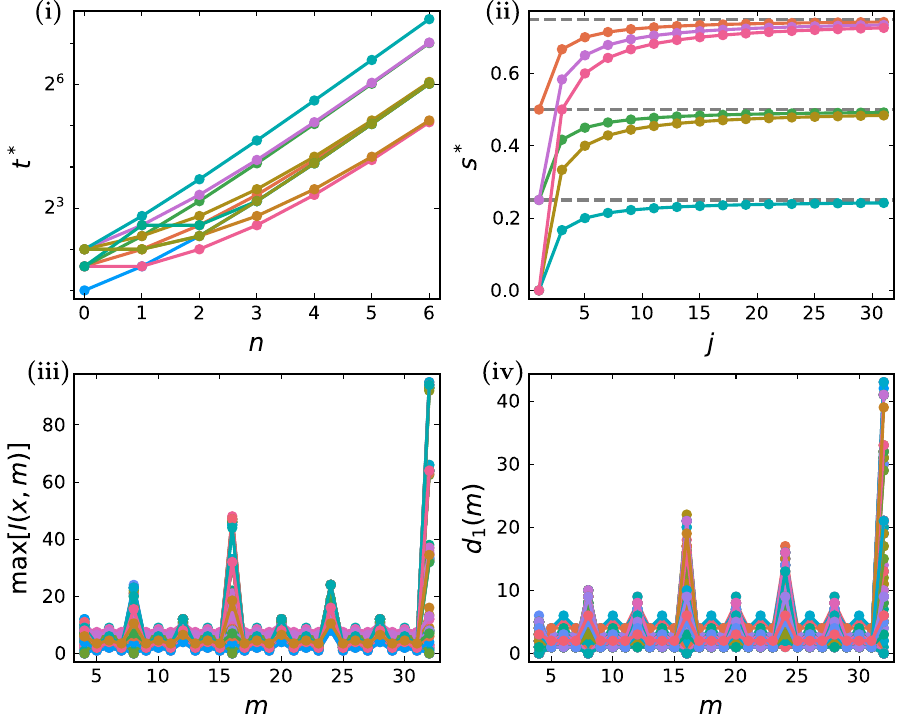}
\label{fig:gapless}}
\hfill
\subfloat[]{
\includegraphics[width=0.45\linewidth,height=2.7in,keepaspectratio]{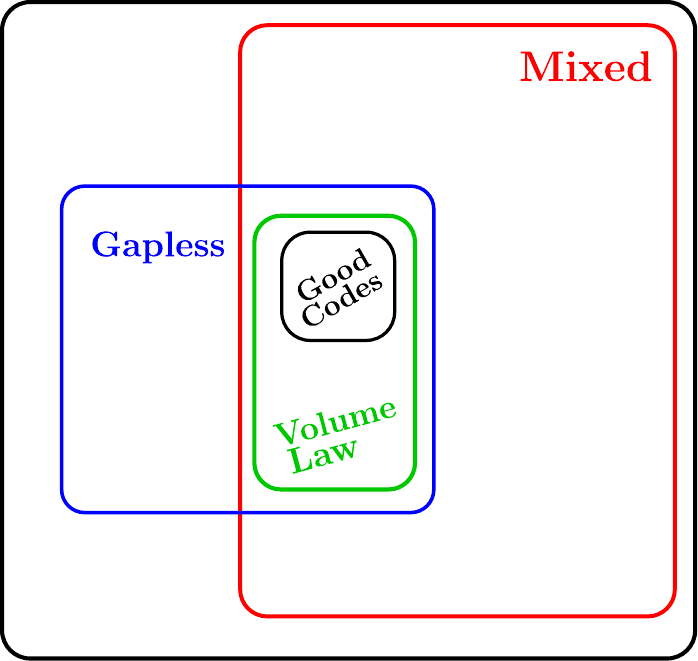}\label{fig:hierarchy}}
\caption{(a) Properties of the plateau group in gapless $(b,a,d)=(4,4,0)$ circuits on $m=j 2^n$ unit cells with PBCs. (i) Time to reach plateau, $t^*$, as a function of $n$. (ii) Entropy density of the plateau, $s^*$, as a function of $j$. (iii) Mutual information $\max_x I(x,m)$ of the plateau group. (iv) Code length $d_1$ of the plateau group. If the plateau group is pure (which occurs for a subset of $j=1$ circuits), this is indicated as $d_1=0$. (b) Zoology of STTI hybrid circuits, classified based on their behavior for $m=2^n$ with PBCs.}
\end{figure}

~\autoref{fig:hierarchy} depicts the classification of $j=1$ STTI hybrid circuits. Within the group of gapless circuits, we can further distinguish between those for which the plateau group (possibly pure) has area-law mutual information vs. volume-law mutual information. By "volume law", we mean that the function
\begin{equation}
    I(x, m) = \langle I(A:\overline{A}) \rangle_{|A| = x},
\end{equation}
i.e. the mutual information between contiguous region $A$ and its  complement $\overline{A}$, averaged over all contiguous regions of the same length $x$, has a maximum that scales linearly with $m$. Volume-law mutual information implies that the set of stabilizer generators, compressed to their shortest length, contains $O(m)$ long generators. The volume-law gapless $(4,4,0)$ circuits appear as spikes in $\max_x[I(x,m)]$ at $m=2^n$ in the lower left panel of~\autoref{fig:gapless}.

Within the intersection of gapless, volume-law, mixed circuits is a subset of circuits whose purification dynamics generates a "good code": with a code rate of $s^*$, the plateau group encodes $s^* L$ logical operators into $L$ physical qubits, with a contiguous code distance (length of the shortest logical operator)~\cite{Bravyi2009} that scales linearly in $L$. These codes are attractive for quantum error correction since they can have finite thresholds for local noise, although it should be emphasized that since many of the stabilizer generators are long, this limits their fault tolerance potential compared to good low-density-parity-check (LDPC) codes~\cite{Panteleev2021,Campbell2017,Breuckmann2021,Rakovszky2023,Rakovszky2024}.

From a more physics-oriented perspective, note that an extensive code length means the codespace is not Wannierizable, which in this context means that it cannot be written as a product Hilbert space of logical qubits~\cite{Ashcroft1976}. We leave open the question of Wannierizability for the plateau group for generic system sizes, such as odd $m$ where the gapless purification stage stops after $O(1)$ steps.

\subsection{Gapless, pure, product}\label{sect:gpp}
As another example of the sensitivity to boundary conditions in hybrid circuits, consider the class of circuits with $(b,a,d)=(2,4,0)$ which purify gaplessly to a pure product state with PBCs when $j=1$, but have a volume-law plateau group with OBCs, with $s_{eq}=1/2$.

This leads to a situation that cannot be explained by a membrane with a local free energy. In the protocol depicted in~\autoref{fig:page-sketch}, we first run the circuit with the vertical cut at $x=0$ and OBCs, so each half becomes well scrambled (and belongs to the plateau group of the OBC circuit). Suppose we take $m$ to be a power of 2. After terminating the cut at $t=0$, we can run two slightly different circuits. With open boundary conditions for $t>0$, the entanglement after the cut fills in with an infinite butterfly velocity to the right (\autoref{fig:gpp-obc}). The implied line tension is in fact identical to that of the dual hybrid circuit (left panel of \autoref{fig:dual-hybrid}), although we do not have an explicit multi-unitary construction for this circuit. In contrast, if we instead run the circuit with \textit{periodic} boundary conditions for $t>0$, a stark difference in the entanglement appears even at $O(1)$ times (\autoref{fig:gpp-pbc}): rather than increasing linearly in time as is the case with OBCs, $S(x,t)$ for small $x$ saturates at a constant under PBCs. Then at $O(m)$ times, entanglement starts to be destroyed in the PBC circuit (right panel), until the product steady state is reached. This indicates a sort of "interference" phenomenon between the two domain walls involved in $S(x,t)$ when we impose periodic boundary conditions.

\begin{figure}[t]
\subfloat[]{
\label{fig:gpp-obc}
\includegraphics[width=0.49\linewidth]{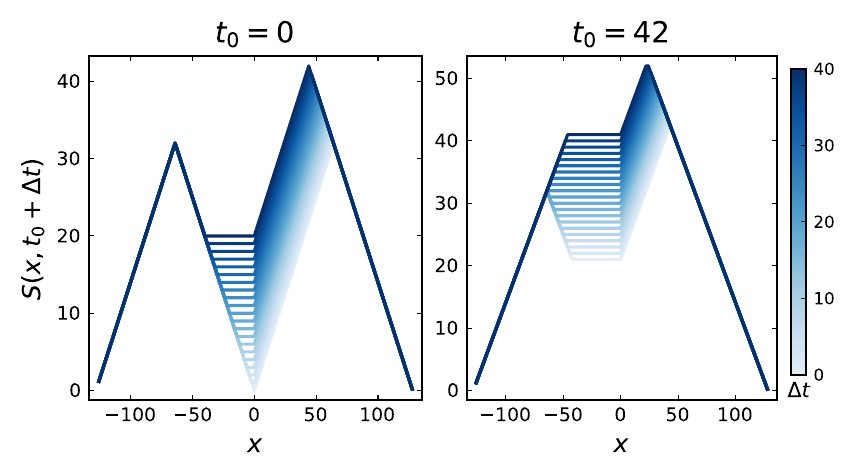}}
\subfloat[]{
\label{fig:gpp-pbc}
\includegraphics[width=0.49\linewidth]{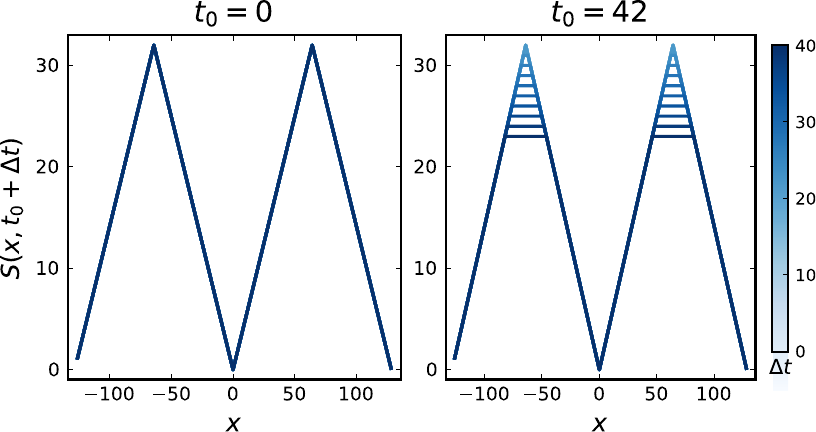}}
\caption{Subsystem entropy on a chain of $L=256$ qubits after terminating cut, for a circuit that (a) has a volume-law steady state with OBCs, but (b) is "gapless, purifying, product" with PBCs when the number of unit cells is a power of 2. Entropy is plotted immediately after measurements, every two layers, for even $x$. \label{fig:gpp}}
\end{figure}

While a local free energy cannot explain this phenomenon, we emphasize that the early-time sensitivity to boundary conditions in the above example is in fact a consequence of the "long time" correlations that accumulate during the preparation of the initial state ($t<0$). If we replace the $t<0$ time evolution with a random Clifford circuit, then the interference effect is absent. Namely, we recover
\begin{equation}
S_{pbc}(x,t) = S_{obc}(x,t) + S_{obc}(0,t).
\end{equation}
\subsection{Recurrences}
Another phenomenon not explained by the local free energy is the recurrence time $\tau$ of pure states within the plateau group. In unitary STTI Clifford circuits, also known as Clifford quantum cellular automata (CQCA)~\cite{Schlingemann2008,Gutschow2010solo,Gutschow2010long}, the recurrence time is linear in $m$ if $m=2^n$, but for the fractal class of CQCA is exponentially large for odd $m$~\cite{VonKeyserlingk2018,Sommers2023}. We observe an analogous fractal trend in the effective unitary dynamics within the plateau group of our hybrid Clifford circuits, although the effective unitary is not a local CQCA. For $j=1$, the $O(m)$ recurrence time means that if we start from an area-law entangled pure state within the plateau, then even if the circuit succeeds in generating volume-law entanglement governed by a finite-tension membrane, this entanglement is periodically destroyed on linear time scales. 

\subsection{Random hybrid circuits}
\gs{In the main text, we claimed that the entanglement membranes in a \textit{random} monitored Clifford circuit in the volume-law phase likewise are at zero temperature, owing to an emergent plateau.~\autoref{fig:lcf} shows how this property manifests in the purification dynamics starting from a fully mixed state at a low measurement rate $p=0.05<p_c \approx 0.16$~\cite{Gullans2020}.}

\gs{In an individual realization, the entropy density approaches a finite value $s^*$ according to the DPRM scaling~\cite{Li2023}:
\begin{equation}
    \frac{S(L,t)}{L} – s^* \sim L^{-2/3} F(t/L^{2/3}).
\end{equation}
At early times $t \ll L^{2/3}$, the deviation from the plateau entropy density is independent of $L$, so the scaling function $F(x)$ satisfies $F(x) \sim x^{-1}$ for $x\ll 1$.}\footnote{\gs{This scaling function differs from that used in~\cite{Li2023,Sang2023}. In those works, the entropy of a subsystem $|A|$ where $|A| \ll L$ is analyzed, and for $t\ll L^{2/3}$, it is the entropy difference $S(\rho_A,t)-s^* |A|$, not the difference in entropy \textit{density}, which is independent of $t$, resulting in a scaling function $\Phi(x) \sim x^{1/2}$ as $x\rightarrow 0$.}}

\gs{Differentiating the total entropy in this early time regime we obtain $dS/dt \sim -L/t^2$, so the time spacing between the discrete purification events scales as}
\begin{equation}
    \frac{t^2}{L} G(t/L^{2/3}).
\end{equation}

\gs{Thus, the purification events are sparse in time for $t\gg L^{1/2}$. At much later times, these events become exponentially rare in $L$, corresponding to an exponentially long lived plateau~\cite{Gullans2020}. Each purification event (marked with dashed vertical lines for $L=512$ in~\autoref{fig:lcf}) is associated with the membrane finding a new minimum energy path. Since the entropy density only changes at sparse intervals, whereas terms are continuously added to the partition function of the corresponding polymer (1D membrane) as it is given more room to fluctuate, we can say that the polymer is at zero temperature.}

\begin{figure}[hbtp]
\centering
\includegraphics[width=0.5\linewidth]{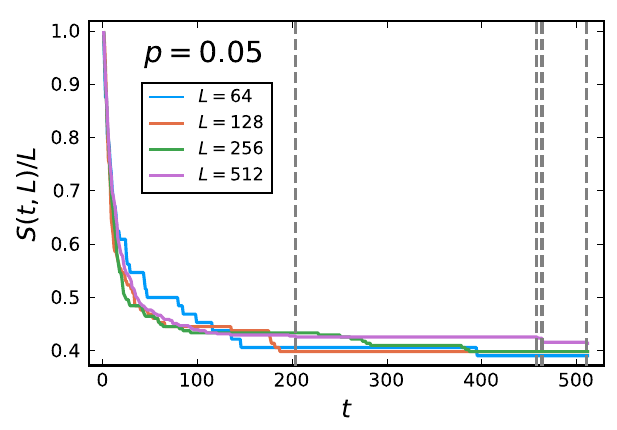}
\caption{\gs{Entropy density vs. time in individual runs of a random brickwork (1+1d) Clifford circuit with measurement rate $p$. Different colors are different system sizes. For the largest system size, vertical lines mark the purification events at $t>200$.} \label{fig:lcf}}
\end{figure}
\section{Ancilla probes}\label{app:probes}
In the main text, we controlled the entanglement membrane by pinning its ends by performing measurements or omitting gates. In this section, we demonstrate a method for probing the membrane in between the pinned ends, by introducing ancilla(s) at spacetime location $(x_p,t_p)$. We use the good-scrambling CNOT-NOTC circuit defined in~\autoref{eq:idx9} as our case study.

To be concrete, we consider a few different ancilla probes. In three of the probes, a single ancilla qubit is initialized in the state $\ket{+}$, and then "measures" the qubit at $x_p$ in the $X$, $Y$, or $Z$ basis by applying the gate CX, CY, or CZ. Another option ("Bell probe") is to bring in a pair of ancillas in a Bell state, and swap one of the ancillas with the qubit at $x_p$. 

 \subsection{Timelike probes}
To probe timelike domain walls, we start from a product state, insert the ancilla(s) at location $(x_p, t_p)$, and then at a later time $t$ evaluate the quantity
 \begin{equation}
\Delta S_{x_p,t_p}(x,t) = S_{x_p,t_p}(x+A,t) - S_{x_p,t_p}(x,t)
\end{equation}
where $S_{x_p,t_p}(x,t)$ denotes the entropy of the interval of the interval $[-L/2,x]$ not including the ancilla qubit(s) A in the subsystem, and $S_{x_p, t_p}(x+A,t)$ is the entropy where A is included in the subsystem. If \textit{all} minimum-energy interfaces connecting the point $(x,t)$ to the bottom boundary pass to the left of $(x_p,t_p)$, then $\Delta S < 0$ (since an extra cost is paid to exclude $A$). Conversely, if all minimum-energy domain walls pass to the right of $(x_p, t_p)$, then $\Delta S > 0$. 
 And if $\Delta S=0$, that indicates a degeneracy of domain walls passing left and right of the probe location.

Two comments are in order. First, $\Delta S$ satisfies the equation:
\begin{equation}\label{eq:DeltaS}
\Delta S_{x_p,t_p}(x,t) = I_{x_p,t_p}([-L/2,x]:A, t) - |A|
\end{equation}
where $I_{x_p,t_p}([-L/2,x]:A,t)$ is the mutual information between the subsystem $[-L/2,x]$ and the ancilla(s) inserted at $(x_p, t_p)$. Thus, for the single-ancilla probes, left/right degeneracy ($\Delta S=0$) corresponds precisely to the interval $[-L/2,x]$ containing only a classical bit. 

Second, note that the simple act of inserting the probe can change the landscape of domain walls nonperturbatively. For example, the "Bell probe" starts by destroying the entanglement between $x_p$ and the rest of the original system.
In that sense, these probes are in some sense "invasive", so are not fully cleanly interpreted. However, we find good agreement between all four probes, and all are consistent with the absence of thermal fluctuations. To wit,~\autoref{fig:probe} shows $\Delta S$ for the CX probe, for two choices of $(x,t)$. When $x=0$, the $\Delta S=0$ region (shown in red) has a width of $O(1)$, indicating an absence of thermal fluctuations which would tend to broaden this region. When $t\sqrt{3}>|x|>0$, the $\Delta S=0$ region forms a parallelogram with vertices at $(x_p,t_p)=(0,0),(0,t-|x|/\sqrt{3}), (x,|x|/\sqrt{3}), (x,t)$, consistent with the butterfly velocity $v_B=\sqrt{3}$ and line tension in~\autoref{eq:tri-tension}.  Because the line tension is linear for $v\in [0,v_B]$, all domain walls within the parallelogram, which are composed of piecewise segments of slope $\in [0, v_B]$ if $x\geq 0$ (or slope $\in [-v_B,0]$ if $x<0$) are degenerate, with free energy $(t\sqrt{3} + |x|)$. 

\begin{figure}[t]
\centering
\subfloat[]{
\includegraphics[width=0.5\linewidth]{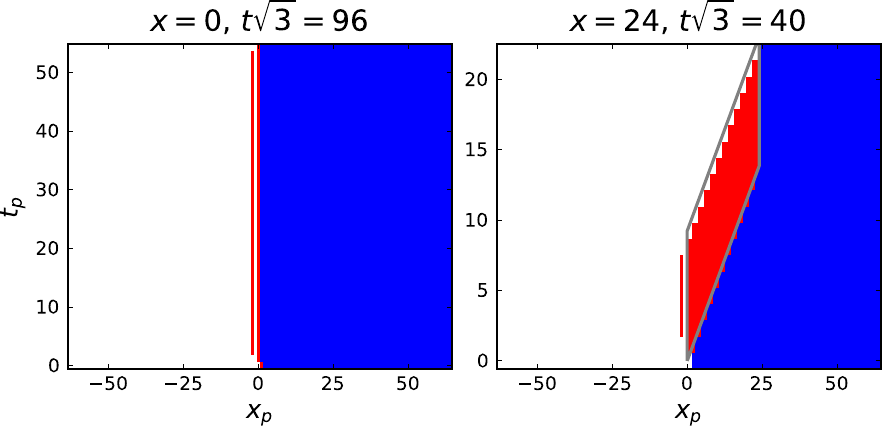}\label{fig:probe}}
\subfloat[]{
\includegraphics[width=0.3\linewidth]{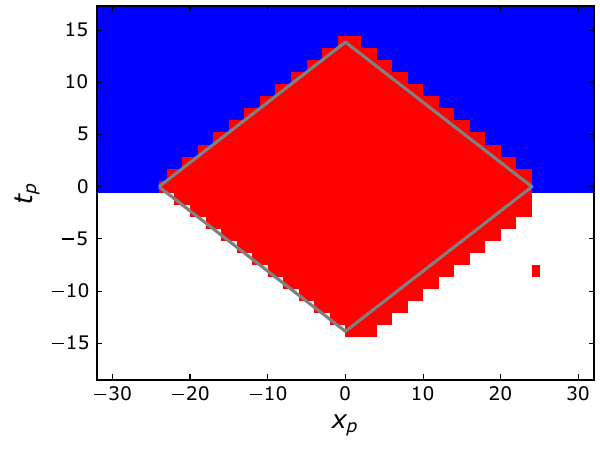}\label{fig:spacelike}}
\caption{$\Delta S_{x_p,t_p}$ for (a) timelike and (b) spacelike CX probes of the entanglement membrane in the representative CNOT-NOTC STTI circuit. Blue pixels are $\Delta S = 1$, red is $\Delta S = 0$, and white is $\Delta S = -1$. In (a), the system size is $L=128$ qubits, and each panel shows $\Delta S_{x_p, t_p}(x,t)$ for the $(x,t)$ indicated in the caption. In (b), $L=64$ and $L_{cut}=8$ qubits measured on each end at $t=0$. Gray lines indicate the predicted boundaries of the $\Delta S = 0$ region.}
\end{figure}

\subsection{Spatial probes}
We also probe domain walls forced to run in the spatial direction, by initializing in a fully mixed state. The dynamics is unitary except at two times: at $t=t_{cut}$, the rightmost and leftmost $L_{cut}$ qubits are measured, and at $t=t_{p}$, an ancilla probe is inserted at $x_p$. While unitarity forbids downward-pointing kinks in the membrane~\cite{Zhou2020}, the nonunitary operation at the insertion of the probe allows a "wrong turn" at $(x_p,t_p)$.

In~\autoref{fig:spacelike}, we plot the difference between the entropy of system + ancilla and system alone, at $t > \max(t_{cut}, t_p)$. 
If $\Delta S < 0$ that means all domain walls separating the system at the top from the reference at the bottom pass above $(x_p,t_p)$, and if $\Delta S>0$ that means all domain walls separating the system from the reference pass below $(x_p,t_p)$. The $\Delta S = 0$ region forms a rhombus with slopes $\pm v_B$. Spacelike membranes running inside this rhombus have free energy $L-2 L_{cut}$, since $\mathcal{E}(v) = |v|$ for $|v| \geq v_B$, whereas membranes running outside the rhombus pay an energy penalty for their timelike segments, since $\mathcal{E}(v) > |v|$ for $|v|<v_{B}$. This prediction for the extent of the $\Delta S=0$ region would hold in a generic unitary circuit, as well: $\mathcal{E}(v) \geq |v|$ is a general inequality, with equality at and above the butterfly velocity~\cite{Jonay2018}. But when the dynamics are not fully multi-unitary, we would expect a broadening of the boundary between $\Delta S=0$ and $\Delta S \neq 0$ due to tails outside the butterfly cone.  

\section{Unbinding transition}\label{app:unbinding}
As further evidence of the zero-temperature nature of the entanglement domain walls, we study the first-order unbinding transition from the edge due to quasiperiodic dissipation~\cite{Lovas2023}.

Starting from a pure product state defined on the interval $[0,L]$, we run the CNOT-NOTC STTI circuit defined in~\autoref{eq:idx9}, with open boundary conditions. The leftmost qubit is erased at an interval of $\tau$ if time is marked continuously. In the discrete time of the brickwork circuit, the erasure events occur quasiperiodically, at times:
\begin{equation}
t_{diss} = \lfloor (\tau (j + r))\rfloor,
\end{equation}
where $j=0,1,...$ and $r$ is a random number between 0 and 1, which varies from sample to sample. There is no randomness, except for the randomness in $r$ (and optionally, randomness in the choice of single-site stabilizers of the initial state). Since the maximum dissipation rate is 1 erasure per two layers, we define the dissipation rate $p=2/\tau$. 

\begin{figure}[t]
\subfloat[]{
    \centering
    \includegraphics[height=1.5in,keepaspectratio]{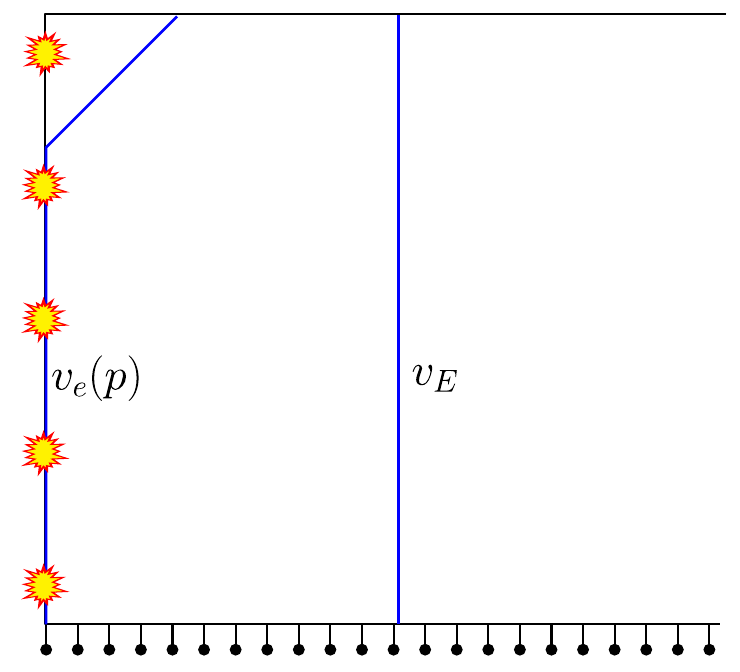}
    \label{fig:edge-dissipation}}
    \subfloat[]{
    \includegraphics[height=1.5in,keepaspectratio]{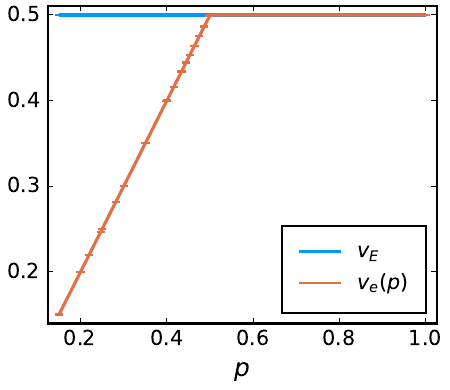}
    \label{fig:first-order}
    }
    \subfloat[]{
    \includegraphics[height=1.5in,keepaspectratio]{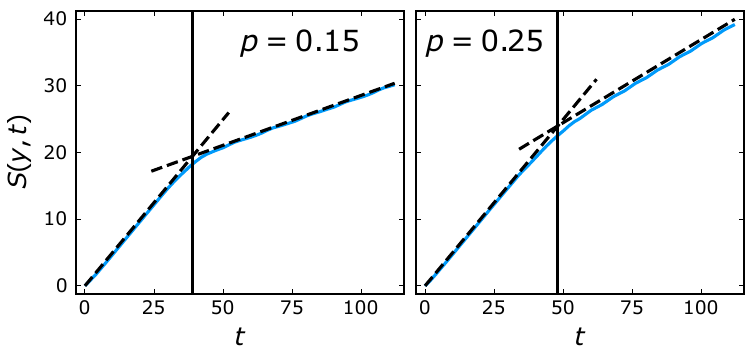}\label{fig:Sy}}
\caption{(a) Sketch showing the two paths the domain wall can take, and dissipation on the edge. (b) $v_E$ (blue) vs. $v_e(p)$ (orange), as a function of the quasiperiodic erasure rate, $p=2/\tau$. $v_e(p)$ is obtained from a linear fit to $S(0,t)$ for $t\in[L/8,15L/8]$. (c) $S(y,t)$ with $y=L/8$ and $L=128$ for two dissipation rates deep in the bound phase, $p=0.15$ (left) and $p=0.25$ (right). Black dashed lines show the expected free energy of the two domain wall configurations shown in (a) according to~\autoref{eq:Sy}, and the vertical black line marks the predicted transition time $t^*$.}
\end{figure}

The unbinding transition arises from the competition between two sources of entanglement, shown in~\autoref{fig:edge-dissipation}. The bulk of the unitary circuit generates entanglement at a rate $v_E$, the line tension of a membrane which travels straight down. On the other hand, each erasure injects at most 2 bits of entropy into the left edge. The effectiveness of this injection is quantified by the entropy of the (now mixed) state, which grows as $v_e(p) t$. The line tension of a membrane which binds to the left edge, $v_e(p)$, increases as a function of $p$: at sufficiently small $p$ the boundary is attractive ($v_e(p) < v_E$), whereas at large dissipation rates, it becomes energetically favorable for the membrane whose top endpoint is at $(0,t)$ to travel into the bulk (unbind).
The bulk and boundary entanglement rates, $v_E$ and $v_e(p)$, are plotted together in~\autoref{fig:first-order}. The transition is at $p_c=1/2$, with $v_e(p)$ growing and sticking at $v_E$ with a discontinuity in the derivative  $dv_e(p)/dp$  at the transition. 

In fact, $v_e(p)$ takes the piecewise linear form

\begin{equation}
v_e(p) = \begin{cases}
 p & p<1/2 \\
 1/2 & p\geq 1/2
 \end{cases}
 \end{equation}

Note that the growth of $S(0,t)=pt$ in the bound phase saturates the upper bound of 2 bits per erasure.

Within the bound phase, there is also a first-order transition between the two paths shown in~\autoref{fig:edge-dissipation}, as a function of time for a fixed distance from the edge. The minimal cut prediction for the entropy of the interval $[0,y]$, where $y<L/2$, is
\begin{equation}\label{eq:Sy}
S(y,t) = \min(v_E t, y + (t-y)p)
\end{equation}
so that, for a given $y$, we expect the membrane to switch between traveling straight down (slope $v_E=1/2$) and connecting to the edge (slope $p$). This transition occurs at $t^* = y(1-p)/(1/2-p)$ and is marked by a discontinuity of $\Delta = 1/2 - p$ in the slope at $t^*$. As $p$ increases, the transition shifts to higher $t$ and the discontinuity becomes less pronounced. For sufficiently small $p$ and $y/L$ this ansatz works fairly well, as shown in~\autoref{fig:Sy} for $y/L=1/8$ and $p=0.15,0.25$.

\begin{figure}[t]
\subfloat[]{
\includegraphics[width=0.35\linewidth]{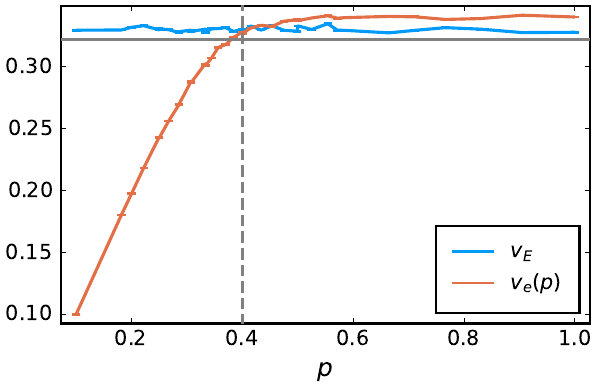}
\label{fig:random}}
\subfloat[]{
\centering
\includegraphics[width=0.35\linewidth]{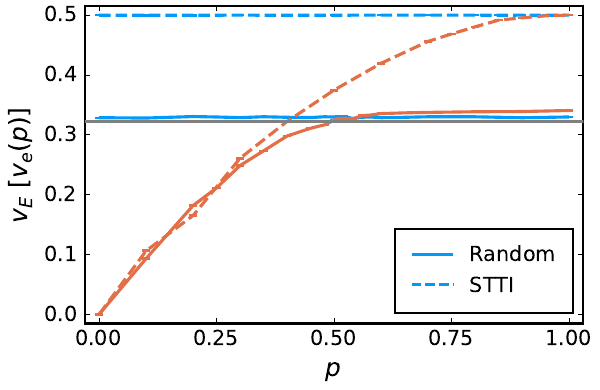}
\label{fig:edge-random}}
\caption{
(a) $v_E$ (from a linear fit to $S(L/2,t)$ for $t\in [L/8, 4L/3]$) vs. $v_e(p)$ (from a linear fit to $S(0,t)$ for $t\in [L/8, 3L/2]$) for quasiperiodic edge dissipation with random Clifford gates in the bulk. Gray solid line is the theoretical prediction for $v_E=\log_2(5/4)=0.3219...$ for a Haar-random circuit~\cite{Jonay2018}, and gray dashed line indicates the estimated critical point $p_c \approx 0.4$. (b) Fits to $v_E$ and $v_e(p)$ in a random Clifford circuit (solid, transition near $p=0.5$) and the CNOT-NOTC STTI circuit (dashed, transition near $p=1$), with random edge dissipation.}
\end{figure}

In contrast, when the gates are random, the phase transition is continuous. This clear qualitative difference from the STTI circuit is shown in~\autoref{fig:random}. Quantitatively, both $v_E$ and $v_e(p)$ are smaller than in the STTI circuit, as the random circuit scrambles more slowly (which also inhibits the injection of entropy into the left edge). As another point of contrast, deep in the unbound phase, when the gates are STTI, $S(y,t)-S(0,t)$ drops to zero in an $O(1)$ length scale, in a way that is independent of $t$ at large $t$. That is, consistent with the fact that the entanglement membranes are straight with no wandering, it suffices to travel $O(1)$ distance into the "bulk" to find a minimal energy path. In contrast, when the gates are random, the bulk wandering of the membrane is governed by DPRM physics~\cite{Nahum2017entanglement}, so that $S(y,t) - S(0,t) \sim t^{1/3} f(y/t^{2/3})$. 

A complete understanding of the critical behavior in the random case eludes us at this stage, but we do find that, as expected~\cite{Lovas2023}, edge randomness appears to be irrelevant to the critical behavior, although it shifts the transition to higher $p$ (random edge is less effective at generating entanglement). The transition in the presence of random edge dissipation is shown in~\autoref{fig:edge-random}, for both a random Clifford circuit and the CNOT-NOTC STTI circuit. At low $p$, $v_e^{rand}(p)$ in the random Clifford circuit tracks $v_e^{STTI}(p)$ in the STTI circuit quite closely. As $v_e(p)$ approaches the bulk entanglement velocity of the random circuit, $v_e^{rand}(p)$ peels away from $v_e^{STTI}(p)$. The transition is around $p_c\approx 0.5$ in the random circuit, vs. $p_c$ close to 1 in the STTI circuit, with the higher $p_c$ in the latter resulting from its higher bulk entanglement velocity. Notably, whereas the randomness of the boundary is irrelevant when the bulk is random, it is relevant at the bulk STTI transition, changing the transition from first-order to continuous.

\section{Haar-random initial states}\label{sect:haar}
The numerics in the main text and preceding sections of the Supplement have all involved initial stabilizer states. Here we provide evidence for the same entanglement growth for generic initial states. Concretely, we consider feeding the product of two Haar-random states, on the intervals $[-L/2,0]$ and $[0,L/2]$, into a circuit defined by the Floquet unitary $U_F$. We focus on the Floquet unitary of the good-scrambling CNOT-NOTC circuit (\autoref{eq:idx9}). The key insight is that although the circuit is nonrandom, the randomness in the initial state is sufficient to leverage tools from Haar calculus to compute ensemble averages such as the average and variance of the purity of intervals $A=[-L/2,x]$.

\subsection{Average purity}
For any pure state $\ket{\psi}$ defined on $n$ qubits, the $k$-replica average is:~\cite{Mele2023}:
\begin{equation}\label{eq:haar-ave}
\mathbbm{E}_{U \sim \mathrm{Haar}} [U^{\otimes k} (\ket{\psi} \bra{\psi}) ^{\otimes k} {U^\dag}^{\otimes k}] = \frac{(d-1)!}{(k+d-1)!} \sum_{\pi \in \mathcal{S}_k} V_d(\pi)
\end{equation}
where $d=2^n$ is the total Hilbert space dimension, and $V_d(\pi)$ is the permutation matrix corresponding to the permutation $\pi$ on $k$ elements.

The average purity of the interval $A$ involves a two-replica average, which at time $t=0$ is: 
\begin{align}\label{eq:purity}
\overline{\Tr[\rho_A(t=0)^2]} = \mathbbm{E}[\Tr[(S_A \otimes \mathbbm{1}_{\overline{A}}) \rho(0)^{\otimes{2}}]] = \frac{1}{2^L(2^{L/2}+1)^2} \Tr\left[\left(S_A \otimes \mathrm{1}_{\overline{A}}\right) ((\mathbbm{1} + S)_1 \otimes (\mathbbm{1} + S)_2)\right]
\end{align}
where $S$ is the SWAP operator acting between the two replicas, and the subscripts $1,2$ refer to the left and right halves of the physical Hilbert space.


After $t$ layers of the Floquet operator $U_F$, the purity evolves to:
\begin{align}\label{eq:purity-t}
\overline{\Tr[\rho_A(t)^2]} = \frac{1}{2^L(2^{L/2}+1)^2} \Tr\left[\left(S_A \otimes \mathrm{1}_{\overline{A}}\right) \left(U_F(t)^{\otimes{2}} ((\mathbbm{1} + S)_1 \otimes (\mathbbm{1} + S)_2) U_F^\dag(t)^{\otimes{2}}\right)\right].
\end{align}

We can split this into four parts: 
\begin{equation}
\overline{\Tr[\rho_A(t)^2]} = \frac{Z_1+Z_2+Z_3+Z_4}{2^L(2^{L/2}+1)^2}
\end{equation}
where
\begin{subequations}
\begin{align}
Z_1 &= \Tr\left[\left(S_A \otimes \mathbbm{1}_{\overline{A}}\right) \left(U_F(t)^{\otimes{2}} (\mathbbm{1}_1 \otimes \mathbbm{1}_2) U_F^\dag(t)^{\otimes{2}}\right)\right] =  \Tr\left[S_A \otimes \mathbbm{1}_{\overline{A}}\right]
=2^{3L/2-x} \\
Z_2 &= \Tr\left[\left(S_A \otimes \mathbbm{1}_{\overline{A}}\right) \left(U_F(t)^{\otimes{2}} (S_1 \otimes S_2) U_F^\dag(t)^{\otimes{2}}\right)\right] = \Tr\left[\mathbbm{1}_{A} \otimes S_{\overline{A}}\right] = 2^{3L/2+x} \\
Z_3(t)&=\mathrm{Tr}[(S_A \otimes \mathbbm{1}_{\overline{A}}) (U_F(t)^{\otimes 2}) (S_1 \otimes \mathbbm{1}_2) U_F^\dag(t)^{\otimes 2}] \equiv \Tr[U_3(t)] \\
Z_4(t) &=\mathrm{Tr}[(S_A \otimes \mathbbm{1}_{\overline{A}}) (U_F(t)^{\otimes 2}) (\mathbbm{1}_1 \otimes S_2) U_F^\dag(t)^{\otimes 2}] \equiv \Tr[U_4(t)]
\end{align}
\end{subequations}
Each of these terms can be interpreted as the partition function of a domain wall configuration, as visually depicted in~\autoref{fig:z}: the choice of permutation on the left and right halves sets the bottom boundary condition, the choice of subsystem $A$ sets the top boundary condition, and the form of $U_F(t)$ determines the cost of the domain wall in the bulk.
\begin{figure}[hbtp]
\subfloat[]{
\includegraphics[width=0.65\linewidth]{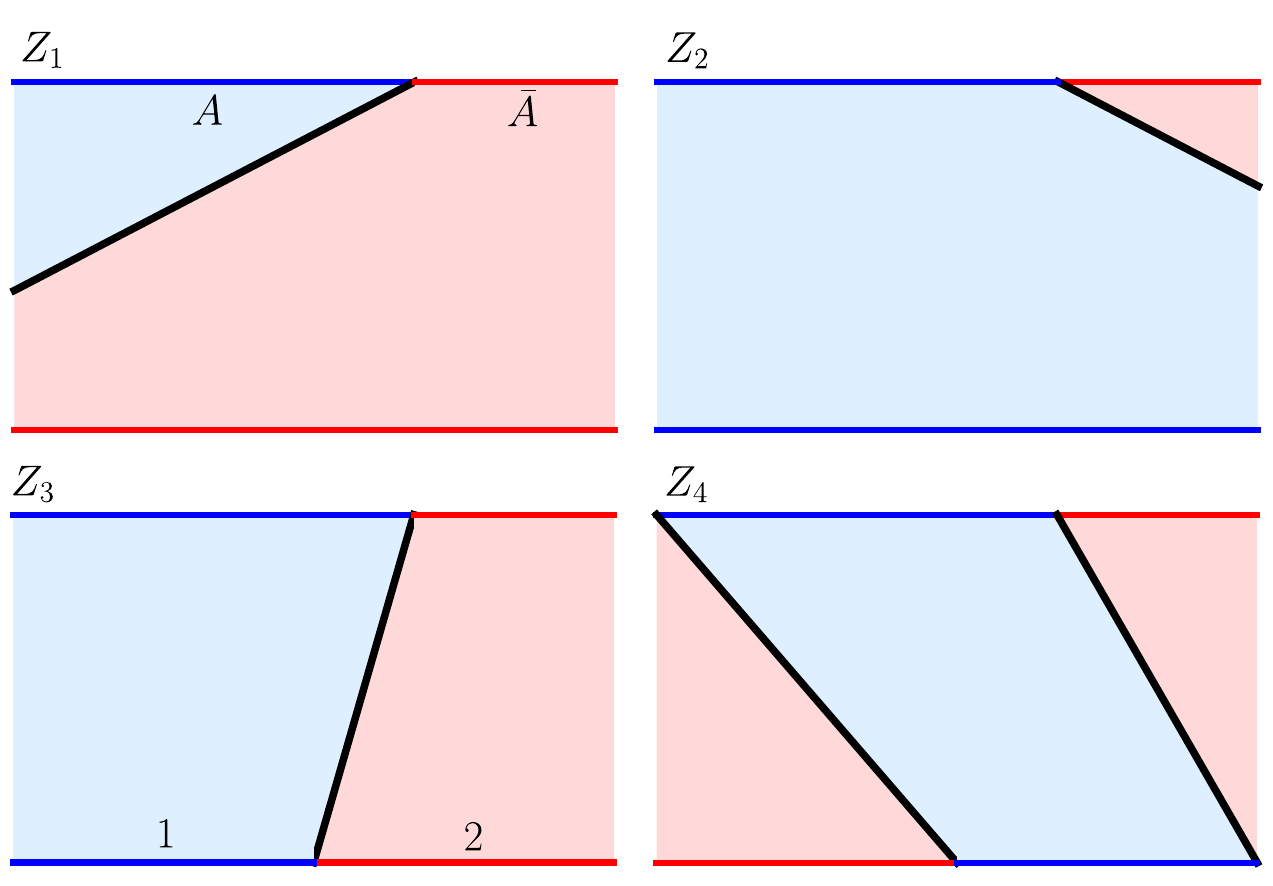}\label{fig:z}}
\subfloat[]{
\includegraphics[width=0.35\linewidth]{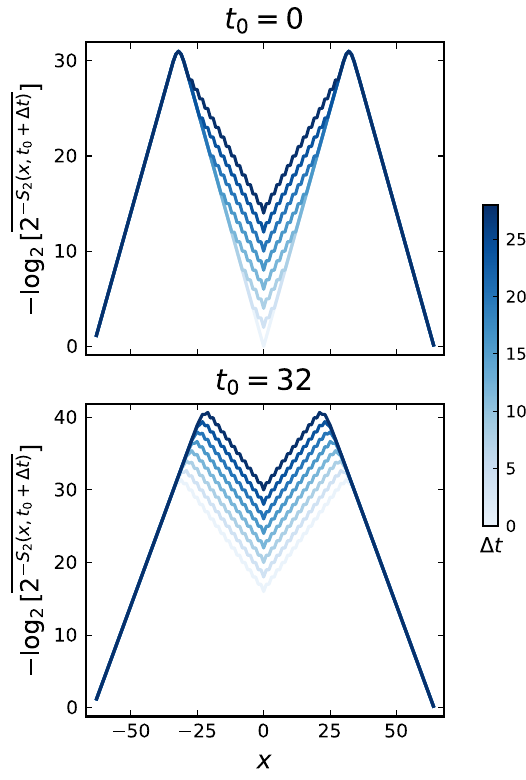}\label{fig:purity}}
\caption{(a) Sketch of the four terms that enter into the average purity: bright blue and red lines signify identity and swap, respectively, between the two replicas, on the indicated interval. (b) Log average purity in a CNOT-NOTC circuit with single-qubit gates given by~\autoref{eq:idx9}, for $L=128$, starting from the product of two Haar-random states at $t=0$. Times in increments of 4 layers are shown. }
\end{figure}
To compute $Z_3$ and $Z_4$, we leverage the fact that $U_3$ and $U_4$ are Clifford unitaries on the doubled Hilbert space. For a Clifford unitary $U$, let $G(U)$ denote a generating set of Pauli strings $P$ that are invariant up to a sign under conjugation by $U$: $U^\dag P U = \pm P$. Then the quantity $\mathrm{Tr}(U) \mathrm{Tr}(U^\dag)$ evaluates to~\cite{Su}:
\begin{equation}
|\Tr(U)|^2 = \begin{cases} 2^{|G(U)|} & U^\dag P U = +P \, \forall P \in G(U) \\
0 & \mathrm{otherwise}.
\end{cases}
\end{equation}
$G(U)$ can be obtained as the 1-eigenspace of the binary symplectic representation of $U$~\cite{Su}. When $U_F(t)$ is given by the good-scrambling CNOT-NOTC circuit defined in~\autoref{eq:idx9}, we find that the corresponding $|\Tr(U_3)|^2$ and $|\Tr(U_4)|^2$ are always nonzero.\footnote{As a simple example where this is not the case, consider $V=H$: $G(H)=Y$, but since $H Y H =-Y$, we recover the obvious result $\Tr(H) = \sqrt{|\Tr(H)|^2} = 0$.}

On small system sizes, we verify that both $Z_3$ and $Z_4$ are positive, i.e.
\begin{equation}\label{eq:positive}
\Tr[U_j(t)] = \Tr[U_j^\dag(t)] = + \sqrt{|\Tr[U_j(t)]|^2}
\end{equation}
Positing that~\autoref{eq:positive} holds for all $t$\footnote{This assumption is justified a posteriori by the agreement between~\autoref{fig:page} and~\autoref{fig:purity}, combined with the low variance.}, we find:
\begin{equation}
    \overline{\Tr(\rho_A(t)^2)} = \frac{1}{2^L(2^{L/2}+1)^2} (2^{3L/2-x} + 2^{3L/2+x} + 2^{|G(U_3(t))|/2} + 2^{|G(U_4(t))|/2} = \overline{2^{-S_2(x,t)}}.
\end{equation}
where $S_2(x,t)$ is the second-Renyi entropy. The log average purity, $-\log_2[\overline{\Tr(\rho_A(t)^2)}]$, is shown in~\autoref{fig:purity}.

Since the Clifford group is a unitary two-design~\cite{DiVincenzo2002},
\begin{equation}
   \mathbbm{E}_{\mathrm{Haar}}[\Tr(\rho_A(t)^2)] = \mathbbm{E}_{\mathrm{stab}}[\Tr(\rho_A(t)^2]
\end{equation}
where the subscripts indicate an average over Haar-random states or random stabilizer states, respectively,  on the two halves. Moreover, the close agreement between $-\log_2[\mathbbm{E}_{\mathrm{stab}}[2^{-S_2(x,t)}]]$ and the average entropy $\mathbbm{E}_{\mathrm{stab}} [S_2(x,t)]$ (obtained via sampling) implies a low variance of the purity across initial stabilizer states. Indeed,~\autoref{fig:purity} is nearly identical to the entropy growth observed for a specific initial state (generated by acting with $O(L)$ layers of the CNOT-NOTC circuit on a random product state) in~\autoref{fig:page}. 

\subsection{Variance of purity}
Next, to quantify the state-to-state fluctuations of the purity across Haar-random initial states on the two halves, we calculate the variance
\begin{equation}
\overline{\Tr[\rho_A(t)^2]^2} - \overline{\Tr[\rho_A(t)^2]}^2.
\end{equation}

The first term involves a four-replica average. Substituting $k=4$ into~\autoref{eq:haar-ave}, we obtain:
\begin{align}
\overline{\Tr[\rho_A(t)^2]^2} &= \frac{1}{(d(d+1)(d+2)(d+3))^2} \sum_{\pi_1,\pi_2 \in \mathcal{S}_4} \Tr \left[(S_A \otimes \mathbbm{1}_{\overline{A}})^{\otimes{2}}\left(U_F(t)^{\otimes{4}} \left(V_d(\pi_1)_1 \otimes V_d(\pi_2)_2 \right)(U_F(t)^\dag)^{\otimes{4}}\right)\right] \\
&\equiv \frac{1}{(d(d+1)(d+2)(d+3))^2} \sum_{\pi_1,\pi_2 \in \mathcal{S}_4} \Tr[U(\pi_1,\pi_2;t)]
\end{align}
where $d=2^{L/2}$ and $S^{\otimes 2}$ permutes replicas $(1,2)$ and $(3,4)$.
This yields the upper bound:
\begin{equation}
\mathrm{Var}(\Tr[\rho_A(t)]) \leq \frac{\sum_{\pi_1,\pi_2 \in \mathcal{S}_4} 2^{|G(U(\pi_1,\pi_2;t))|/2}}{(d(d+1)(d+2)(d+3))^2}  - \frac{(Z_1+Z_2+Z_3+Z_4)^2}{(d(d+1))^4}.
\end{equation}

By explicit computation, we find that the upper bound is either zero or decays exponentially with $L$. This indicates that \textit{typical} generic initial states result in the piecewise-linear entanglement generation shown in~\autoref{fig:purity}.\footnote{This small variance across Haar-random states could be anticipated from the small variance across sampled stabilizer initial states, because the fourth moment over the Clifford group matches the Haar average up to one additional term whose contribution we expect to be small~\cite{Zhu2016}.}

\end{document}